%% file: dipole-effective-dynamics.tex
\documentclass{birkjour}

\usepackage{amsmath}
\usepackage{amssymb}
\usepackage{amsfonts}
\usepackage{amsthm}
\usepackage{bm}
\usepackage{graphics}
\usepackage{color}
\usepackage[sans]{dsfont}
\usepackage[toc,page]{appendix}

\usepackage{enumerate}

\usepackage{hyperref}
\usepackage{float}
\usepackage{epsfig}
\usepackage{lscape}
\usepackage[english]{babel}
\usepackage[latin1]{inputenc}
\usepackage[T1]{fontenc}
\usepackage{wasysym}
\usepackage{tikz,pgflibraryshapes}
\usetikzlibrary{arrows}
\usetikzlibrary{snakes}
\usetikzlibrary {shapes}

\newtheorem{theorem}{Theorem}[section]
\newtheorem{lemma}[theorem]{Lemma}

\newtheorem{corollary}[theorem]{Corollary}

\newtheorem*{remark}{Remarks}
\numberwithin{equation}{section}

\numberwithin{equation}{section}

\renewcommand{\d}{\mathrm{d}}

\newcommand{\D}{\mathcal{D}}

\def\R{{\mathbb R}}

\renewcommand{\figurename}{Fig.}
\title[Isolated systems]{On the probabilistic nature of quantum mechanics and the notion of closed systems}
 
\begin{document}

\author[J. Faupin]{J{\'e}r{\'e}my Faupin}
\address[J. Faupin]{Institut Elie Cartan de Lorraine, 
Universit{\'e} de Lorraine,
57045 Metz Cedex 1, France}
\email{jeremy.faupin@univ-lorraine.fr}
\author[J. Fr\"ohlich]{J\"urg Fr\"ohlich}
\address[J. Fr{\"o}hlich]{Institut f{\"u}r Theoretische Physik, ETH H{\"o}nggerberg, CH-8093 Z{\"u}rich, Switzerland}
\email{juerg@phys.ethz.ch}
\author[B. Schubnel]{Baptiste Schubnel}
\address[B. Schubnel]{Institut Elie Cartan de Lorraine, 
Universit{\'e} de Lorraine,
57045 Metz Cedex 1, France}
\email{baptiste.schubnel@yahoo.fr}

\date{\today}

\small
\begin{abstract}
The notion of ``closed systems'' in Quantum Mechanics  is discussed. For this purpose, we study two models of a quantum-mechanical system $P$ spatially far separated from the ``rest of the universe'' $Q$. Under reasonable assumptions on the interaction between $P$ and $Q$, we show that the system $P$ behaves as a closed system  if the initial state of $P \vee Q$ belongs  to a large class of states, including ones exhibiting entanglement between $P$ and $Q$. We use our results to illustrate the non-deterministic nature of quantum mechanics. Studying a specific example, we show that assigning an initial  state and a unitary time evolution to a quantum system is generally not sufficient  to predict the results of a measurement with certainty.
\end{abstract}

\maketitle

\section{Introduction} 
\input{Introduction}

\section{Proof of Theorem \ref{thm:isolated} and Corollary \ref{cor:effective}} \label{proofs}
\subsection{Plan of the proof}

The estimates used in the proof of Theorem \ref{thm:isolated}  are insensitive to the presence of the potential $V_{\varepsilon}$ (see Corollary \ref{cor:effective}).  The bounds   derived in the next sections are valid for  both $H^{\varepsilon}$   and $H$. To keep consistent notations, we prove Theorem \ref{thm:isolated} with $H$ replaced by $H^{\varepsilon}$ and $H_{P \vee E}$ by $H_{P \vee E}^{\varepsilon}$.  
In Section \ref{section:loc}, we prove that, in the dressed atom state $\mathcal{J}(u)$, with $u$ as in Hypothesis \textbf{(B1)}, \emph{most} photons are localized in the ball of radius $d \gg R$ centered at the origin. Using the fact that the propagation velocity of photons is finite, we show, in addition, that after time $t$, for the dynamics generated by the atom-field Hamiltonian $H_{P \vee E}^\varepsilon$, most photons in the state $e^{-it H_{P \vee E}^\varepsilon} \mathcal{J}(u)$ remain localized in the ball of radius $d$ centered at the origin.

In Section \ref{sec:partition}, we introduce a partition of unity in Fock space (see \cite{DeGe99_01}) separating photons localized near the origin from photons localized near infinity. We rewrite the Hamiltonian $H^\varepsilon$ in the factorization of the Fock space determined by this partition of unity.

In Section \ref{section:main_proof}, we prove Theorem \ref{thm:isolated}, using  Cook's method, the partition of unity  of Section \ref{sec:partition} and the localization lemmas of Section \ref{section:loc}.

Proofs of some  technical lemmas are postponed to the appendix.

\subsection{Notations and conventions}
We remind the reader that for $a,b >0$,  we write
\begin{equation*}
a=\mathcal{O}(b)
\end{equation*}
if there is a constant $\mathrm{C}>0$ independent of  $t$, $d$ and $R$  such that $a \le \mathrm{C} b$.  For two vectors $\Psi_1, \Psi_2 \in \mathcal{H}$ and a constant $b>0$, we write $\Psi_1=\Psi_2 +\mathcal{O}(b)$ if $\|\Psi_1 -\Psi_2\|=\mathcal{O}(b)$.

Given two self-adjoint operators $A$ and $B$, the commutator $[ A , B ]$ is defined in the sense of quadratic forms on $\mathcal{D}(A) \cap \mathcal{D}( B )$ by
\begin{equation*}
\langle u , [ A , B ] v \rangle = \langle A u , B v \rangle - \langle B u , A v \rangle.
\end{equation*}
In our proof, we will encounter such a commutator that extends continuously to some suitable domain. The corresponding extension will be denoted by the same symbol, unless  confusion may arise. In the same spirit, we will often make use of ``Cook's method'' to compare two different dynamics. Suppose, for instance, that $B$ is $A$-bounded. Then we will write
\begin{equation*}
\big \| e^{ -i t B } u - e^{ - i t A } u \big \| \le \int_0^t \big \| ( A - B ) e^{ - i s A } u \big \| d s ,
\end{equation*}
for $u \in \D ( A )$. A proper justification of the previous inequality would be
\begin{align*}
\big \| e^{ -i t B } u - e^{ - i t A } u \big \| &= \sup_{ v \in \mathcal{D}( B ) , \| v \| = 1 } \big \vert  \big \langle v , u - e^{ i t B } e^{ - i t A } u \big \rangle  \big \vert \\
&= \sup_{ v \in \mathcal{D}( B ) , \| v \| = 1 }  \big \vert  \int_0^t  \langle v , e^{ i s B } ( A - B ) e^{ - i s A } u \big \rangle ds \big \vert \\
&= \Big \| \int_0^t  e^{ i s B } ( A - B ) e^{ - i s A } u \text{ }  ds \Big \| ,
\end{align*}
the last equality being a consequence of the fact that $A-B$ extends to an $A$-bounded operator. We will proceed similarly to estimate quantities like $\| B e^{ - i t A } u \| = \| e^{ i t A } B e^{ - i t A } u \|$ assuming for instance that $B$ is bounded and that the commutator $[ A , B ]$ extends to an $A$-bounded operator. Since  such arguments are standard, we will not repeat them in the rest of the paper.

\input{Localization}

\input{Partition_unity}
\subsection{Proof of Theorem \ref{thm:isolated}}\label{section:main_proof}

In this section, we prove our main result, Theorem \ref{thm:isolated}. 

 \vspace{2mm}

\begin{proof}[Proof of Theorem \ref{thm:isolated}]
To simplify the exposition, we assume that the  initial (non-normalized) state $\psi$ is  given by  $\psi = I \mathcal{J}( u ) \otimes \varphi$. The more general  initial condition  presented in the statement of Theorem \ref{thm:isolated} can be  directly  deduced from this special  case. We begin by  applying Lemma \ref{lm:b1}. Using unitarity of $e^{-itH^\varepsilon}$, this gives
\begin{align*}
& e^{ - i t H^\varepsilon } \psi = e^{ - i t H^\varepsilon } \psi_{ \mathrm{loc} } + \mathcal{O} \big ( (d/R)^{ \frac{-1+\gamma}{2} } \big ) , 
\end{align*}
for all $0 < \gamma \le 1$, with $\psi_{ \mathrm{loc} }$ defined in \eqref{eq:def-psi-loc}. Lemma \ref{lm:e1} then implies that
\begin{align}
& e^{ - i t H^\varepsilon } \psi = e^{ - i t \tilde{H}^\varepsilon } \psi_{ \mathrm{loc} } + \mathcal{O} \big ( (d/R)^{ \frac{-1+\gamma}{2} } \big ) + \mathcal{O} ( t d^{-1} ) +\mathcal{O} (td^{-\beta}) ,  \label{eq:k1}
\end{align}
where $\tilde{H}_\varepsilon$ is defined in \eqref{eq:def-tilde-Heps}.

Next, we show that
\begin{align}
\check{\Gamma}( \mathbf{j} ) e^{ - i t \tilde{H}^\varepsilon } \psi_{ \mathrm{loc} }  = \big ( e^{-i t H_{P \vee E}^\varepsilon } \mathcal{J}( u ) \big ) \otimes \big ( e^{- i t H_{Q \vee E} } \varphi \big ) + \mathcal{O} \big ( \langle t \rangle (d / R^2 )^{ - \frac{1}{2} } \big ) + \mathcal{O} ( t^2 d^{-\frac12} ) . \label{eq:i1}
\end{align}
We begin  by using the localization lemmas established above in order to rewrite the tensor product $ ( e^{-i t H_{P \vee E}^\varepsilon } \mathcal{J}( u ) ) \otimes ( e^{- i t H_{Q \vee E} } \varphi )$. Applying Lemma \ref{lm:f1} and Lemma \ref{lm:g1}, we obtain that
\begin{align}
 \big ( e^{-i t H_{P \vee E}^\varepsilon } \mathcal{J}( u ) \big ) \otimes \big ( e^{- i t H_{Q \vee E} } \varphi \big ) =& \big ( \Gamma( \mathds{1}_{ | \vec{y} | \le d } ) e^{-i t H_{P \vee E}^\varepsilon } \mathcal{J}( u ) \big ) \otimes \big ( \Gamma( \mathds{1}_{ | \vec{y} | \ge 2 d } ) e^{- i t H_{Q \vee E} } \varphi \big ) \notag \\
 & + \mathcal{O}\big ( t^{\frac54} d^{-\frac12} \big ) + \mathcal{O} \big ( \langle t \rangle^{\frac34} ( d / R )^{-\frac12} \big ) + \mathcal{O} ( t^2 d^{-1} ) . \label{eq:h1}
\end{align}
We observe that, since $\check{\Gamma}( \mathbf{j} )$ is a partial isometry, we have the relation
\begin{equation}
\check{\Gamma}( \mathbf{j} ) \check{\Gamma}( \mathbf{j} )^* = \mathds{1}_{Ê\mathrm{Ran}( \check{\Gamma}( \mathbf{j} ) ) } , \label{eq:f1}
\end{equation}
where $\mathds{1}_{Ê\mathrm{Ran}( \check{\Gamma}( \mathbf{j} ) ) }$ stands for the projection onto the (closed) subspace $\mathrm{Ran}( \check{\Gamma}( \mathbf{j} ) )$ of $\mathcal{H}_0 \otimes \mathcal{H}_\infty$. Moreover it is not difficult to verify that $\mathrm{Ran} ( \Gamma( \mathds{1}_{ | \vec{y} | \le d } ) \otimes \Gamma( \mathds{1}_{ | \vec{y} | \ge 2 d } ) ) \subset \mathrm{Ran}( \check{\Gamma}( \mathbf{j} ) )$. Thus we deduce from \eqref{eq:h1} and \eqref{eq:f1} that
\begin{align}
 \big ( & e^{-i t H_{P \vee E}^\varepsilon } \mathcal{J}( u ) \big ) \otimes \big ( e^{- i t H_{Q \vee E} } \varphi \big ) \\
 &= \check{\Gamma}( \mathbf{j} ) \check{\Gamma}( \mathbf{j} )^* \big ( \Gamma( \mathds{1}_{ | \vec{y} | \le d } ) e^{-i t H_{P \vee E}^\varepsilon } \mathcal{J}( u ) \big ) \otimes \big ( \Gamma( \mathds{1}_{ | \vec{y} | \ge 2 d } ) e^{- i t H_{Q \vee E} } \varphi \big ) \notag \\
 & + \mathcal{O}\big ( t^{\frac54} d^{-\frac12} \big ) + \mathcal{O} \big ( \langle t \rangle^{\frac34} ( d / R )^{-\frac12} \big ) + \mathcal{O} ( t^2 d^{-1} ). \label{eq:h2}
\end{align}
Next we rewrite
\begin{align}
& \big ( e^{-i t H_{P \vee E}^\varepsilon } \mathcal{J}( u ) \big ) \otimes \big ( e^{- i t H_{Q \vee E} } \varphi \big ) \notag \\
&= e^{-i t ( H_{P \vee E}^\varepsilon \otimes \mathds{1}_{ \mathcal{H}_\infty } + \mathds{1}_{ \mathcal{H}_0Ê} \otimes  H_{Q \vee E} ) } \big ( \mathcal{J}( u ) \otimes \varphi \big ) \notag \\
&= e^{-i t ( H_{P \vee E}^\varepsilon \otimes \mathds{1}_{ \mathcal{H}_\infty } + \mathds{1}_{ \mathcal{H}_0Ê} \otimes  H_{Q \vee E} ) }  \big ( \big ( \Gamma( \chi_{ | \vec{y} | \le d } ) \mathcal{J}( u ) \big ) \otimes \big ( \Gamma( \mathds{1}_{ | \vec{y} | \ge 2d } ) \varphi \big ) \big ) + \mathcal{O} ( (d/R)^{-1+\gamma} ) , \label{eq:h3}
\end{align}
for all $0 < \gamma \le 1$, the last equality being a consequence of Lemma \ref{lm:a3} and Hypothesis \textbf{(B2)}. To shorten notations, we set
\begin{equation*}
\tilde{\psi}_{ \mathrm{loc} } := \big ( \Gamma( \chi_{ | \vec{y} | \le d } ) \mathcal{J}( u ) \big ) \otimes \big ( \Gamma( \mathds{1}_{ | \vec{y} | \ge 2d } ) \varphi \big ),
\end{equation*}
so that $\psi_{ \mathrm{loc} } = \check{\Gamma}( \mathbf{j} )^* \tilde{\psi}_{ \mathrm{loc} }$ according to \eqref{eq:def-psi-loc}.  Remark that $\tilde{\psi}_{ \mathrm{loc} } \in \mathcal{D}( H_{P \vee E}) \otimes  \mathcal{D}(H_{Q \vee E})$ because $\chi_{\vert \cdot \vert \leq d}$ is a smooth function with compact support. Combining \eqref{eq:h2} and \eqref{eq:h3} (with $0 < \gamma \le 1/2$), we obtain that
\begin{align}
 \big ( e^{-i t H_{P \vee E}^\varepsilon } \mathcal{J}( u ) \big ) \otimes \big (  &e^{- i t H_{Q \vee E} } \varphi \big ) = \check{\Gamma}( \mathbf{j} ) \check{\Gamma}( \mathbf{j} )^* e^{-i t ( H_{P \vee E}^\varepsilon \otimes \mathds{1}_{ \mathcal{H}_\infty } + \mathds{1}_{ \mathcal{H}_0Ê} \otimes  H_{Q \vee E} ) } \tilde{\psi}_{ \mathrm{loc} } \notag \\
 & + \mathcal{O}\big ( t^{\frac54} d^{-\frac12} \big ) + \mathcal{O} \big ( \langle t \rangle^{\frac34} ( d / R )^{-\frac12} \big ) + \mathcal{O} ( t^2 d^{-1} ). \label{eq:h4}
\end{align}

Now we prove \eqref{eq:i1}. It follows from \eqref{eq:h4} and $\psi_{ \mathrm{loc} } = \check{\Gamma}( \mathbf{j} )^* \tilde{\psi}_{ \mathrm{loc} }$ that 
\begin{align}
& \big \|Ê\check{\Gamma}( \mathbf{j} ) e^{ - i t \tilde{H}^\varepsilon } \psi_{ \mathrm{loc} }  - \big ( e^{-i t H_{P \vee E}^\varepsilon } \mathcal{J}( u ) \big ) \otimes \big ( e^{- i t H_{Q \vee E} } \varphi \big ) \big \| \notag \\
& = \big \| \check{\Gamma}( \mathbf{j} ) e^{ - i t \tilde{H}^\varepsilon } \check{\Gamma}( \mathbf{j} )^* \tilde{\psi}_{ \mathrm{loc} } - \check{\Gamma}( \mathbf{j} ) \check{\Gamma}( \mathbf{j} )^* e^{-i t ( H_{P \vee E}^\varepsilon \otimes \mathds{1}_{ \mathcal{H}_\infty } + \mathds{1}_{ \mathcal{H}_0Ê} \otimes  H_{Q \vee E} ) } \tilde{\psi}_{ \mathrm{loc} }  \big \|   \notag \\
&\quad + \mathcal{O}\big ( t^{\frac54} d^{-\frac12} \big ) + \mathcal{O} \big ( \langle t \rangle^{\frac34} ( d / R )^{-\frac12} \big ) + \mathcal{O} ( t^2 d^{-1} ) \notag \\
& = \big \| e^{ - i t \check{\Gamma}( \mathbf{j} ) \tilde{H}^\varepsilon \check{\Gamma}( \mathbf{j} )^* } \tilde{\psi}_{ \mathrm{loc} } - \check{\Gamma}( \mathbf{j} ) \check{\Gamma}( \mathbf{j} )^* e^{-i t ( H_{P \vee E}^\varepsilon \otimes \mathds{1}_{ \mathcal{H}_\infty } + \mathds{1}_{ \mathcal{H}_0Ê} \otimes  H_{Q \vee E} ) } \tilde{\psi}_{ \mathrm{loc} }  \big \| \notag \\
&\quad + \mathcal{O}\big ( t^{\frac54} d^{-\frac12} \big ) + \mathcal{O} \big ( \langle t \rangle^{\frac34} ( d / R )^{-\frac12} \big ) + \mathcal{O} ( t^2 d^{-1} ) . \label{eq:k2}
\end{align}
 We  compute 
 \begin{align}
&  \big \| e^{ - i t \check{\Gamma}( \mathbf{j} ) \tilde{H}^\varepsilon \check{\Gamma}( \mathbf{j} )^* } \tilde{\psi}_{ \mathrm{loc} } - \check{\Gamma}( \mathbf{j} ) \check{\Gamma}( \mathbf{j} )^* e^{-i t ( H_{P \vee E}^\varepsilon \otimes \mathds{1}_{ \mathcal{H}_\infty } + \mathds{1}_{ \mathcal{H}_0Ê} \otimes  H_{Q \vee E} ) } \tilde{\psi}_{ \mathrm{loc} }  \big \| \notag \\
&  \le \int_0^t \big \| \check{\Gamma}( \mathbf{j} )\big (  \tilde{H}^\varepsilon \check{\Gamma}( \mathbf{j} )^* - \check{ \Gamma }( \mathbf{j} )^* ( H_{P \vee E}^\varepsilon \otimes \mathds{1}_{ \mathcal{H}_\infty } + \mathds{1}_{ \mathcal{H}_0Ê} \otimes  H_{Q \vee E} ) \big ) \notag \\
&\qquad \qquad  e^{-i s ( H_{P \vee E}^\varepsilon \otimes \mathds{1}_{ \mathcal{H}_\infty } + \mathds{1}_{ \mathcal{H}_0Ê} \otimes  H_{Q \vee E} ) } \tilde{\psi}_{ \mathrm{loc} }  \big \| ds , \label{eq:j4}
\end{align}
where we used that $\check{\Gamma}( \mathbf{j} ) \check{\Gamma}( \mathbf{j} )^*  \tilde{\psi}_{ \mathrm{loc} }=\tilde{\psi}_{ \mathrm{loc} }  $ and that $\check{\Gamma}( \mathbf{j} )^* \check{\Gamma}( \mathbf{j} ) = \mathds{1}_{\mathcal{H}}$. Applying Lemma \ref{lm:b2} with $\delta = 3/2$, we obtain that
\begin{equation}
\big \| \big (Ê\tilde H^\varepsilon \check{\Gamma}( \mathbf{j} )^* - \check{\Gamma}( \mathbf{j} )^* \big ( H_{P \vee E}^\varepsilon \otimes \mathds{1}_{ \mathcal{H}_\infty } + \mathds{1}_{ \mathcal{H}_0 } \otimes H_{Q \vee E} \big )\big ) \big ( N_0 + N_\infty + \langle \vec{x} \rangle \big )^{-1} \big \| = \mathcal{O}( d^{-\frac12} ) . \label{eq:j1}
\end{equation}
Adapting in a straightforward way the proof of Lemma \ref{lm:f1} and using that $\Gamma( \chi_{ | \vec{y} | \le d } ) \mathcal{J}( u ) \in \mathcal{D}(N) \cap \mathcal{D}(H_{P \vee E})$, we deduce that 
\begin{equation}
\big \| ( N + \langle x \rangle ) e^{-i s H_{P \vee E}^\varepsilon }  \Gamma( \chi_{ | \vec{y} | \le d } ) \mathcal{J}( u ) \big \| = \mathcal{O} ( s ) + \mathcal{O} ( R ).  \label{eq:j2}
\end{equation}
By Hypothesis \textbf{(B5)}, we also have that
\begin{equation}
\big \| N e^{ -i s H_{Q \vee E} } \varphi \big \| = \mathcal{O}( \langle s \rangle ) .  \label{eq:j3}
\end{equation}
Equations \eqref{eq:j1}, \eqref{eq:j2} and \eqref{eq:j3} yield that
\begin{align*}
& \big \| \big (  \tilde{H}^\varepsilon \check{\Gamma}( \mathbf{j} )^* - \check{ \Gamma }( \mathbf{j} )^* ( H_{P \vee E}^\varepsilon \otimes \mathds{1}_{ \mathcal{H}_\infty } + \mathds{1}_{ \mathcal{H}_0Ê} \otimes  H_{Q \vee E} ) \big )   e^{-i s ( H_{P \vee E}^\varepsilon \otimes \mathds{1}_{ \mathcal{H}_\infty } + \mathds{1}_{ \mathcal{H}_0Ê} \otimes  H_{Q \vee E} ) } \tilde{\psi}_{ \mathrm{loc} }  \big \|  \\
& =  \mathcal{O}( s d^{-\frac12} ) + \mathcal{ O }( R d^{ - \frac12 } ).
\end{align*}
Integrating this estimate and using again that $\check{\Gamma}( \mathbf{j} )$ is isometric, we deduce from \eqref{eq:j4} that
\begin{align}
&  \big \| e^{ - i t \check{\Gamma}( \mathbf{j} ) \tilde{H}^\varepsilon \check{\Gamma}( \mathbf{j} )^* } \tilde{\psi}_{ \mathrm{loc} } - \check{\Gamma}( \mathbf{j} ) \check{\Gamma}( \mathbf{j} )^* e^{-i t ( H_{P \vee E}^\varepsilon \otimes \mathds{1}_{ \mathcal{H}_\infty } + \mathds{1}_{ \mathcal{H}_0Ê} \otimes  H_{Q \vee E} ) } \tilde{\psi}_{ \mathrm{loc} }  \big \| \notag \\
& =  \mathcal{O}( t^2 d^{ -\frac12 } ) + \mathcal{ O }( t R d^{ - \frac12  } ).  \label{eq:k3}
\end{align}
Putting together \eqref{eq:k2} and \eqref{eq:k3}, we obtain \eqref{eq:i1}.

To conclude the proof, it suffices to combine \eqref{eq:k1} and \eqref{eq:i1}, which gives
\begin{align}
\check{\Gamma}( \mathbf{j} ) e^{ - i t H^\varepsilon } \psi_{ \mathrm{loc} }  =& \big ( e^{-i t H_{P \vee E}^\varepsilon } \mathcal{J}( u ) \big ) \otimes \big ( e^{- i t H_{Q \vee E} } \varphi \big ) \notag \\
&+ \mathcal{O} \big ( (d/R)^{ \frac{-1+\gamma}{2} } \big ) + \mathcal{O} \big ( \langle t \rangle (d / R^{2 }  )^{ - \frac{1}{2} } \big ) + \mathcal{O} ( t^2 d^{ - \frac12}  )+  \mathcal{O} ( t  d^{ -  \beta}  )  , \label{eq:l1}
\end{align}
for all $0 < \gamma \le 1$. Since $\check{\Gamma}( \mathbf{j} )$ is an isometry commuting with any bounded operator $O_{P}$ on $\mathcal{H}_{P}$, the last equation directly implies the statement of the theorem.
\end{proof}
 \vspace{4mm}

\appendix 
\section{Appendix for Section \ref{proofs}}
\label{AppA}
In this Appendix, we gather the proofs of several technical lemmas that were used in the proof of our main results. In Section \ref{B}, we prove Lemmas \ref{lm:a1}, \ref{lm:a2}, \ref{lm:a3} and \ref{lm:b1}. In Section \ref{sec:lmb2} we prove Lemma \ref{lm:b2}. Section \ref{section:eff-dyn} is devoted to the proof of Corollary \ref{cor:effective}. Finally, in Section \ref{app:comm}, we recall a few well-known relative bounds for operators on Fock space, and we estimate some commutators. 

\subsection{Proofs of the localization lemmas}\label{B}

We begin with recalling the expression of the fibre Hamiltonian $H( \vec{p} )$. The unitary map
\begin{equation*}
U : \mathcal{H}_{P} \otimes \mathcal{H}_E \rightarrow \int_{ \mathbb{R}^3 }^\oplus \mathbb{C}^{2 } \otimes \mathcal{H}_E \, d^3 p ,
\end{equation*}
such that
\begin{equation*}
U H_{P \vee E} U^{-1} = \int_{ \mathbb{R}^3 }^\oplus H( \vec{p} ) d^3 p ,
\end{equation*}
is the ``generalized Fourier transform'', defined by  
\begin{equation}
\label{cal_U}
(U \varphi) ( \vec{p}) = \frac{1}{(2 \pi)^{3/2}}\int_{\mathbb{R}^{3}}  e^{-i(\vec{p} -\vec{P}_E) \cdot \vec{x}} \varphi(\vec{x}) d^3 x 
\end{equation} 
for all $\varphi \in \mathcal{H}_P \otimes \mathcal{H}_E$ such that each $\varphi^{(n)}$ decays sufficiently rapidly at infinity. Introducing the notations 
\begin{equation}
b  (\underline{k}):= U e^{   i \vec{k} \cdot \vec{x}} a (\underline{k}) U^{-1} , \qquad \qquad b^{*} (\underline{k}):= U e^{ - i \vec{k} \cdot \vec{x}} a^{*} (\underline{k}) U^{-1} ,
\end{equation}
one verifies that
\begin{align}
H(\vec{p}) =& \frac12 \left(\vec{p}-\vec{P}_E \right)^2 +  \left(\begin{array}{cc} \omega_0 & 0 \\ 0& 0 \end{array} \right) \notag \\
& + i \lambda_0 \int_{ \underline{\mathbb{R} }^3 } \chi( \vec{k} ) \vert \vec{k} \vert^{\frac12} \vec{\varepsilon} (\underline{k}) \cdot \vec{\sigma}     \left( b (\underline{k})  - b^{*} (\underline{k}) \right) d \underline{k} + H_E , \label{Hp}
\end{align}
where $H_E =  \int_{ \underline{ \mathbb{R} }^3 }  \vert \vec{k}  \vert b ^{*}(\underline{k}) b (\underline{k}) d \underline{k}$. It follows from the Kato-Rellich theorem that the fiber Hamiltonians $H(\vec{p})$ are self-adjoint operators on $\mathcal{D}(H_E + \vec{P}_E^2)$.

We recall the main result of \cite{FaFrSc13_01}, which is used in our proofs.
 \vspace{2mm}
\begin{theorem}[Real analyticity of $\vec{p} \mapsto E( \vec{p} )$ \cite{FaFrSc13_01}] \label{thm:anal}
Let $0 < \nu < 1$. There exists  a constant $\lambda_c ( \nu )>0$ such that, for any coupling constant $\lambda_0 \in \mathbb{R}$ satisfying $\vert \lambda_0 \vert < \lambda_c( \nu )$,  the ground state energy $E( \vec{p} )$ of $H(\vec{p})$ is a non-degenerate eigenvalue of $H( \vec{p} )$, and the map $\vec{p} \mapsto E(\vec{p})$ and its associated eigenprojection $\vec{p} \mapsto \pi( \vec{p})$ are  real analytic on $B_\nu := \{ \vec{p} \in \mathbb{R}^3 , | \vec{p} | < \nu \}$. 
\end{theorem}
 \vspace{2mm}

We mention that, in \cite{FaFrSc13_01}, for simplicity, we have used a sharp ultraviolet cutoff $\mathds{1}_{ | \cdot | \le 1 }( \vec{k} )$ instead of the smooth ultraviolet cutoff $\chi( \vec{k} )$ used in the present paper. This modification, however, does not affect  the proof given in \cite{FaFrSc13_01}.

We also observe that the uncoupled Hamiltonian
\begin{align*}
H_0(\vec{p}) :=  \frac12 \left(\vec{p}-\vec{P}_E \right)^2 +  \left(\begin{array}{cc} \omega_0 & 0 \\ 0& 0 \end{array} \right) + H_E
\end{align*}
has a unique ground state (up to a phase) associated with the eigenvalue $\vec{p}\,^2 / 2$, given by 
\begin{equation*}
\psi_0 :=  \left(\begin{array}{cc} 0 \\ 1 \end{array} \right) \otimes \Omega ,
\end{equation*}
where $\Omega$ denotes the vacuum  Fock state. It is not difficult to verify that the ground state of $H( \vec{p} )$ overlaps with the ground state of $H_0( \vec{p} )$, in the sense that
\begin{equation*}
\big \| \pi( \vec{p} ) \psi_0 \big \| =  1 - \mathcal{O}( \vert \lambda_0 \vert ) ,
\end{equation*}
for small enough $\vert \lambda_0 \vert$. Therefore,
\begin{equation*}
\psi( \vec{p} ) := \pi( \vec{p} ) \psi_0
\end{equation*}
is a (non-normalized) ground state of $H( \vec{p} )$, and the map $\vec{p} \mapsto \psi( \vec{p} )$ is real analytic on $B_\nu$. In what follows, we keep the notation $\psi( \vec{p} )$ for $\pi( \vec{p} ) \psi_0$.

We now prove Lemma  \ref{lm:a1}.

\subsubsection{Proof of Lemma \ref{lm:a1}}
\begin{proof} 
By an interpolation argument, we see that it suffices to establish the statement of the lemma for $\mu = N \in \mathbb{N} \cup \{ 0 \}$. Recall from Hypothesis \textbf{(B1)} that $u( \vec{x} ) = R^{-3/2} v ( \vec{x} / R )$ where $v$ is a function independent of $R$ such that $ \mathrm{Supp}( v ) \subset \{ \vec{x} \in \mathbb{R}^3 , | \vec{x} | \le 1 \}$. Using \eqref{eq:defJ}, we compute
\begin{align*}
& \big \| ( 1 + | \vec{x} | )^{N} Ê\mathcal{J}( u ) ( \vec{x} ) \big \|_{ \mathcal{H}_{P} \otimes \mathcal{H}_E } \\
& = \frac{1}{(2 \pi)^{\frac32}} \Big \| ( 1 + | \vec{x} | )^{N } \int_{ \mathbb{R}^3 } \hat{u}(\vec{p})  e^{i\vec{x} \cdot \vec{p}} \chi_{\bar{B}_{\nu/2}} (\vec{p}) \psi ( {\vec{p}} ) \, d^3 p \, \Big \|_{ \mathcal{H}_{P} \otimes \mathcal{H}_E } \\
& = \frac{R^{\frac32}}{(2 \pi)^{\frac32}} \Big \| ( 1 + | \vec{x} | )^{N } \int_{ \mathbb{R}^3 } \hat{v}( R \vec{p} )  e^{i\vec{x} \cdot \vec{p}} \chi_{\bar{B}_{\nu/2}} (\vec{p}) \psi ( {\vec{p}} ) \, d^3 p \, \Big \|_{ \mathcal{H}_{P} \otimes \mathcal{H}_E } .
\end{align*}
Since $v$ is compactly supported, $\hat v \in \mathrm{C}^\infty( \mathbb{R}^3 )$, and hence, since in addition $\chi_{\bar{B}_{\nu/2}} \in \mathrm{C}_0^\infty( \mathbb{R}^3 )$ and since $\vec{p} \mapsto \psi( \vec{p} )$ is smooth on the support of $\chi_{\bar{B}_{\nu/2}}$ by Theorem \ref{thm:anal}, we deduce that
\begin{equation*}
\vec{p} \mapsto \hat{v}( R \vec{p}) \chi_{\bar{B}_{\nu/2}} (\vec{p}) \psi ( {\vec{p}} ) \in \mathrm{C}_0^\infty( \mathbb{R}^3 ; \mathbb{C}^2 \otimes \mathcal{H}_E ).
\end{equation*}
The result then follows from standard properties of the Fourier transform, using in particular that $| \partial^\alpha_{\vec{p}} \hat v ( R \vec{p} ) | \le \mathrm{C}_\alpha R^{| \alpha |}$ for all multi-index $\alpha$.
\end{proof}

To prove Lemma \ref{lm:a2}, we need to establish a preliminary lemma. Let $0 < \nu < 1$. It is shown in  \cite{FaFrSc13_01}, Section 5.1, that there is  a critical coupling constant $\lambda_{c}(\nu)>0$ such that  the map $\vec{k} \mapsto E(\vec{k})$ is analytic on the open set 
$$U [\vec{p}]:= \big \lbrace \vec{k} \in \mathbb{C}^3 \mid \vert \vec{p}- \vec{k}  \vert < \frac{1-\nu}{2} \big \rbrace$$
for all   $ \vert \vec{p}  \vert < \nu$ and all $ \lambda_0 \in \mathbb{R}$ with $ \vert \lambda_0 \vert < \lambda_c(\nu)$. Moreover, in the proof of Lemma 4.4. in  \cite{FaFrSc13_01}, we show that the choice made for  $\lambda_c(\nu)$ implies that 
\begin{equation}
\label{inene}
\Big \vert E(\vec{k})- \frac{\vert \vec{k} \vert^2}{2} \Big \vert < \left(\frac{1-\nu}{6} \right)^2 
\end{equation}
for all $\vert \vec{p} \vert < \nu$ and for all $\vec{k} \in U[\vec{p}]$, if $\vert \lambda_0 \vert < \lambda_c(\nu)$.
 \vspace{2mm}
\begin{lemma}
\label{app:GS} 
 Let $0<\nu<1$.  We assume that $ \vert  \lambda_0  \vert < \lambda_c( \nu )$. Then
  \begin{align}\label{eq:GSE}
E( \vec{p} - \vec{k} ) - E( \vec{p} ) + | \vec{k} | \ge \frac{1-\nu}{2} \big | \vec{k}  \big|
\end{align}
for all  $\vec{p} \in B_{\nu}$ and for all $\vec{k} \in B_{(1-\nu)/6}$.
\end{lemma}
 \vspace{2mm}
\begin{proof}
\noindent 
Let $  \vec{p}  \in  B_{\nu}= \{\vec{p} \in \mathbb{R}^3 \mid \vert \vec{p} \vert < \nu \}$. We set
$$ \tilde{E}(\vec{k}):=E(\vec{k})-\frac{\vert  \vec{k} \vert^2}{2}.$$
Since $\tilde{E}$  is analytic on $U [\vec{p}]$, we have that
\begin{equation*}
\vert \tilde{E}( \vec{p} - \vec{k} ) - \tilde{E}( \vec{p} ) \vert  \leq \underset{\vec{l} \in  U[\vec{p}] }{ \sup }   \vert \vec{\nabla} \tilde{E}(\vec{l})  \vert  \text{ } \vert \vec{k} \vert,
\end{equation*}
for all complex vectors $\vec{k}$ with  $\vert \vec{k} \vert < (1-\nu)/2$.  Let  now $\vec{\xi}=(\xi_1,\xi_2,\xi_3) \in \mathbb{C}^3$ with $\vert   \vec{p}-\vec{\xi} \vert <(1-\nu)/6$.   Using Cauchy formula for holomorphic functions of several complex variables and Eq.\eqref{inene},  we get that 
\begin{equation}
\Big  \vert (\partial_{z_1} \tilde{E})(\vec{\xi})  \Big  \vert \leq  \frac{1}{2 \pi} \Big \vert  \int_{\partial  D_{\frac{1-\nu}{6}} (\xi_1) } \frac{\tilde{E}(z,\xi_2,\xi_3)}{( z- \xi_1 )^2} dz \Big \vert \leq \frac{1- \nu}{6}, 
\end{equation}
where $D_{(1-\nu)/6}(\xi_1)$ is the complex open disk of radius $(1-\nu)/6$ centered at $\xi_1 \in \mathbb{C}$ and  $  \partial_{z_1} \tilde{E}$ denotes the partial derivative of $\tilde{E}$ with respect to the first component $z_1$. 
 Similar bounds hold for the partial derivatives with respect to $z_2$ and $z_3$, which implies that 
\begin{equation} \label{A17}
\vert \tilde{E}( \vec{p} - \vec{k} ) - \tilde{E}( \vec{p} ) \vert  \leq  \frac{1- \nu}{2}  \vert \vec{k} \vert
\end{equation}
for all $\vec{k} \in \mathbb{C}^3$ with  $\vert \vec{k} \vert <(1-\nu)/6$. The right side of \eqref{A17} is independent of $\vec{p}$ for all $ \vert \vec{p} \vert< \nu$.  Therefore,
\begin{equation}
 E( \vec{p} - \vec{k} ) - E( \vec{p}) + \vert \vec{k} \vert = \frac{\vert \vec{k} \vert^2}{2} - \vec{k} \cdot \vec{p}  +  \tilde{E}( \vec{p} - \vec{k} ) - \tilde{E}( \vec{p}) + \vert \vec{k} \vert \geq \frac{1-\nu}{2}  \vert \vec{k} \vert
\end{equation}
for all  $\vec{p} \in B_{\nu}$ and for all $\vec{k} \in B_{(1-\nu)/6}$.
\end{proof}
\vspace{2mm}

We now proceed to the proof of  Lemma \ref{lm:a2}. The proof follows \cite[Lemma 1.5]{juerg}. It relies on Lemma \ref{app:GS} and the pull-through formula
\begin{equation}\label{eq:pull-through}
b(\underline{k}) g(H_E ,\vec{P}_E)= g(H_E + \vert \vec{k} \vert, \vec{P_E}+ \vec{k}) b(\underline{k}),
\end{equation}
for any measurable function $g : \mathbb{R}^4 \rightarrow \mathbb{R}$. We do not present all the details.

 \vspace{2mm}
\subsubsection{Proof of  Lemma \ref{lm:a2}}
\begin{proof} 
Let $\vec{p} \in \mathbb{R}^3$, $| \vec{p} | < \nu$. Using Lemma \ref{lm:a2}, \eqref{eq:pull-through}, and adapting the proof of \cite[Lemma 1.5]{juerg} in a straightforward way, we deduce that there exists a constant $D( \vec{p} ) > 0$ such that, for any $n \in \mathbb{N}$,
\begin{equation*}
\Big \| \prod_{i=1}^n b( \underline{k}_i ) \psi( \vec{p} ) \Big \| \le D( \vec{p} )^n \vert \lambda_0 \vert^n \Big (  \prod_{i=1}^n \chi( \vec{k}_i ) | \vec{k}_i |^{-\frac12}  \Big ) \| \psi( \vec{p} ) \|.
\end{equation*}
The projection of $\psi( \vec{p} )$ onto the $n$-photons sector in Fock space is given by
\begin{equation*}
\psi^{(n)}( \underline{k}_1 , \dots , \underline{k}_n ) ( \vec{p} ) = \frac{1}{ \sqrt{ n! } } \langle \Omega, \prod_{i=1}^n b( \underline{k}_i ) \psi( \vec{p} ) \rangle ,
\end{equation*}
from which we obtain that
\begin{equation*}
e^{ \delta n } \big \| \psi^{(n)}( \underline{k}_1 , \dots , \underline{k}_n ) ( \vec{p} ) \big \| \le \frac{ e^{ \delta n } \big ( D( \vec{p} )   \vert \lambda_0 \vert  \big )^n }{ \sqrt{ n ! } } \frac{ \chi( k_1 ) \cdots \chi( k_n ) }{ | k_1 |^{\frac12} \cdots | k_n |^{ \frac12 } } \| \psi( \vec{p} ) \|.
\end{equation*}
Taking the square and integrating over $\underline{\mathbb{R}}^{3n}$, a direct computation then gives
\begin{equation*}
e^{ 2 \delta n } \int_{ \underline{ \mathbb{R} }^{3n} } \big \| \psi^{(n)}( \underline{k}_1 , \dots , \underline{k}_n ) ( \vec{p} ) \big \|^2 d \underline{k}_1 \cdots d \underline{k}_n  \le \frac{ e^{2 \delta n } ( 4 \pi )^n ( D( \vec{p} )  \vert \lambda_0 \vert  )^{2n} }{ n ! } \| \psi( \vec{p} ) \|^2,
\end{equation*}
and therefore
\begin{equation*}
\big \| e^{ \delta N } \psi( \vec{p} ) \big \| \le e^{ 2 \pi e^{ 2 \delta } ( D( \vec{p} )  \vert \lambda_0 \vert )^2 } \| \psi( \vec{p} ) \|.
\end{equation*}
This shows that $\psi( \vec{p} ) \in \mathcal{D} ( e^{ \delta N } )$. Moreover, one can verify that the constant $D( \vec{p} )$ can be chosen to be uniformly bounded on $\bar{B}_{\nu/2}:=\{ \vec{p} \in \mathbb{R}^3 , | \vec{p} | \le \nu/2 \}$, and hence
\begin{equation*}
\sup_{ \vec{p} \, \in \bar{B}_{ \nu/2 } } \| e^{ \delta N } \psi( \vec{p} ) \| \le \mathrm{C}_\delta .
\end{equation*}
This concludes the assertion of the proof of the lemma. 
\end{proof}

Next we prove Lemma \ref{lm:a3}.
 \vspace{2mm}
\subsubsection{Proof of Lemma \ref{lm:a3}}
\begin{proof}
Since $\Gamma( \mathds{1}_{ | \vec{y} | \le d } )$ is a projection, we can write
\begin{align}
\big \| \big ( \mathds{1} - \Gamma( \mathds{1}_{ | \vec{y} | \le d } ) \big ) \mathcal{J}( u ) \big \|^2 &= \big \langle \mathcal{J}( u ) , \big ( \mathds{1} - \Gamma( \mathds{1}_{ | \vec{y} | \le d } ) \big ) \mathcal{J}( u ) \big \rangle \notag \\
& \le \big \langle \mathcal{J}( u ) , \mathrm{d} \Gamma( \mathds{1}_{ | \vec{y} | \ge d } ) \mathcal{J}( u ) \big \rangle \notag \\
& \le d^{-2 + 2 \gamma} \big \langle \mathcal{J}( u ) , \mathrm{d} \Gamma( |\vec{y}|^{2- 2 \gamma} ) \mathcal{J}( u ) \big \rangle. \label{eq:a10}
\end{align}
It remains to show that $\mathcal{J}( u ) \in \mathcal{D} ( \mathrm{d} \Gamma( |\vec{y}|^{2- 2 \gamma} )^{1/2} )$ and that
\begin{equation*}
\big \langle \mathcal{J}( u ) , \mathrm{d} \Gamma( |\vec{y}|^{2- 2 \gamma} ) \mathcal{J}( u ) \big \rangle \le \mathrm{C}_\gamma R^{2-2\gamma}.
\end{equation*}

\noindent Using \eqref{eq:defJ} and the fact that 
\begin{equation*}
e^{i\vec{x} \cdot\vec{P}_E} \mathrm{d} \Gamma( |\vec{y}|^{2- 2 \gamma} ) e^{ - i \vec{x} \cdot \vec{P}_E} = \mathrm{d} \Gamma( |\vec{y} + \vec{x} |^{2- 2 \gamma} ) ,
\end{equation*}
we obtain
\begin{align*}
& \big \langle \mathcal{J}( u ) , \mathrm{d} \Gamma( |\vec{y}|^{2- 2 \gamma} ) \mathcal{J}( u ) \big \rangle = \frac{1}{(2 \pi)^{3}}  \Big \| \mathrm{d} \Gamma( |\vec{y} + \vec{x} |^{2- 2 \gamma} )^{\frac12} \int_{ \mathbb{R}^3 }  \hat{u}(\vec{p}) e^{i\vec{x} \cdot \vec{p}} \chi_{\bar{B}_{\nu/2}} (\vec{p}) \psi ( {\vec{p} }  )  d^3 p \, \Big \|^2 .
\end{align*}
The inequality $| \vec{y} + \vec{x} |^{2- 2 \gamma} \le \mathrm{c}_\gamma ( | \vec{y} |^{2 - 2\gamma} + | \vec{x} |^{2 - 2\gamma} )$ then gives
\begin{align}
 \big \langle \mathcal{J}( u ) , \mathrm{d} \Gamma( |\vec{y}|^{2- 2 \gamma} ) \mathcal{J}( u ) \big \rangle & \le \mathrm{C}_\gamma \Big ( \int_{ \mathbb{R}^3 } | \hat{u}(\vec{p}) | \chi_{\bar{B}_{\nu/2}} (\vec{p}) \big \| \mathrm{d} \Gamma( |\vec{y}|^{2- 2 \gamma} )^{\frac12} \psi ( {\vec{p} } ) \big \|  d^3 p \Big )^2 \notag \\
& \quad + \mathrm{C}_\gamma \Big \| \mathrm{d} \Gamma( | \vec{x} |^{2- 2 \gamma} )^{\frac12} \int_{ \mathbb{R}^3 }  \hat{u}(\vec{p}) e^{i\vec{x} \cdot \vec{p}} \chi_{\bar{B}_{\nu/2}} (\vec{p}) \psi ( {\vec{p} }  )  \text{ }d^3 p \, \Big \|^2 .  \label{eq:a5}
\end{align}
The two terms appearing on the right  side of the previous inequality are estimated separately. For the second one, we use that $\mathrm{d} \Gamma( | \vec{x} |^{2- 2 \gamma} ) = | \vec{x} |^{2- 2 \gamma} N$ and estimate, thanks to the Cauchy-Schwarz inequality,
\begin{align*}
& \Big \| \mathrm{d} \Gamma( | \vec{x} |^{2- 2 \gamma} )^{\frac12} \int_{ \mathbb{R}^3 }  \hat{u}(\vec{p}) e^{i\vec{x} \cdot \vec{p}} \chi_{\bar{B}_{\nu/2}} (\vec{p}) \psi ( {\vec{p} }  )  \text{ } d^3 p \, \Big \|^2 \\
& \le \int_{ \mathbb{R}^3 } | \hat{u}(\vec{p}) | \chi_{\bar{B}_{\nu/2}} (\vec{p}) \big \| N \psi ( {\vec{p} } ) \big \|  d^3 p \times \Big \| | \vec{x} |^{2- 2 \gamma} \int_{ \mathbb{R}^3 }  \hat{u}(\vec{p}) e^{i\vec{x} \cdot \vec{p}} \chi_{\bar{B}_{\nu/2}} (\vec{p}) \psi ( {\vec{p} }  )  \text{ } d^3 p \, \Big \| \\
& = \int_{ \mathbb{R}^3 } | \hat{u}(\vec{p}) | \chi_{\bar{B}_{\nu/2}} (\vec{p}) \big \| N \psi ( {\vec{p} } ) \big \|  d^3 p \times \big \| | \vec{x} |^{2- 2 \gamma} \mathcal{J}(u)(x) \big \|.
\end{align*}
By Lemma \ref{lm:a2}, $\sup_{ \vec{p} \, \in \bar{B}_{ \nu/2 } } \| N \psi( \vec{p} ) \| < \infty$ and, by Lemma \ref{lm:a1}, $\| | \vec{x} |^{2- 2 \gamma} \mathcal{J}(u)(x) \| \le \mathrm{C}_\gamma R^{2-2\gamma}$. This proves that
\begin{align}
& \Big \| \mathrm{d} \Gamma( | \vec{x} |^{2- 2 \gamma} )^{\frac12} \int_{ \mathbb{R}^3 }  \hat{u}(\vec{p}) e^{i\vec{x} \cdot \vec{p}} \chi_{\bar{B}_{\nu/2}} (\vec{p}) \psi ( {\vec{p} }  )  \text{ }  d^3 p \, \Big \| \le \mathrm{C}_\gamma R^{2-2\gamma} . \label{eq:a9}
\end{align}

It remains to estimate the first term on the right side of \eqref{eq:a5}. Using the pull-through formula \eqref{eq:pull-through} and setting $f( \underline{k}):=- i \lambda_0 \chi( \vec{k} ) | \vec{k} |^{\frac12}\vec{ \varepsilon }( \underline{k} ) \cdot \vec{ \sigma }$, we obtain that
\begin{equation}
\label{pupu}
b( \underline{k}) (H(\vec{p}) + \xi)=(H(\vec{p}-\vec{k}) + \vert \vec{k} \vert +  \xi) b(\underline{k}) + f( \underline{k}).
\end{equation}
A direct application of \eqref{pupu} then yields
\begin{align}
\label{psij}
b( \underline{k} ) \psi( \vec{p} ) = - \big ( H( \vec{p} - \vec{k} ) - E( \vec{ p } ) + | \vec{k} | \big )^{-1}    f( \underline{k})  \psi( \vec{p} ) ,
\end{align}
and by Lemma \ref{app:GS}, this implies
\begin{equation}\label{eq:a6}
\big \| b( \underline{k} ) \psi( \vec{p} ) \big \| \le \mathrm{C} \, \mathds{1}_{ | \cdot | \le 1 }( \vec{k} ) | \vec{k} |^{-\frac12} \| \psi( \vec{p} ) \| ,
\end{equation}
where $\mathrm{C}$ is a positive constant.

Differentiating \eqref{psij} with respect to $\vec{k}$, we obtain  that
\begin{equation}\label{eq:a7}
\big \| | i \vec{\nabla}_k | b( \underline{k} ) \psi( \vec{p} ) \big \| = \big \| \vec{\nabla}_k b( \underline{k} ) \psi( \vec{p} ) \big \| \le \mathrm{C} \, \mathds{1}_{ | \cdot | \le 1 }( \vec{k} ) | \vec{k} |^{-\frac32} \| \psi( \vec{p} ) \| .
\end{equation}
Equations \eqref{eq:a6} and \eqref{eq:a7}, together with an interpolation argument, yield
\begin{equation}\label{eq:a8}
\big \| |i \vec{\nabla}_k |^{1-\gamma} b( \underline{k} ) \psi( \vec{p} ) \big \| \le \mathrm{C} \, \mathds{1}_{ | \cdot | \le 1 }( \vec{k} ) | \vec{k} |^{-\frac32 + \gamma} \| \psi( \vec{p} ) \| ,
\end{equation}
for all $0 \le \gamma \le 1$. This shows that $\underline{k} \mapsto | i \vec{\nabla}_k |^{1 - \gamma} b( \underline{k} ) \psi( \vec{p} ) \in L^2 ( \underline{ \mathbb{R} }^3 ; \mathbb{C}^2 \otimes \mathcal{H}_E )$ for $0 < \gamma \le 1$ and, more precisely, that
\begin{equation*}
\sup_{ \vec{p} \in \bar{B}_{ \nu / 2 } } \big \| \mathrm{d} \Gamma( |\vec{y}|^{2- 2 \gamma} )^{\frac12} \psi ( {\vec{p} } ) \big \| < \infty.
\end{equation*}
Therefore we have proven that
\begin{equation*}
\int_{ \mathbb{R}^3 } | \hat{u}(\vec{p}) | \chi_{\bar{B}_{\nu/2}} (\vec{p}) \big \| \mathrm{d} \Gamma( |\vec{y}|^{2- 2 \gamma} )^{\frac12} \psi ( {\vec{p} } ) \big \|  d^3 p < \infty .
\end{equation*}
Together with \eqref{eq:a10}, \eqref{eq:a5} and \eqref{eq:a9}, this concludes the proof.
\end{proof}
\vspace{2mm}

We conclude this paragraph with the proof of Lemma \ref{lm:b1}.
\vspace{2mm}
\subsubsection{Proof of Lemma \ref{lm:b1}}
\begin{proof}
We begin with justifying that $\mathcal{J}( u ) \otimes \varphi \in \mathcal{D}( I )$. It is not difficult to verify that, for any $0 \le a,  b \le 1$ such that $a^2 + b^2 \le 1$, the operator $I \Gamma( a \mathds{1} ) \otimes \Gamma( b \mathds{1} )$ extends to a bounded operator satisfying
\begin{equation*}
\big \| I \Gamma( a \mathds{1} ) \otimes \Gamma( b \mathds{1} ) \big \| \le 1.
\end{equation*}
Since $\varphi$ satisfies Hypothesis   \textbf{(B2)}, there is $\delta > 0$ such that $\varphi \in \mathcal{D}( e^{ \delta N } )$. Choosing $\delta' > 0$ such that $e^{- 2\delta'} + e^{ - 2 \delta} \le 1$, we deduce that
\begin{align*}
\big \|ÊI \mathcal{J}( u ) \otimes \varphi \big \| &= \big \| I \big ( \Gamma ( e^{ - \delta' } \mathds{1} ) \otimes \Gamma( e^{ - \delta } \mathds{1} ) \big ) ( e^{\delta' N } \mathcal{J}( u ) ) \otimes ( e^{\delta N} \varphi ) \big \| \\
& \le \big \|Êe^{\delta' N } \mathcal{J}( u ) \big \| \big \| e^{\delta N} \varphi \big \| < \infty ,
\end{align*}
the fact that $\big \|Êe^{\delta' N } \mathcal{J}( u ) \big \| < \infty$ being a consequence of Lemma \ref{lm:a2}. Hence $\mathcal{J}( u ) \otimes \varphi \in \mathcal{D}( I )$.
Now we prove that
\begin{equation*}
I \mathcal{J}( u ) \otimes \varphi = I \big ( \Gamma( \chi_{ | \vec{y} | \le d } ) \mathcal{J}( u ) \big ) \otimes \big ( \Gamma( \mathds{1}_{ | \vec{y} | \ge 3d } ) \varphi \big ) + \mathcal{O}( (d/R)^{ \frac{-1+\gamma}{2} } ),
\end{equation*}
for $0 < \gamma \le 1$.
Since $\Gamma( \mathds{1}_{ | \vec{y} | \ge 3d } ) \varphi = \varphi$ by Hypothesis \textbf{(B2)}, we have that
\begin{align*}
I \mathcal{J}( u ) \otimes \varphi &= I \big ( \Gamma( \chi_{ | \vec{y} | \le d } ) \mathcal{J}( u ) \big ) \otimes \big ( \Gamma( \mathds{1}_{ | \vec{y} | \ge 3d } ) \varphi \big ) + I \big ( ( \mathds{1} - \Gamma( \chi_{ | \vec{y} | \le d } ) ) \mathcal{J}( u ) \big ) \otimes \varphi .
\end{align*}
We observe that all the terms of the previous equations are well-defined, as follows from  the fact that $\mathcal{J}( u ) \otimes \varphi \in \mathcal{D}( I )$  and  the remark after the statement of Lemma \ref{lm:b1}. Therefore we have to prove that
\begin{equation}
I \big ( ( \mathds{1} - \Gamma( \chi_{ | \vec{y} | \le d } ) ) \mathcal{J}( u ) \big ) \otimes \varphi = \mathcal{O}( (d/R)^{ \frac{ -1+\gamma }{ 2 } } ). \label{eq:b2}
\end{equation}
Proceeding as above, we estimate
\begin{align*}
\big \|ÊI & \big ( ( \mathds{1} - \Gamma( \chi_{ | \vec{y} | \le d } ) )  \mathcal{J}( u ) \big ) \otimes \varphi \big \| \\
&= \big \| I \big ( \Gamma ( e^{ - \delta' } \mathds{1} ) \otimes \Gamma( e^{ - \delta } \mathds{1} ) \big ) ( e^{\delta' N } ( \mathds{1} - \Gamma( \chi_{ | \vec{y} | \le d } ) ) \mathcal{J}( u ) ) \otimes ( e^{\delta N} \varphi ) \big \| \\
& \le \big \|Êe^{\delta' N } ( \mathds{1} - \Gamma( \chi_{ | \vec{y} | \le d } ) ) \mathcal{J}( u ) \big \| \big \| e^{\delta N} \varphi \big \| .
\end{align*}
Next, since $e^{\delta' N }$ commutes with $\Gamma( \chi_{ | \vec{y} | \le d } )$, we deduce that
\begin{align*}
\big \|Êe^{\delta' N } ( \mathds{1} - \Gamma( \chi_{ | \vec{y} | \le d } ) ) \mathcal{J}( u ) \big \|^2 & \le \big \|Êe^{2 \delta' N } \mathcal{J}( u ) \big \| \big \| ( \mathds{1} - \Gamma( \chi_{ | \vec{y} | \le d } ) )^2 \mathcal{J}( u ) \big \| \\
& \le \big \|Êe^{2 \delta' N } \mathcal{J}( u ) \big \| \big \| ( \mathds{1} - \Gamma( \chi_{ | \vec{y} | \le d } ) ) \mathcal{J}( u ) \big \| .
\end{align*}
It follows from Lemma \ref{lm:a2} that $\big \|Êe^{ 2 \delta' N } \mathcal{J}( u ) \big \| < \infty$, and by Lemma \ref{lm:a3}, we have that 
\begin{equation} 
\big \| ( \mathds{1} - \Gamma( \chi_{ | \vec{y} | \le d } ) ) \mathcal{J}( u ) \big \| = \mathcal{O}( (d/R)^{-1 + \gamma} ) , \label{eq:b7}
\end{equation}
for all $0 < \gamma \le 1$. The last three estimates prove \eqref{eq:b2}, which concludes the proof of the lemma.
\end{proof}
\vspace{2mm}

\subsection{Proof of Lemma \ref{lm:b2}}\label{sec:lmb2}

\begin{proof}[Proof of Lemma \ref{lm:b2}]
Since $\check{\Gamma}( \mathbf{j} )$ only acts on the photon Fock space, we obviously have that
\begin{equation}
 \big ( H_{P} + H_Q \big ) \check{\Gamma}( \mathbf{j} )^* = \check{\Gamma}( \mathbf{j} )^* \big ( H_{P} \otimes \mathds{1}_{ \mathcal{H}_\infty } + \mathds{1}_{ \mathcal{H}_0 } \otimes H_Q  \big ). \label{eq:e1}
\end{equation}
Moreover, it follows from Hypothesis \textbf{(B4)} that
\begin{equation}
 H_{Q,E} \check{\Gamma}( \mathbf{j} )^* = \check{\Gamma}( \mathbf{j} )^* \big ( \mathds{1}_{ \mathcal{H}_0 } \otimes H_{Q,E} \big ) .  \label{eq:e2}
\end{equation}

It remains to consider $ H_E \check{\Gamma}( \mathbf{j} )^*$ and $H_{P,E} \check{\Gamma}( \mathbf{j} )^*$. A direct computation (see e.g. \cite[Lemma 2.16]{DeGe99_01}) gives
\begin{align}
H_E \check{\Gamma}( \mathbf{j} )^* &= \check{\Gamma}( \mathbf{j} )^* \big ( H_E \otimes \mathds{1}_{ \mathcal{H}_\infty } + \mathds{1}_{ \mathcal{H}_0 } \otimes H_E \big ) - \d \Gamma( \mathbf{j}^* , \check{\mathrm{ad}}( |k| , \mathbf{j}^* ) ) \mathcal{U}^* , \label{eq:d1}
 \end{align}
where $\mathcal{U}$ is the unitary operator defined in \eqref{defU_1}--\eqref{defU_2} and, given $a,b : L^2( \underline{\mathbb{R}}^3 ) \oplus L^2( \underline{\mathbb{R}}^3 ) \to L^2( \underline{\mathbb{R}}^3 )$, the operator $\d \Gamma( a , b ) : \mathcal{F}_+( L^2( \underline{\mathbb{R}}^3 ) \oplus L^2( \underline{\mathbb{R}}^3 ) ) \to  \mathcal{H}_E $ is defined by its restriction to ${\otimes_s^n} (L^2( \underline{\mathbb{R}}^3 ) \oplus L^2( \underline{\mathbb{R}}^3 ) )$ as
\begin{align}
&\d \Gamma( a, b )\vert_\mathbb{C} := 0 , \label{eq:checkdGamma} \\
& \d \Gamma( a, b )\vert_{  {\otimes_s^n} L^2( \underline{\mathbb{R}}^3 ) \oplus L^2( \underline{\mathbb{R}}^3 ) } := \sum_{j=1}^{n} \underbrace{a \otimes \cdots \otimes a }_{j-1} \otimes b \otimes \underbrace{a \otimes \cdots \otimes a }_{n - j} . \label{eq:checkdGamma2}
\end{align}
The operators $\mathbf{j}^*$ and $\check{\mathrm{ad}}( |k| , \mathbf{j}^* ) : L^2( \underline{\mathbb{R}}^3 ) \oplus L^2( \underline{\mathbb{R}}^3 ) \to L^2( \underline{\mathbb{R}}^3 )$ in \eqref{eq:d1} are defined by
\begin{align*}
& \mathbf{j}^* ( h_0 , h_\infty ) = \mathbf{j}_0 h_0 + \mathbf{j}_\infty h_\infty , \\
& \check{\mathrm{ad}}( |k| , \mathbf{j}^* ) ( h_0 , h_\infty ) := [ |k| , \mathbf{j}_0 ] h_0 + [ |k| , \mathbf{j}_\infty ] h_\infty ,
\end{align*}
for all $(h_0 , h_\infty ) \in L^2( \underline{\mathbb{R}}^3 ) \oplus L^2( \underline{\mathbb{R}}^3 )$. By Lemma \ref{lm:B5} of Appendix \ref{app:comm}, we have that $\| [ |k| , \mathbf{j}_0 ] \| = \mathcal{O}( d^{-1} )$ and $\| [ |k| , \mathbf{j}_\infty ] \| = \mathcal{O}( d^{-1} )$. This yields
\begin{equation}
\big \| \d \Gamma( \mathbf{j}^* , \check{\mathrm{ad}}( |k| , \mathbf{j}^* ) ) \mathcal{U}^* ( N_0 + N_\infty + \mathds{1} )^{ - 1} \big \| = \mathcal{O}( d^{-1} ). \label{eq:d6}
\end{equation}
Equations \eqref{eq:d1} and \eqref{eq:d6} yield
\begin{align}
H_E \check{\Gamma}( \mathbf{j} )^* &= \check{\Gamma}( \mathbf{j} )^* \big ( H_E \otimes \mathds{1}_{ \mathcal{H}_\infty } + \mathds{1}_{ \mathcal{H}_0 } \otimes H_E \big ) + \mathrm{Rem}_1 , \label{eq:e3}
 \end{align}
 with
 \begin{equation}
\big \| \mathrm{Rem}_1 ( N_0 + N_\infty + \mathds{1} )^{ - 1 } \big \| = \mathcal{O}( d^{-1} ). \label{eq:e4}
\end{equation}

Now we treat the interaction Hamiltonian $H_{P,E}$. Using the notations \eqref{eq:def_hx1}--\eqref{eq:def_hx2}, we have that (see e.g. \cite[Lemma 2.15]{DeGe99_01}),
\begin{align}
H_{P,E} \check{\Gamma}( \mathbf{j} )^* = \Phi( h_x ) \check{\Gamma}( \mathbf{j} )^* &= \check{\Gamma}( \mathbf{j} )^* \big ( \Phi( \mathbf{j}_0 h_x ) \otimes \mathds{1}_{ \mathcal{H}_\infty } + \mathds{1}_{ \mathcal{H}_0 } \otimes \Phi( \mathbf{j}_\infty h_x ) \big ) . \label{eq:d2}
 \end{align}
Here it should be understood that the operators $\mathbf{j}_0$, $\mathbf{j}_\infty$ are applied to the $L^2( \underline{\mathbb{R}}^3 ; \mathcal{H}_{P} )$ functions $\underline{k} \mapsto h_x( \underline{k} )$ defined in \eqref{eq:def_hx2}. By Lemma \ref{lm:B4} of Appendix \ref{app:comm}, we have that
\begin{equation*}
\big \| \Phi( \mathbf{j}_\infty h_x ) \langle \vec{x} \rangle^{-2 + \delta } ( N + \mathds{1} )^{-\frac12} \big \| = \mathcal{O}( d^{-2+\delta} ) ,
\end{equation*}
for all $0 < \delta \le 2$ , and likewise that
\begin{equation*}
\big \| \big ( \Phi( h_x ) - \Phi( \mathbf{j}_0 h_x ) \big ) \langle \vec{x} \rangle^{-2 + \delta } ( N + \mathds{1} )^{-\frac12} \big \| = \mathcal{O}( d^{-2+\delta} ) .
\end{equation*}
Therefore we can conclude that
\begin{align}
 H_{P,E} \check{\Gamma}( \mathbf{j} )^* &= \check{\Gamma}( \mathbf{j} )^* \big ( H_{P,E} \otimes \mathds{1}_{ \mathcal{H}_\infty } \big )  + \mathrm{Rem}_2,  \label{eq:e5}
 \end{align}
 with 
\begin{equation}
\big \| \mathrm{Rem}_2 \big ( N_0 + N_\infty + \langle \vec{x} \rangle^{4 - 2 \delta } \big )^{-1} \big \| = \mathcal{O}( d^{-2+\delta} ) .  \label{eq:e6}
\end{equation} 
Equations \eqref{eq:e1}, \eqref{eq:e2}, \eqref{eq:e3}, \eqref{eq:e4}, \eqref{eq:e5} and \eqref{eq:e6} prove the statement of the lemma.
\end{proof}

\subsection{Proof of Corollary \ref{cor:effective}}  \label{section:eff-dyn}
\input{AppendixA}

\vspace{2mm}

\subsection{Relative bounds in Fock space and commutator estimates}\label{app:comm}
We begin this appendix with some useful estimates concerning creation and annihilation operators on Fock space and second quantized operators. 

We introduce the notation
\begin{equation}\label{eq:defh_0}
\mathfrak{h}_0 := \Big \{ h \in L^2( \underline{ \mathbb{R} }^3 ) , \| h \|_{ \mathfrak{h}_0 } := \int_{ \underline{ \mathbb{R} }^3 } ( 1 + |k|^{-1} ) | h ( \underline{k} ) |^2 d \underline k < \infty \Big \}.
\end{equation}
We recall the following standard result (see e.g. \cite[Lemma 17]{FrGrSc01_01}).
\begin{lemma}\label{lm:B1}
Let $f_i \in L^2( \underline{ \mathbb{R} }^3 )$ for $i=1,\dots,n$. Then
\begin{align*}
\big \| a^\#(f_1) \cdots a^\#(f_n) ( N+\mathds{1} )^{- \frac{n}{2} } \big\| \le C_n \| f_1 \|_{ L^2( \underline{ \mathbb{R} }^3 ) } \dots \| f_n \|_{ L^2( \underline{ \mathbb{R} }^3 ) } ,   
\end{align*}
where $a^\#$ stands for $a$ or $a^*$. If in addition $f_i \in \mathfrak{h}_0$ for $i=1,\dots,n$ (where $\mathfrak{h}_0$ is defined in \eqref{eq:defh_0}), then
\begin{align*}
\big \| a^\#(f_1) \cdots a^\#(f_n) ( H_E + \mathds{1} )^{- \frac{n}{2} } \big\| \le C_n  \| f_1 \|_{ \mathfrak{h}_0 } \dots \| f_n \|_{ \mathfrak{h}_0} .
\end{align*}
\end{lemma}
The following lemma was used (sometimes implicitly) several times in the main text. Its proof can be found in \cite[Section 3]{GeGeMo04_01}.
\begin{lemma}\label{lm:B2bis}
Let $\omega,\omega'$ be two self-adjoint operators on $L^2( \underline{ \mathbb{R} }^3 )$ with $\omega' \ge 0$, $\mathcal{D} (\omega') \subset \mathcal{D} (\omega)$ and $\| \omega \varphi \|_{ L^2( \underline{ \mathbb{R} }^3 ) } \le \| \omega'  \varphi \|_{ L^2( \underline{ \mathbb{R} }^3 ) }$ for all $\varphi \in \mathcal{D} (\omega')$. Then $\mathcal{D} ( \d\Gamma(\omega') ) \subset \mathcal{D}( \d\Gamma(\omega) )$ and $\| \d \Gamma(\omega) \Phi \| \le \| \d \Gamma(\omega') \Phi \|$ for all $\Phi \in \mathcal{D} ( \d\Gamma(\omega') )$.
\end{lemma}

Now we turn to a few localization estimates that were used in the main text. The next lemma is a particular case of \cite[Lemma 3.1]{BoFaSi12_01}. We do not  present the proof.
\begin{lemma}\label{lm:B3}
Let $F \in C^\infty( \mathbb{R}_+ ; [0,1] )$ be a smooth function such that $\mathrm{Supp}( F ) \subset [ 1 , \infty )$. Let $a \in [ 0 , 3/2 )$, $b \in \mathbb{R}$, $\chi \in C^{\infty}_{0} ( \R^{3} )$ and $h^{b}_{x} (\vec{k})$ be such that, for all $\alpha \in \mathbb{N}^3$, $| \partial_{\vec{k}}^\alpha h^{b}_{x} ( \vec{k} ) |  \leq C_{\alpha}  | \vec{k} |^{b-|\alpha|} \langle \vec{x} \rangle^{|\alpha|}$ with $ C_{\alpha} \geq 0$. Assume that $b > a - 3/2$. Then, for all $c \in [ 0 , b - a + 3/2 )$ and $d > 0$,
\begin{equation*}
\forall \vec{x} \in \mathbb{R}^{3}, \quad \big \| | \vec{k} |^{- a} F ( | i \vec{\nabla}_k | / d ) \chi (\vec{k}) h^{b}_{x} (\vec{k}) \big \|_{L^{2} ( \mathbb{R}^{3}_{k} )} \le C d^{- c} \langle \vec{x} \rangle^{a + c} .
\end{equation*}
\end{lemma}
Combining Lemmas \ref{lm:B1} and \ref{lm:B3}, we obtain the following estimates that have been used in the proof of Lemma \ref{lm:b2}. Recall that the operators $\mathbf{j_0}$, $\mathbf{j_\infty}$ are defined at the beginning of Section \ref{sec:partition} and that the coupling function $h_x$ was defined in \eqref{eq:def_hx2}.
\begin{lemma}\label{lm:B4}
For all $0 < \delta \le 2$, we have that
\begin{align*}
\big \| \Phi( \mathbf{j}_\infty h_x ) \langle \vec{x} \rangle^{-2 + \delta } ( N + \mathds{1} )^{-\frac12} \big \| = \mathcal{O}( d^{-2+\delta} ) , \\
\big \| \Phi( ( \mathds{1} - \mathbf{j}_0 ) h_x ) \langle \vec{x} \rangle^{-2 + \delta } ( N + \mathds{1} )^{-\frac12} \big \| = \mathcal{O}( d^{-2+\delta} ) .
\end{align*}
\end{lemma}
\begin{proof}
The proofs of the two stated estimates being the same, we only consider the first one. Applying Lemma \ref{lm:B1}, we obtain that, for all $\varphi \in \mathcal{H}_{P} \otimes \mathcal{H}_E \simeq L^2 ( \mathbb{R}^3 ; \mathbb{C}^2 \otimes \mathcal{H}_E )$,
\begin{align}
& \Big \| \Phi( \mathbf{j}_\infty h_x ) \langle \vec{x} \rangle^{-2 + \delta } ( N + \mathds{1} )^{-\frac12} \varphi \Big \|^2 \notag \\
&\le C \int_{\mathbb{R}^3}  \langle \vec{x} \rangle^{-4 + 2 \delta} \big\Vert \mathbf{j}_\infty h_x ( \underline{k} ) \big \|^2_{ L^2( \underline{ \mathbb{R} }^3 ) } \| \varphi(\vec{x}) \|^2_{ \mathbb{C}^2 \otimes \mathcal{H}_E } \, d ^3 x .
\end{align}
Lemma \ref{lm:B3} (applied with $a=0$ and $b=1/2$) then yields that
\begin{align*}
\Big \| \Phi( \mathbf{j}_\infty h_x ) \langle \vec{x} \rangle^{-2 + \delta } ( N + \mathds{1} )^{-\frac12} \varphi \Big \|^2 \le C d^{-4+2\delta}  \| \varphi \|^2 ,
\end{align*}
for all $0 < \delta \le 2$.
\end{proof}
Another related consequence of Lemmas \ref{lm:B1} and \ref{lm:B3} is given by the following commutator estimate used in the proof of  Lemma \ref{lm:f1}. Recall that $\chi_{ | \vec{y} | \le d } = j_0( 2 | \vec{y} | / d )$ and that $H_{P,E} = \Phi( h_x ) = a^*( h_x ) + a( h_x )$.
\begin{lemma}\label{lm:B8}
For all $0 < \delta \le 2$, we have that
\begin{equation*}
\big \| \langle \vec{x} \rangle^{-2 + \delta} \big [ \Gamma( \chi_{ | \vec{y} | \le d } ) , H_{P,E} \big ] ( N + \mathds{1} )^{-\frac12} \big \| = \mathcal{O}( d^{-2+\delta} ).
\end{equation*}
\end{lemma}
\begin{proof}
A direct computation shows that
\begin{align*}
\big [ \Gamma( \chi_{ | \vec{y} | \le d } ) , a( h_x ) \big ] = \Gamma( \chi_{ | \vec{y} | \le d } ) a ( ( \mathds{1} - \chi_{ | \vec{y} | \le d } )	 h_x ) .
\end{align*}
Applying Lemma \ref{lm:B3} (with $a=0$ and $b=1/2$), we conclude as in Lemma \ref{lm:B4} that
\begin{align*}
\big \| \langle \vec{x} \rangle^{-2 + \delta} \big [ \Gamma( \chi_{ | \vec{y} | \le d } ) , a( h_x ) \big ] ( N + \mathds{1} )^{-\frac12} \big \|  = \mathcal{O}( d^{-2+\delta} ).
\end{align*}
Since the estimate for $a^*(h_x)$ instead of $a(h_x)$ follows in the same way, the lemma is proven.
\end{proof}
The next lemma is similar to \cite[Lemma 5.2]{BoFaSi12_01} and relies on Helffer-Sj{\"o}strand functional calculus. We refer the reader to \cite{BoFaSi12_01} for the proof.
\begin{lemma}\label{lm:B5}
Let $f \in C_0^\infty( [ 0 , + \infty ) ; \mathbb{R} )$ be a smooth function satisfying the estimates $| \partial_s^m f(s) | \le C_m \langle s \rangle^{-m}$ for all $m \ge 0$. For all $d > 0$, we have that 
\begin{equation}
\big [ |\vec{k}| , f (  \vec{y}^{\, 2} / d^2  ) \big ] = \mathcal{O}( d^{-1} ) . \label{eq:s5}
\end{equation}
\end{lemma}
As a consequence of Lemma \ref{lm:B5}, we prove the following.
\begin{lemma}\label{lm:B9}
\begin{align*}
& \big \| \big [ H_E  , \Gamma( \chi_{ | \vec{y} | \le d } ) ) \big ] ( N + \mathds{1} )^{-1}  \big \| = \mathcal{O}( d^{-1} ) \\
& \big \| \big [ H_E ,  \Gamma( \chi_{ | \vec{y} | \ge 2 d } ) ) \big ] ( N + \mathds{1} )^{-1}  \big \| = \mathcal{O}( d^{-1} ).
\end{align*}
\end{lemma}
\begin{proof}
The two estimates are proven in the same way, we only establish the first one. On the $n$-photons sector, a direct computation gives
\begin{align*}
\big [ H_E , \Gamma( \chi_{ | \vec{y} | \le d } ) ) \big ] = \d\Gamma( \chi_{ | \vec{y} | \le d } , [ |k| , \chi_{ | \vec{y} | \le d } ] ) ,
\end{align*}
where $\d\Gamma( a , b )$ is defined by \eqref{eq:checkdGamma}--\eqref{eq:checkdGamma2}. Applying Lemma \ref{lm:B5}, we immediately deduce that
\begin{align*}
\big \|Ê\big [ H_E , \Gamma( \chi_{ | \vec{y} | \le d } ) ) \big ] |_{ \mathcal{H}_E^{(n)} } \big \| \le C n d^{- 1}  ,
\end{align*}
where $\mathcal{H}_E^{(n)}$ denotes the $n$-photons subspace. Since the constant $C$ in the previous estimate is uniform in $n \in \mathbb{N}$, the lemma follows.
\end{proof}
\vspace{3mm}

\input{AppendixD}

\nocite{*}
\bibliographystyle{plain}
\bibliography{dipole-effective-dynamics}

\end{document}

%% file: Introduction.tex
A key reason why, in science, we are able to successfully describe natural processes quantitatively is that if some process of interest is far isolated from the rest of the world it can be described as if nothing else were present in the universe; i.e., it can be viewed as a process happening in a ``closed system''. This means, for example, that a condensed-matter experimentalist studying a magnetic material does not have to worry about astrophysical processes inside the sun, in order to understand the magnetic properties of the material in his earthly laboratory. Nor, for that matter, does he have to worry about what his colleague in the laboratory next door is doing, provided he is not experimenting with strong magnetic fields. It is the purpose of our paper to show that the notion of ``closed systems'', in the sense just sketched, makes sense in quantum mechanics -- in spite of the phenomena of entanglement and of the ``non-locality'' of Bell correlations.

Rather than engaging in a general, abstract discussion, we propose to study some concrete models of quantum systems, $S:=P \vee Q \vee E$, composed of two spatially far separated subsystems $P$ and $Q$ coupled to a common environment $E$ (that can be empty). We discuss various sufficient conditions implying that, for a large class of initial states of $S$ including ones exhibiting entanglement between $P$ and $Q$, the time-evolution of expectation values of observables, $O_P$, of the subsystem $P$ behaves as if the subsystem $Q$ were absent. The interesting ones among our sufficient conditions turn out to be uniform in the number of degrees of freedom of the subsystem $Q$. Our results can be interpreted as saying that there is ``no signaling'' between $P$ and $Q$, provided that these subsystems are spatially far separated from one another, independently of whether the initial state of $P\vee Q$ is entangled, or not, and independently of the number of degrees of freedom of $Q$. In other words, our conditions guarantee that $P$ can be considered to be a ``closed system'' (or ``isolated system'').
(Absence of signaling has previously been discussed, e.g., in \cite{shimony1984,EbRo,Peres2004}.)

We will discuss two models. In a first model, we choose $P$ to describe a quantum particle moving away from a system $Q$ that may have very many degrees of freedom; the environment $E$ is absent. It will be assumed that, in a sense to be made precise below, interactions between $P$ and $Q$ become weaker and weaker, as the distance between the two subsystems increases. One purpose of our discussion of this model is to show that quantum mechanics does not admit a realistic interpretation -- in the sense that knowing the unitary time evolution of a system and its initial state does not enable one to predict what happens in the future -- and that it is intrinsically probabilistic. In a second, more elaborate model, the subsystems $P$ and $Q$ are allowed to exchange quanta of a quantum field (such as photons or phonons), i.e., $P$ and $Q$ can\textit{``communicate''} by emitting and absorbing field quanta; accordingly, the environment $E$ is  chosen to consist of a quantum field, e.g., the electromagnetic field or a field of lattice vibrations. The goal of our discussion is to isolate conditions that enable us to derive an ``effective dynamics'' of the subsystem $P$ that does not explicitly involve the environment $E$ and is independent of $Q$.

To keep our analysis down to earth, we will only study systems $P$ and $Q$ (with finitely many, albeit arbitrarily many degrees of freedom) that can be described in the usual Hilbert-space framework of non-relativistic quantum mechanics, with the time evolution given by a unitary one-parameter group. The ``observables'' are taken to be bounded selfadjoint operators on a Hilbert space. (For simplicity, the environment $E$ will be assumed to have temperature zero, with pure states corresponding to unit rays in Fock space.)

Concretely, the Hilbert space of pure state vectors of the system $S = P\vee Q \vee E$ is given by
\begin{equation}
\mathcal{H} = \mathcal{H}_P \otimes \mathcal{H}_Q \otimes \mathcal{H}_E,
\end{equation}
where $\mathcal{H}_P$, $\mathcal{H}_Q$ and $\mathcal{H}_E$ are separable Hilbert spaces.
General states of $S$ are given by \textit{density matrices}, i.e, positive trace-class operators, $\rho$, of trace $1$ acting on $\mathcal{H}$.
General observables of the entire system $S = P\vee Q \vee E$ are self-adjoint operators in $B(\mathcal{H}_P \otimes \mathcal {H}_Q \otimes \mathcal{H}_E)$, where $B(\mathcal{H})$ is the algebra of all bounded operators on the Hilbert space $\mathcal{H}$.
Observables refering to the subsystem $P$ are selfadjoint operators of the form
\begin{equation}
O_P = O \otimes \mathds {1}_{\mathcal {H}_{Q\vee E}}, \qquad O = O^{*} \in B(\mathcal{H}_P).
\end{equation}
 
 A state $\rho$ of the entire system $S$ determines a state $\rho_P$ of the subsystem $P$ (a reduced density matrix) by
 \begin{equation}
 Tr_{\mathcal{H}_P}(\rho_P A) := Tr(\rho (A\otimes \mathds{1}_{\mathcal{H}_{Q\vee E}})),
 \end{equation}
 for an arbitrary operator $A \in B(\mathcal{H}_P)$.
 
Time evolution of $S$ is given by a unitary one-parameter group $(U(t))_{t\in \mathbb{R}}$ on 
$\mathcal{H}$.

We are now ready to clarify what we mean by \textit{``closed systems''}:
Informally, $P$  can be viewed as a  closed subsystem of $S$ if there exists a one-parameter unitary group $(U_P(t))_{t\in \mathbb{R}}$ on $\mathcal{H}_P$ such that
\begin{equation}
\label{iso2}
Tr(\rho  U(t)^{*}(A\otimes  \mathds{1}_{\mathcal{H}_{Q\vee E}}) U(t))  \approx Tr_{\mathcal{H_P}}(\rho_PU_P(t)^{*}AU_P(t)),
\end{equation}
for a suitably chosen  subset of density matrices $\rho$  and all times in some  interval contained in $\mathbb{R}$. Mathematically precise notions of ``closed subsystems'' will be proposed in the context of the two models analyzed in this paper, and we will subsequently present sufficient conditions for $P$ to be a closed subsystem of $S$.\\

 The plan of our article is as follows. In  subsections \ref{Sec2-1} and \ref{Sec2-2} we introduce the models analyzed in this paper. The first model describes a quantum particle, $P$, with spin $1/2$ interacting with a large quantum system $Q$ and moving away from $Q$. (The subsystem $Q$ may consist of another quantum particle entangled with $P$ and a ``detector''. The two particles are prepared in an initial state chosen such that they move away from each other, with $P$ moving away from the detector.) This example will be useful in a discussion of some aspects of the foundations of quantum mechanics, in particular of the intrinsically probabilistic nature of quantum mechanics. The second model describes a neutral atom $P$ with a non-vanishing electric dipole moment that interacts with a large  quantum system $Q$. Both $P$ and $Q$ are coupled to the quantized electromagnetic field, $E$. In this model, $S$ corresponds to the composition $P\vee Q\vee E$. The point is to identify an effective dynamics for $P$ that does not make explicit reference to the electromagnetic field $E$. Our results on these models are stated and interpreted in subsections \ref{Sec2-1r} and  \ref{Sec2-2r}, respectively. In subsection \ref{Sec2-1e}, we sketch some concrete experimental situations described, at least approximately, by our models.
 
Proofs of our main results are presented in section \ref{proofs}. Many of the techniques used in our proofs  are inspired by ones used in previous works on scattering theory; see, e.g., \cite{Schi, DeGe99_01,FrGrSc01_01,FrGrSc02_01,FaSi14_01,FaSi14_02}. Some technical lemmas are proven in two appendices.

\subsubsection*{Acknowledgement} J. Fr. thanks P. Pickl and Chr. Schilling for numerous stimulating discussions on models closely related to the first model discussed in our paper. J. Fa. and J. Fr. are grateful to I.M. Sigal for many useful discussions on problems related to the second model and, in particular, on scattering theory. J. Fa.'s research is supported by ANR grant ANR-12-JS01-0008-01.
\vspace{2mm}

\section{Summary and interpretation of main results}

\subsection{Model 1:  A quantum particle $P$ interacting with a large quantum system $Q$ } 
\subsubsection{Description of the model} \label{Sec2-1}
We consider  a quantum particle,  $P$, of mass $m=1$ and  spin $1/2$; (throughout this paper, we employ  units where $\hbar=c=1$).  The particle interacts with a  large quantum system, $Q$, which we keep as general as possible. The pure states of the composed system, $P\vee Q$, correspond to unit rays in the Hilbert space $\mathcal{H}= \mathcal{H}_P \otimes \mathcal{H}_{Q}$, where $ \mathcal{H}_P := L^{2}(\mathbb{R}^3) \otimes \mathbb{C}^2$ and $\mathcal{H}_{Q}$ is a separable Hilbert space. The dynamics of $P\vee Q$ is specified by a selfadjoint Hamiltonian
\begin{equation}\label{Ham}
H=H_{P} \otimes \mathds{1}_{\mathcal{H}_Q} + \mathds{1}_{\mathcal{H}_P} \otimes H_{Q} + H_{P,Q}
\end{equation}
defined on a dense domain $\mathcal{D}(H) \subset \mathcal{H}$. In  \eqref{Ham}, $$H_{P}:=-\frac{\Delta}{2} \otimes \mathds{1}_{\mathbb{C}^2}.$$ 
The operator $H_P$  and $H_Q$, defined on their respective domains, are self-adjoint. \\

\noindent {\bf{Remark}}. \textit{It is not important to exclude the presence of external fields or potentials acting on the particle. All that matters is that the propagation of the particle approaches the one of a free particle as time tends to $\infty$.
To keep our analysis simple we assume that  if the interaction between $P$ and $Q$ is turned off then 
$P$ propagates freely.}\\

To identify $P$ as a closed subsystem of $S= P \vee Q$, one assumes that\\[2pt]
\begin{enumerate} 
\item  the strength of the interaction between $P$ and $Q$ (described by the operator $H_{P,Q}$) decays to zero rapidly as the ``distance'' between $P$ and $Q$ tends to $\infty$; and\\[4pt]
\item the initial state of the system is chosen such that the particle $P$ propagates away from $Q$, the distance between $P$ and $Q$ growing ever larger. (We will actually choose the initial state such that, with very high probability, the particle $P$ is scattered into a cone far separated from the subsystem $Q$.)\\[2pt]
\end{enumerate}
A graphical illustration of  Assumptions (1) and (2) is given below.
\input{Fig1}

\noindent  Next, we reformulate Assumptions (1) and (2) in mathematically precise terms. It is convenient to identify $L^{2}(\mathbb{R}^3) \otimes \mathbb{C}^2 \otimes  \mathcal{H}_Q$ with $L^{2}(\mathbb{R}^3;  \mathbb{C}^2 \otimes  \mathcal{H}_Q)$. We denote by $(\vec{e}_x,\vec{e}_y,\vec{e}_z)$ three orthonormal vectors in  $\mathbb{R}^3$.
\vspace{2mm}

\begin{itemize}\label{hyp:part}

\item[ \textbf{(A1)} ] (\textit{Location of $Q$ and properties of the interaction Hamiltonian}) There is an open subset  $\Omega \subset \mathbb{R}^3$ (possibly unbounded), the ``spatial location'' of the subsystem $Q$, separated from the cone 
 \begin{equation}
\mathcal{C}_{2\theta_0}:= \lbrace  \vec{k} \in \mathbb{R}^3 \mid  \vec{k}\cdot \vec{e}_x \geq  \vert \vec{k} \vert \cos(2 \theta_0)\rbrace
\end{equation}
with $\pi/4>\theta_0>0$,  by  a  distance  $d>0$,
and a covering 
$$\Omega  = \bigcup_{n \in I} \Omega_n, \qquad I \subseteq \mathbb{N} $$
of $\Omega$ by open cubes $\Omega_n$ of uniformly bounded diameter such that
\vspace{1mm}

\begin{enumerate}[(i)]
\item  the interaction Hamiltonian $H_{P,Q}$ can  be written  as a strongly convergent sum of  operators,
\begin{equation}
H_{P,Q}=\sum_{n \in I} H_{P,Q_n },
\end{equation}
on the dense domain $ \mathcal{D}(H_{P,Q}) \supseteq \mathcal{D}(H_P) \otimes \mathcal{D}(H_{Q})$.

The operator $H_{P,Q_n}$ encodes the interaction between the particle $P$ and the subsystem of $Q$  located in the cube $\Omega_n$. The distance between $\Omega_n$ and the cone $\mathcal{C}_{2\theta_0}$ is denoted by $d_n$ and is supposed to tends to $+ \infty$, as $n$ tends to $\infty$;
\vspace{1mm}

\item  there is a constant $\alpha>1$ and a sequence $\lbrace N_n \rbrace_{n\in I}$ of  operators on  $\mathcal{H}_Q$ with the properties that
 \begin{equation}\label{2.6}
\| (H_{P,Q_n}   \Psi) (\vec{x})\|_{\mathbb{C}^2 \otimes \mathcal{H}_Q} \leq \frac{\|(\mathds{1}_{\mathcal{H}_P} \otimes N_{n}) \Psi(\vec{x}) \|_{\mathbb{C}^2 \otimes \mathcal{H}_Q}  }{[\text{dist}(\Omega_n,\vec{x})]^{\alpha}}, \qquad \vec{x} \in  \Omega^{c},
\end{equation}
 for all $n \in I$ and for all $\Psi \in  \mathcal{D}(H_P) \otimes \mathcal{D}(H_{Q})$, and
\begin{equation}\label{2.8}
\sum_{n\in I} d_{n}^{\frac{1 - \alpha}{2}} \leq  C d^{-\beta}, \text{ for some }  \beta > 0, C < \infty.
\end{equation}
 Furthermore,
 $[H_{P,Q_n},\vec{x}]=0$ for all $n \in I$.

\end{enumerate}
\vspace{2mm}

\item[ \textbf{(A2)} ]  (\textit{Choice of initial state}) The initial  state
$\Psi_0 \in \mathcal{S}(\mathbb{R}^3;\mathbb{C}^2 \otimes  \mathcal{H}_Q)$, $\| \Psi_0 \|=1$, is a smooth function of 
$\vec{x}$ of rapid decay with values in $\mathbb{C}^2 \otimes  \mathcal{H}_Q$.  Its Fourier transform,
\begin{equation}
\widehat{\Psi}_0(\vec{k}):=\frac{1}{(2 \pi)^{3/2}}\int_{\mathbb{R}^3}   e^{-i\vec{k} \cdot \vec{x}} \Psi_0(\vec{x}) \text{ }d^3 x ,
\end{equation}
has  support  in the conical region $ \mathcal{C}_{\theta_0;v}$ defined by
\begin{equation}
\mathcal{C}_{\theta_0;v}:= \lbrace  \vec{k} \in \mathbb{R}^3 \mid  \vec{k}\cdot \vec{e}_x \geq  \vert \vec{k} \vert \cos(\theta_0),   \vert \vec{k} \vert >v \rbrace
\end{equation}
for some $v>0$.

\vspace{2mm}

\item[ \textbf{(A3)}]  (\textit{Bound on the number of particles in $\Omega_n$}) For $s \in \{1,2\}$,
\label{clas}
\begin{equation} \label{2.7}
\|( \mathds{1}_{\mathcal{H}_P} \otimes N_n) e^{-itH_Q} \Psi_0  \|_{L^s(\mathbb{R}^3; \mathbb{C}^2 \otimes \mathcal{H}_Q)} < C,  \qquad  \forall t \geq 0, \forall n \in I. 
\end{equation} 

\end{itemize}
\vspace{2mm}

\begin{remark}
\noindent  Assumption \textbf{(A2)} guarantees that the distance between $P$ and $Q$ grows in time  with  very high  probability. The hypotheses \textbf{(A1)} and \textbf{(A3)} are  mathematical reformulations of Assumption (1). The operator $N_{n}$ can be thought of as counting the number of  ``particles'' of the system $Q$ contained  in the subset  $\Omega_n$, for all $n \in I$.  If  the system $Q$ is composed   of   identical particles, $ \mathcal{H}_{Q}$  is the bosonic/fermionic Fock space over $L^{2}(\mathbb{R}^3 ; \mathbb{C}^p)$, ($p=1,2,...$) and the operator $N_n$ is  the second quantization of the multiplication operator by the characteristic function $\mathds{1}_{\Omega_n}$. 

The decomposition of $\Omega$ into cubes is used  to get bounds that are  uniform in the number of degrees of freedom of the system $Q$.

 We  observe that the decay in \eqref{2.6}  is faster than the one of  the Coulomb potential. To justify \eqref{2.6} one would have to invoke screening.

\end{remark}
\subsubsection{ Result }\label{Sec2-1r}
The reduced density matrix, $\rho_{P}$, of the particle $P$ corresponding to the state $\Psi_0 \in L^{2}(\mathbb{R}^3 ; \mathbb{C}^2\otimes \mathcal{H}_Q)$ of the entire system $S$ is  defined by  
$$\langle \varphi_1, \rho_{P} \varphi_2 \rangle_{L^2(\mathbb{R}^3 ; \mathbb{C}^2)}: = \sum_{j \in J} \langle  \varphi_1 \otimes e_j,\Psi_0\rangle \langle \Psi_0, \varphi_2 \otimes e_j \rangle,$$
where $\varphi_1,\varphi_2$ are arbitrary vectors in $L^{2}(\mathbb{R}^3 ; \mathbb{C}^2)$ and 
$\lbrace e_j \rbrace_{j \in J}$ is an orthonormal basis in $\mathcal{H}_Q$.
\vspace{2mm}

\begin{lemma}
\label{1.1}
We require assumptions  \textbf{(A1)},   \textbf{(A2)} and   \textbf{(A3)}. Then, for  all $\eta>0$, there exists a length $d(\eta,v)>0$ such that, for any $d>d(\eta,v)$, 
\begin{equation}
\begin{split}
\label{res1}
\big  \vert &\langle e^{-itH}  \Psi_0 , (O_{P} \otimes \mathds{1}_{\mathcal{H}_Q}) e^{-itH}  \Psi_0  \rangle-  \text{Tr}_{\mathcal{H}_P} (\rho_{P} e^{itH_{P}} O_{P} e^{-itH_{P}} )  \big  \vert \leq  \eta \|O_P\|, 
\end{split}
\end{equation}
for all $O_P \in \mathcal{B}(L^{2}(\mathbb{R}^3 ; \mathbb{C}^2))$ and for all $t \geq 0$. 
\end{lemma}
\vspace{2mm}

Lemma \ref{1.1}  justifies considering $P$ as a closed subsystem:  any  observable of the subsystem $P$ evolves as if $Q$ were absent,  up to an error term  that can be made arbitrarily  small by increasing the   separation between $P$ and $Q$. A similar  result   was already discussed in \cite{Schi}, but  with a finite range interaction between $P$ and $Q$. The proof of Lemma \ref{1.1} is given in Appendix \ref{D}.

\subsubsection{A concrete example where Lemma \ref{1.1}  can be applied} \label{Sec2-1e}
We choose $Q$  to be composed of   a particle $P'$  of spin $1/2$ (electron) and of a spin filter $D$.  The particles $P$ and $P'$ are scattered into  opposite cones, and the spin filter $D$ selects  the particle $P'$ according to  its   spin component along  the  axis corresponding to a  unit vector $\vec{n}$.   A Stern-Gerlach-type experiment is  added to the setup  to measure  a component of the spin of the particle $P$.
\begin{center}
\begin{tikzpicture}[scale=0.8]
\draw  [fill=black!0](-1,-1.5) rectangle (0,1.5); 
\draw  [ultra thick,color=black,->](-0.6,-0.5) -- (-0.6,0.5); 
\draw[thick,dotted,color=black,-] (-2,0) -- (12,0);
\draw[-] (0,1.3) -- (4.07,0.5);
\draw[-] (0,-1.3) -- (4.1,-0.5);
\draw[-] (5.7,0.5) -- (10.5,1.3);
\draw[-] (5.73,-0.5) -- (9.2,-1.2);

\draw[very thick,color=gray!50,->] (8,0) -- (8,0.5);
\draw[very thick,color=black,->] (7,0) -- (7,-0.5);
\draw[very thick,color=black,->] (-1.5,0) -- (-1.5,0.5);
\draw[very thick,color=gray!50,->] (-0.45,0) -- (-0.45,-0.5);
\draw(8,0.25) node[right]{$P$};
\draw(7,-0.25) node[right]{$P$};
\draw(-1.5,0.25) node[right]{$P'$};
\draw(-0.45,-0.25) node[right]{$P'$};
\draw(11.5,+0.3) node[right]{\small $50\%$};
\draw(11.5,-0.3) node[right]{\small $50\%$};
\draw(-0.5,-1.8) node[below]{\small Spin filter};

\begin{scope}[rotate=0,xshift=9.6cm,yshift=-1.35cm]
\draw (0,0,0)--(1,0,0)--(1,0.5,0)--(0,0.5,0)--cycle; 
\draw  (0,0,1)--(1,0,1)--(1,0.5,1)--(0,0.5,1)--cycle; 
\draw (0,0,0) -- (0,0,1); 
\draw (1,0,0) -- (1,0,1); 
\draw (1,0.5,0) -- (1,0.5,1); 
\draw(0,0.5,0) -- (0,0.5,1); 
\end{scope}

\begin{scope}[rotate=0,xshift=11cm,yshift=1.2cm]
\draw (0,0,0)--(1,0,0)--(1,0.5,0)--(0,0.5,0)--cycle; 
\draw  (0,0,1)--(1,0,1)--(1,0.5,1)--(0,0.5,1)--cycle; 
\draw (0,0,0) -- (0,0,1); 
\draw (1,0,0) -- (1,0,1); 
\draw (1,0.5,0) -- (1,0.5,1); 
\draw(0,0.5,0) -- (0,0.5,1); 
\end{scope}

\draw [->,color=gray!70](8.5,0) .. controls (11,0.05) and (10.25,0) .. (12.5,1);
\draw [->,color=black](8.5,0) .. controls (11,-0.05) and (10.25,0) .. (11.5,-1);

\draw [->, very thick](5.5,0)--(6.3,0);
\draw [->, very thick](4.5,0)--(3.7,0);
\draw(6,0) node[above]{\tiny $ \text{ } particle$ $ P$};
\draw(4.1,0) node[above]{\tiny $particle$ $ P'$};

\draw [-,dashed,color=gray](10.4,-0.85) .. controls (11.2,-0.2) and (11.2,-0.2) .. (11.2,0.85);
\draw [-,dashed,color=gray](10.35,-0.85) .. controls (11,-0.2) and (11,-0.2) .. (11.2,0.95);
\draw [-,dashed,color=gray](10.3,-0.85) .. controls (10.8,-0.2) and (10.8,-0.2) .. (11.2,1.05);
\draw [-,dashed,color=gray](10.1,-0.85) .. controls (10.1,0.3) and (10.1,0.3) .. (11.2,1.1);
\draw [-,dashed,color=gray](10.05,-0.85) .. controls (9.9,0.3) and (9.9,0.3) .. (11.2,1.2);
\draw [-,dashed,color=gray](10.,-0.85) .. controls (9.8,0.3) and (9.8,0.3) .. (11.2,1.3);

\end{tikzpicture}
\end{center}
The Hilbert space of  the system   $S=P \vee P' \vee D$  is $\mathcal{H}_P \otimes \mathcal{H}_{P'} \otimes \mathcal{H}_{D}$, where $\mathcal{H}_P=\mathcal{H}_{P'}=L^2(\mathbb{R}^3 ; \mathbb{C}^2)$.  We set $$\vert \uparrow \rangle := \left( \begin{array}{c}  1 \\ 0 \end{array} \right), \qquad  \vert \downarrow \rangle := \left( \begin{array}{c}  0 \\ 1 \end{array} \right).$$

\noindent We assume that the initial state  is an entangled state of the form 
\begin{equation}
\Psi_0=\frac{1}{\sqrt{2n}}  \sum_{j=1}^{n} \left(  \vert  \phi_{j,P} , \downarrow; \psi_{j,P'},\uparrow  \rangle -   \vert  \phi_{j,P}, \uparrow; \psi_{j,P'} , \downarrow \rangle  \right) \otimes  \vert  \chi_{j}  \rangle,
\end{equation}
where  $\| \phi_{j,P}\|_{L^2} = \| \psi_{j,P'}\|_{L^2}=1$, $\langle \chi_{i} ,\chi_j \rangle_{ \mathcal{H}_D } = \delta_{ij}$ for all $i,j=1,...,n$.  We assume that  Assumptions \textbf{(A1)}, \textbf{(A2)} and \textbf{(A3)} of Section  \ref{Sec2-1} are  fulfilled  by the Hamiltonian $H$ and by the initial state $\Psi_0$ of  the composed system $S=P \vee Q$, with $Q=P' \vee D$.  For simplicity, we neglect interactions between $P'$ and $P$. However, it is not hard to generalize Lemma \ref{1.1} to a situation where $P$ and $P'$ interact via a repulsive two-body potential, assuming that the momentum space support of the wave function $\psi_{j,P'}$ is contained in a cone opposit to $\mathcal{C}_{2 \theta_0}$, for all $j=1,...n$.

 We denote by  $\vec{S}_{P} =  (\sigma_x,\sigma_y,\sigma_z)$ the spin operator of the particle $P$, where $\sigma_x$, $\sigma_y$, $\sigma_z$ are the Pauli matrices. The next corollary is a direct consequence of Lemma \ref{1.1}.  
\begin{corollary}[No-signaling]  \label{coco}
Let $\eta >0$.  We require  assumptions  \textbf{(A1)}, \textbf{(A2)} and \textbf{(A3)} of Paragraph \ref{Sec2-1}. Then there is a distance $d(\eta,v)>0$ such that, for any $d>d(\eta,v)$,
\begin{equation}
\label{right}
\vert \langle e^{-itH} \Psi_0, \vec{S}_{P} e^{-itH} \Psi_0 \rangle \vert < \eta, \qquad \forall t \geq 0.
\end{equation}
\end{corollary}

 For sufficiently large values of  $d$, Corollary \ref{coco} shows that the mean value of the spin operator of the particle $P$  very nearly vanishes  for  \textit{all} times $t$, regardless of  the initial state vectors $\lbrace\chi_{i}\rbrace$ of the filter $D$. In particular, the expectation value of the spin operator of $P$ is  independent of the kind of measurement on $P'$ performed by the spin filter $D$.  Here we assume that a particle $P'$ with  $\vec{S}_{P'} \cdot \vec{n}=1/2$ passes the filter, while a particle with $\vec{S}_{P'} \cdot \vec{n}=-1/2$ is absorbed by $D$,  with probability very close to $1$.  A realistic interpretation of quantum mechanics, in the sense that the time evolution of pure states in the Schr\"odinger picture would completely predict  what  will happen, necessary fails. It would  lead to the prediction that 
 \begin{equation}
\label{wrong}
\langle e^{-itH} \Psi_0 , (\vec{S}_{P} \cdot \vec{n})  e^{-itH} \Psi_0 \rangle \approx-\frac{1}{2}
\end{equation}
 for  sufficiently large times $t$ if the particle $P'$ has passed the filter $D$. This  contradicts  Eq. \eqref{right}.   It shows that  choosing a unitary time evolution and specifying an initial state does not predict the results of measurement, but only probabilities for  the outcomes of such  measurement.   Our  conclusion remains valid if the particles $P$ and $P'$ are indistinguishable particles; see \cite{FS0,Schi}.

  \subsection{Model 2: A neutral atom coupled to a   quantum system $Q$ and to the  quantized electromagnetic field}
 \subsubsection{The model} \label{Sec2-2}
 We consider a neutral atom $P$ that interacts with a quantum system $Q$ and the quantized electromagnetic field. The atom  either moves freely  or moves in a slowly varying external potential. We assume that, initially, it is localized (with a probability close to one)  far away from the system $Q$, and we allow the system $Q$ to create and annihilate photons. Our aim is to prove a result of the form of  \eqref{iso2}.  Our estimates for this model are however \textit{not}  uniform in the number of degrees of freedom of the subsystem $Q$.  This problem could be  solved by decomposing $Q$ into small subsystems. This complication is avoided to  keep our exposition as simple as possible.  We do not specify the nature of $Q$, but we emphasize that it could represent  another atom or a  molecule.    The internal degrees of freedom of the atom $P$   are  described by a two-level system. The total Hilbert space of the system $S$ is the tensor product space
\begin{equation*}
\mathcal{H} := \mathcal{H}_{P} \otimes \mathcal{H}_{Q} \otimes \mathcal{H}_E,
\end{equation*}
where  
\begin{equation*}
\mathcal{H}_{P}:=L^{2}(\mathbb{R}^{3}) \otimes \mathbb{C}^2  \qquad  \text{ and }  \qquad \mathcal{H}_E := \mathcal{F}_{+}(L^{2}( \underline{\mathbb{R}}^3 )) 
\end{equation*}
are  the Hilbert spaces associated to  the atom and to the electromagnetic field,  respectively. Here  $\mathcal{F}_{+}(L^{2}( \underline{\mathbb{R}}^3 )) $ is the (symmetric) Fock space over  $L^{2}( \underline{\mathbb{R}}^3 )$. We use  the notation
\begin{equation*}
\underline{\mathbb{R}}^3:= \mathbb{R}^3\times\{1,2\} = \left\{\underline{k} := (\vec{k},\lambda)\in\mathbb{R}^3\times\{1,2\} \right\}, \qquad d \underline{k}= \sum_{\lambda=1,2} d^3k,
\end{equation*}
where $\vec{k}$ is the photon momentum  and $\lambda$ denotes the polarization of the photon.  Any element $\Phi \in \mathcal{F}_{+}(L^{2}( \underline{\mathbb{R}}^3 ))$ can be represented as a sequence $(\Phi^{(n)})$ of  totally symmetric $n$-photons functions. The scalar product on $\mathcal{F}_{+}(L^{2}( \underline{\mathbb{R}}^3 ))$ is defined by 
$$\langle \Phi, \Psi \rangle = \sum_{n \geq 0} \int_{\underline{\mathbb{R}}^{3n}} \overline{\Phi}^{(n)}(\underline{k}_1,..., \underline{k}_n) \Psi^{(n)}(\underline{k}_1,..., \underline{k}_n)  d\underline{k}_1...d\underline{k}_n $$
for all $\Phi,\Psi \in \mathcal{F}_{+}(L^{2}( \underline{\mathbb{R}}^3 ))$.

The Hamiltonian of the total system is written as
 \begin{align*}
H := & H_{P} \otimes \mathds{1}_{\mathcal{H}_Q} \otimes \mathds{1}_{\mathcal{H}_E} + \mathds{1}_{\mathcal{H}_P} \otimes H_Q \otimes \mathds{1}_{\mathcal{H}_E} + \mathds{1}_{\mathcal{H}_P} \otimes \mathds{1}_{\mathcal{H}_Q} \otimes H_{E} \\
& + H_{P , E} + H_{ P , Q} + H_{ Q, E } ,
\end{align*}
where
\begin{equation*}
H_{P}:=-\frac{\Delta}{2}  +  \left( \begin{array}{cc} \omega_0 &0\\0&0 \end{array}   \right) 
\end{equation*}
is the free atomic Hamiltonian, with $\omega_0$ the energy of the excited internal state of the atom, $H_Q$ the Hamiltonian for the system $Q$, and
\begin{equation*}
H_{E} := \d \Gamma( | \vec{k} | ) \equiv \int_{ \underline{\mathbb{R}}^3 } \vert \vec{k} \vert a^{*}( \underline{k} ) a ( \underline{k} ) d \underline{k}
\end{equation*}
is the second quantized Hamiltonian of the free electromagnetic field. The operator-valued distributions $a (\underline{k}):=a _{\lambda}(\vec{k})$ and $a ^*(\underline{k}):=a^{*}_{\lambda}(\vec{k})$ are the photon  annihilation and creation operators. We suppose that $H_Q$ is a semi-bounded self-adjoint operator on $\mathcal{H}_Q$. In what follows, we write $H_{P}$ for $H_{P} \otimes \mathds{1}_{\mathcal{H}_Q} \otimes \mathds{1}_{\mathcal{H}_E}$, and likewise for $H_Q$ and $H_E$, unless  confusion may arise.

The interaction Hamiltonians, $H_{P , E}$, $H_{ P , Q}$ and $H_{ Q, E }$ describe the interactions between the atom, the system $Q$, and the quantized field. The atom-field interaction is of the form $H_{P,E} = - \vec{d} \cdot \vec{E}$,  where $\vec{d}= - \lambda_0 \vec{\sigma}$ is the dipole moment of the atom,  $\vec{\sigma}$ is the vector of Pauli matrices, and $\vec{E}$ is the quantized electric field, i.e.,
\begin{equation*}
H_{P,E} := i \lambda_0 \int_{ \underline{\mathbb{R}}^3 } \chi( \vec{k} ) \vert \vec{k} \vert^{\frac12} \vec{\varepsilon} (\underline{k} ) \cdot \vec{\sigma} \left(  e^{i \vec{ k} \cdot \vec{x}}  a ( \underline{k} ) - e^{ - i \vec{ k} \cdot \vec{x}}  a^* ( \underline{k} ) \right) d \underline{k},
\end{equation*}
 where $\chi \in \mathrm{C}_0^\infty( \mathbb{R}^3 ; [0,1] )$ is an ultraviolet-cutoff function such that $\chi \equiv 1$ on $\{ \vec{k} \in \mathbb{R}^3 , | \vec{k} | \le 1/2 \}$ and $\chi \equiv 0$ on $\{ \vec{k} \in \mathbb{R}^3 , | \vec{k} | \ge 1 \}$, and $\vec{\varepsilon}(\underline{k}):=\vec{\varepsilon}_{\lambda}(\vec{k})$ are polarization vectors of the electromagnetic field in the Coulomb gauge. With the usual notations, $H_{P,E}$ can be rewritten in  the form
\begin{equation}
H_{P,E} = \Phi( h_x ) \equiv  a^*( h_x ) + a( h_x )  , \label{eq:def_hx1}
\end{equation}
with
\begin{equation}
h_x( \underline{k} ) := - i \lambda_0 \chi( \vec{k} ) \vert \vec{k} \vert^{\frac12} \vec{\varepsilon} (\underline{k} ) \cdot \vec{\sigma} e^{ - i \vec{ k} \cdot \vec{x}} .  \label{eq:def_hx2}
\end{equation}

By standard estimates (see Lemma \ref{lm:B1}), $H_{P,E}$ is $H_{P}+H_E$-bounded   with relative bound $0$. We suppose that $H_{P , Q}$ and $H_{Q,E}$ are symmetric operators  relatively bounded with respect to $H_{P} + H_Q$ and $H_Q + H_E$, respectively, and that $H$ is a self-adjoint operator with domain $\mathcal{D}( H ) = \mathcal{D}( H_{P} + H_Q + H_E) \supset H^{2}(\mathbb{R}^3) \otimes \mathbb{C}^2 \otimes \mathcal{D}(H_Q) \otimes \mathcal{D}(H_E)$. Further technical assumptions on $H_{P,Q}$ and $H_{Q,E}$ needed to state our main theorem will be  described below.

\subsubsection{Assumptions} \label{subsection:main}
 We assume that: \\[2pt]
\begin{enumerate}
\item   The support of the initial atomic wave function (at time $t=0$) is contained inside a ball $B_R$ of radius $R$.\\[4pt]
\item There is a large distance $d > R$ such that the interaction Hamiltonian \linebreak[4] $H_{P,Q} \mathds{1}_{\mathrm{dist}( B_d , \Omega ) \ge 2d }$ between the ball of radius $d$, $B_d$, centered at the same point as the ball $B_R$ containing the support of the initial wave function of the atom and the region $\Omega$ containing $Q$ is bounded in norm by $Cd^{-\beta}$, for some finite constant $C$ and some exponent $\beta>0$.\\[4pt]
\item With very high probability, there aren't any photons emitted by the subsystem $Q$ towards, nor absorbed by $Q$ from the ball of radius $3d$ centered at the same point as the ball, $B_R$, containing the support of the initial wave function of the atom.  \\[2pt]
\end{enumerate}
\input{Fig2}
 
To simplify the analysis, we suppose that the initial atomic wave function is contained inside a ball (of radius $R$) centered \emph{at the origin}, and that $Q$ is located outside the ball of radius $3d$ centered at $0$, for some \emph{fixed} $d > R$. Assumptions (2) and (3) are then replaced by the hypotheses that $H_{P,Q} \mathds{1}_{ | \vec{x} | \le d }$ is bounded by $Cd^{-\beta}$, and that $Q$  does not emit nor absorb photons inside the ball of radius $3d$ centered at the origin. These assumptions imply that the system $Q$ does not  penetrate into  the ball of  radius $3d$ centered at the origin. This hypothesis can  be weakened for concrete choices of the subsystem  $Q$.

 We recall the definition of the scattering \emph{identification operator} (see \cite{HuSp95_01}, \cite{DeGe99_01} or \cite{FrGrSc02_01} for more details) and a few other standard tools from scattering theory to rewrite  Assumptions (1) through (3) in   mathematically precise terms. Let $\mathcal{F}_{ \mathrm{fin} }$ denote the set of all vectors $ \Phi = ( \Phi^{(n)} ) \in \mathcal{F}_+( L^2( \underline{ \mathbb{R} }^3 ) )$ such that $\Phi^{(n)} = 0$ for all but finitely many $n$'s. The map $I : \mathcal{F}_{ \mathrm{fin} } \otimes \mathcal{F}_{ \mathrm{fin} } \to \mathcal{F}_{ \mathrm{fin} }$ is defined as the extension by linearity of the map
\begin{align}
I : a^*( g_1 ) \cdots a^*( g_m ) \Omega \otimes a^*( h_1 ) \cdots a^*( h_n ) \Omega \mapsto a^*( h_1 ) \cdots a^*( h_n ) a^*( g_1 ) \cdots a^*( g_m ) \Omega , \label{eq:defI}
\end{align}
for all $g_1, \dots , g_m , h_1, \dots , h_n \in L^2( \underline{ \mathbb{R} }^3 )$. The closure of $I$ on $\mathcal{H}_E \otimes \mathcal{H}_E$ is denoted by the same symbol and is called the scattering identification operator.  Observe that $I$ is unbounded.
Let 
\begin{equation*}
\mathcal{H}_0 := \mathcal{H}_{P} \otimes \mathcal{H}_E , \qquad \mathcal{H}_\infty := \mathcal{H}_Q \otimes \mathcal{H}_E .
\end{equation*}
The Hilbert space $\mathcal{H}_0$ corresponds to the atom together with photons located near the origin, whereas $\mathcal{H}_\infty$ corresponds to the system $Q$ together with photons located far from the origin. We extend the operator $I$ to the space $\mathcal{H}_0 \otimes \mathcal{H}_\infty$  by setting
\begin{equation*}
I : \mathcal{H}_0 \otimes \mathcal{H}_\infty \to \mathcal{H}.
\end{equation*}
We use $I$ to ``amalgamate'' $\mathcal{H}_0$ with $\mathcal{H}_{\infty}$. 

We recall that the Hamiltonian $H_{P \vee E}$ on $\mathcal{H}_0 = \mathcal{H}_{P} \otimes \mathcal{H}_E$ associated with the atom and the quantized radiation field,
\begin{equation*}
H_{P \vee E} := H_{P} + H_{E} + H_{P,E} ,
\end{equation*}
 is translation-invariant, in the sense that  $H_{P \vee E}$ commutes with each component of  the total momentum operator
\begin{equation*}
\vec{P}_{P \vee E}:= \vec{P}_P + \vec{P}_E = -i \vec{\nabla}_{x} + \int_{ \underline{ \mathbb{R} }^3 } \vec{k} a^{*}  (\underline{k}) a  (\underline{k}) d \underline{k}.
\end{equation*}
This implies (see e.g. \cite{BaChFaFrSi13_01} or \cite{FaFrSc13_01} for more details) that there exists a unitary map
\begin{equation*}
U : \mathcal{H}_{P} \otimes \mathcal{H}_E \rightarrow \int_{ \mathbb{R}^3 }^\oplus \mathbb{C}^{2 } \otimes \mathcal{H}_E \, d^3 p ,
\end{equation*}
such that
\begin{equation*}
U H_{P \vee E} U^{-1} = \int_{ \mathbb{R}^3 }^\oplus H( \vec{p} ) d^3 p .
\end{equation*}
For any fixed total momentum $\vec{p} \in \mathbb{R}^3$, the Hamiltonian $H( \vec{p} )$ is a self-adjoint, semi-bounded operator on $\mathbb{C}^{2 } \otimes \mathcal{H}_E$. Its expression is given in Appendix \ref{B}. It turns out that, for $| \vec{p} | < 1$ and for small coupling $ \vert \lambda_0 \vert$,  $H(\vec{p})$ has a ground state with associated eigenvalue $E(\vec{p})$, and that this ground state, $\psi(\vec{p})$, is real analytic in $\vec{p}$, for $\vert \vec{p} \vert < 1$; see \cite{FaFrSc13_01} and Theorem \ref{thm:anal} for a more precise statement.  Given $0<\nu<1$, we assume in the rest of this text that $\vert \lambda_0 \vert < \lambda_c(\nu)$, where $\lambda_c(\nu)>0$ is the critical coupling constant such that $H(\vec{p})$ has a ground state for all $\vec{p}$ with $\vert \vec{p} \vert< \nu$.  We  introduce a dressing transformation  $\mathcal{J}: L^{2}(\mathbb{R}^3) \rightarrow \mathcal{H}_{P} \otimes \mathcal{H}_E$, defined, for all $u \in L^{2}(\mathbb{R}^3)$ and for a.e. $\vec{x} \in \mathbb{R}^3$, by the expression
\begin{equation} \label{Ju}
\mathcal{J}(u)(\vec{x}):=\frac{1}{(2 \pi)^{\frac32}} \int_{ \mathbb{R}^3 } \hat{u}(\vec{p})  e^{i\vec{x} \cdot (\vec{p}-\vec{P}_E)} \chi_{\bar{B}_{\nu/2}} (\vec{p}) \psi ( {\vec{p}} ) \, d^3 p ,
\end{equation}
where $\chi_{\bar{B}_{\nu/2}} \in \mathrm{C}_0^\infty( \mathbb{R}^3 ; [0,1] )$ is such that $\chi_{\bar{B}_{\nu/2}} \equiv 1$ on $B_{\nu/4} = \{ \vec{p} \in \mathbb{R}^3 , | \vec{p} | < \nu / 4 \}$, and $\chi_{\bar{B}_{\nu/2}} \equiv 0$, outside $\bar{B}_{\nu/2} := \{ \vec{p} \in \mathbb{R}^3 , | \vec{p} | \le \nu / 2 \}$.  The state $\mathcal{J}(u)$ describes a  dressed single-atom  state.
We recall that, for any operator $a$ on $L^2( \underline{ \mathbb{R} }^3 )$, the second quantization of $a$, $\Gamma( a )$, is the operator defined on $\mathcal{H}_E$ by its restriction to the $n$-photons Hilbert space, which is given by 
\begin{equation}\label{eq:defGamma}
\Gamma( a ) |_{ L^2( \underline{ \mathbb{R} }^3 ) ^{ \otimes_s^n } } := \otimes^n a, \qquad n=0, 1,2,...
\end{equation}
and $\otimes^{0} a =\mathds{1}$.
We denote by 
$$
N:=\int_{\underline{\mathbb{R}}^3} a^{*}(\underline{k})  a(\underline{k})  d\underline{k}
$$
 the photon number operator on Fock Space. We are ready to state our main assumptions.\\[2pt]

\begin{itemize}\label{hyp:atom}
\item[ \textbf{(B1)} ]  (\emph{Initial state of atom}) 
Let $v \in L^2( \mathbb{R}^3 )$ be such that $\mathrm{supp}( v ) \subset \{ \vec{x} \in \mathbb{R}^3 , | \vec{x} | \le 1 \}$. The initial orbital wave function of the atom is supposed to be of the form
\begin{equation*}
u( \vec{x} ) =   R^{-3/2}   v ( R^{-1} \vec{x} ) ,
\end{equation*}
for some   $R \geq 1$. In particular,
\begin{equation*}
\mathrm{supp}( u ) \subset \{ \vec{x} \in \mathbb{R}^3 , | \vec{x} | \le R \},
\end{equation*}
 and  $\|u \|_{L^2}=\|v \|_{L^2}$ is independent of $R$.
\end{itemize}
\begin{itemize}
\item[ \textbf{(B2)} ] (\emph{Initial state of photons far from the atom}) \label{hyp:photons}
The state $\varphi \in \mathcal{H}_\infty = \mathcal{H}_Q \otimes \mathcal{H}_E$ satisfies 
\begin{equation*}
\big ( \mathds{1}_{ \mathcal{H}_Q } \otimes \Gamma( \mathds{1}_{ | \vec{y} | \ge 3d } ) \big ) \varphi = \varphi ,
\end{equation*}
for some $d > 0$, where $\vec{y} := i \vec{ \nabla }_k$ denotes the ``photon position variable'', and
\begin{equation*}
\varphi \in \mathcal{D}( H_{Q \vee E } ) \cap \mathcal{D}( \mathds{1}_{\mathcal{H}_Q} \otimes e^{ \delta N } ) ,
\end{equation*} 
for some $\delta > 0$.
\end{itemize}
\vspace{1mm}
\begin{itemize}
\item[ \textbf{(B3)} ] (\emph{Interaction $P-Q$}) \label{hyp:int-atS}
The interaction Hamiltonian between the atom and the subsystem $Q$, $H_{P,Q}$, is a symmetric operator on $\mathcal{H}_{P} \otimes \mathcal{H}_Q$, relatively bounded with respect to $H$, and satisfying 
\begin{equation} \label{HPQ}
\| H_{P,Q} \mathds{1}_{ | \vec{x} | \le d } \| \leq C d^{-\beta}
\end{equation}
for some constants $C$ and $\beta>0$.
\end{itemize}
\vspace{1mm}
\begin{itemize} 
\item[ \textbf{(B4)} ] (\emph{Interaction $Q-E$}) \label{hyp:int-Sph}
The interaction Hamiltonian between the subsystem $Q$ and the radiation field, $H_{Q,E}$, is a symmetric operator on $\mathcal{H}_{Q} \otimes \mathcal{H}_E$  such that $ H_{Q \vee E}= H_Q+H_E+H_{Q,E}$ is self-adjoint on  
\begin{equation*}
\mathcal{D} ( H_{Q \vee E} ) = \mathcal{D}( H_Q + H_E ).
\end{equation*}
Moreover, in the sense of quadratic forms, $H_{Q,E}$ satisfies
\begin{equation} \label{com}
\big [ H_{Q,E} , a^{\sharp}( \mathds{1}_{ | \vec{y} | \le 3 d } h ) \big ] = 0 , \qquad \big [ H_{Q,E} , \Gamma( j( \vec{y} ) ) \big ] = 0 ,
\end{equation}
for all $h \in L^2( \underline{ \mathbb{R} }^3 )$ and for all Fourier multiplication operators $j( \vec{y} )$ on $L^2( \underline{ \mathbb{R} }^3 )$ such that $j( \vec{y} ) \mathds{1}_{ | \vec{y} | \ge 3d } = \mathds{1}_{ | \vec{y} | \ge 3d }$, where $a^\sharp$ stands for $a$ or $a^*$.
\end{itemize}
\vspace{1mm}

\begin{itemize}
\item[ \textbf{(B5)} ] (\emph{Number of photons emitted by $Q$}) \label{hyp:number-photons}  The initial state $\varphi$ satisfies $e^{ -i t H_{Q \vee E} } \varphi \in \mathcal{D}( N )$ for all times $t \ge 0$, and
\begin{align}
\big \| \mathrm{d} \Gamma(\mathds{1}_{\vert \vec{y} \vert \ge c t}) e^{ -i t H_{Q \vee E} } \varphi \big \| &\le C \langle t \rangle, \label{2.20}
\end{align}
for some $c > 1$, where $C$ is a positive constant depending on $\varphi$ and $\langle t \rangle:=(1+t^2)^{1/2}$.
\end{itemize}
\begin{remark}
\textbf{(B1)}, \textbf{(B3)} and \textbf{(B4)} are direct  mathematical reformulations of the hypotheses (1), (2) and (3) above.  In \textbf{(B2)}, we assume that, initially, photons  \textit{in contact} with $Q$  are ``localized'' outside the ball of radius $3d$ centered at the origin. 

The constant $C$ in \eqref{HPQ}  depends  a priori  on the number of degrees of freedom of the subsystem $Q$. This problem could be circumvented  by decomposing $Q$ into subsystems located ever further away from $P$,  as we did for the first model. 

Assumption \eqref{com}  is very strong and can  be relaxed in concrete examples for the subsystem  $Q$. For instance, if  $H_{Q,E}$ is linear in annihilation and creation operators, Eq. \eqref{com} is not relevant and the estimate of the norm of the commutator of  $H_{Q,E}$ with other operators on Fock space can be  carried  out directly. The calculations are the same as for the operator $H_{P,E}$.

  Assumption  \textbf{(B5)} implies that the number of  photons created by $Q$  does not grow faster than linearly in time. Indeed, using Hardy's inequality  and the fact that $\mathcal{D}( H_{Q \vee E} ) \subset \mathcal{D}( H_E )$, we have that
\begin{align*}
\big \| \mathrm{d} \Gamma(\mathds{1}_{\vert \vec{y} \vert  \leq c t}) e^{ -i t H_{Q \vee E} } \varphi \big \|  \leq  c t  \big \| \mathrm{d} \Gamma \big (   \vert \vec{y} \vert^{-1} \big) e^{ -i t H_{Q \vee E} } \varphi \big \| \leq  c t \|   H_E e^{ -i t H_{Q \vee E} } \varphi \big \| \le C t .
\end{align*}
\eqref{2.20} says that the number of photons emitted by $Q$ and traveling faster than light grows at most linearly in time.  Eq.  \eqref{2.20} could be weakened  by  a polynomial growth.  This would lead to worse estimates in Theorem \ref{thm:isolated} below.  Assumption (\textbf{B5}) is not  fully satisfactory,  since the upper bound may depend on the number of degrees of freedom of $Q$. The main reason why we impose  \eqref{2.20} is that  photons are massless. The operator $N$ is not $H_E$-bounded, and some of our estimates cannot be proven if we do not control the time evolution of the total number of photons.

 For massive particles,  the dispersion law $\omega( \vec k)= \vert \vec k \vert $ is replaced by $\omega( \vec k)= \sqrt{ { \vec k}^2+ m^2}$, where $m>0$ is the mass of the particles of the field. Since $N$ is $H_E$-bounded, and since, under our assumptions, $\mathcal{D}( H_{Q \vee E} ) \subset \mathcal{D}( H_E )$, we have that
\begin{equation*}
\big \| N e^{ -i t H_{Q \vee E} } \varphi \big \| \le C \big ( \| H_{Q \vee E} \varphi \| + \| \varphi \| \big ) .
\end{equation*}
Hypothesis \textbf{(B5)} is therefore obviously satisfied for massive particles.  To simplify our  presentation, we only state and prove our main result for  photons.
\end{remark}

\subsubsection{Main Result } \label{Sec2-2r}  Our aim is to show that, under Assumptions \textbf{(B1)-(B5)},  $P$ behaves as  a closed system  over a finite  interval of times. For $a,b >0$,  we write $a=\mathcal{O}(b)$
if there is a constant $C>0$ independent of  $t$, $d$ and $R$,  such that $a \le C b$.
\begin{theorem}\label{thm:isolated}
 Consider an initial state $\psi \in \mathcal{H}$ of the form
\begin{equation*}
\psi =  \frac{1}{\| \sum_{i=1}^{l} I ( \mathcal{J}( u_i ) \otimes \varphi_i ) \|} \sum_{j=1}^{l} I ( \mathcal{J}( u_j ) \otimes \varphi_j ) ,
\end{equation*}
where $u_i$ and $\varphi_i$ satisfy Assumptions \textbf{(B1)},  \textbf{(B2)} and \textbf{(B5)} with $d >R^2 \ge 1$, for  $i=1,\dots, l$, and $\langle \varphi_i, \varphi_j \rangle=\delta_{ij}$ for all $i,j=1,\dots, l$. We  introduce the density matrix
$$\rho_{P}:=  \frac{1}{ \sum_{i=1}^{l} \| \mathcal{J}( u_i ) \|^2}  \sum_{j=1}^{l} \vert \mathcal{J}( u_j ) \rangle   \langle \mathcal{J}( u_j )\vert \in \mathcal{B}(\mathcal{H}_0).$$
 Suppose, moreover,  that Asumptions \textbf{(B3)} and \textbf{(B4)}  are satisfied.  Then 
 \begin{align*}
\big \langle e^{ - i t H} \psi ,  &(O_{P} \otimes \mathds{1}_{\mathcal{H}_Q} \otimes \mathds{1}_{\mathcal{H}_E}) e^{ - i t H} \psi \big \rangle = \text{Tr}_{\mathcal{H}_0}(\rho_P  e^{  i t H_{P \vee E}} (O_{P} \otimes \mathds{1}_{\mathcal{H}_E}) e^{ - i t H_{P \vee E}}  ) \\
& + \|O_{P} \| \left(\mathcal{O} \big ( (d/R)^{ \frac{-1+\gamma}{2} } \big ) + \mathcal{O} \big ( \langle t \rangle (d / R^2 )^{ - \frac{1}{2} } \big ) + \mathcal{O} ( t^2 d^{-\frac12} ) +  \mathcal{O} (t d^{-  \beta  })\right) 
\end{align*}
 for all  $0 < \gamma \le 1$, all $t \ge 0$, and all  $O_{P} \in B(\mathcal{H}_{P})$.

\end{theorem}

\subsubsection{A corollary: the dressed atom in a slowly varying external potential}
\label{Model2pot}

We now assume that the atom is placed in a slowly varying external potential, $V_{\varepsilon}( \vec{x} ) \equiv V( \varepsilon \vec{x} )$, with $V_\varepsilon \in L^\infty( \mathbb{R}^3 ; \mathbb{R} )$. We set
\begin{equation*}
 H_{P}^{\varepsilon} := H_{P} + V_{\varepsilon}( \vec{x} ), \qquad H_{P \vee E}^{\varepsilon} := H_{P \vee E} + V_{\varepsilon}( \vec{x} ), \qquad H^{\varepsilon} := H + V_{\varepsilon}( \vec{x} ) .
 \end{equation*}
 We define the effective Hamiltonian $H^{\varepsilon}_{P,\mathrm{eff}}$ on $L^{2}(\mathbb{R}^3)$ as
\begin{equation*}
H^{\varepsilon}_{P,\mathrm{eff}} := E( -i \vec{\nabla}_{x}) + V_{\varepsilon} ( \vec{x} ),
\end{equation*}
where $E(\vec{p})$ is the ground state energy of the fiber Hamiltonian $H(\vec{p})$.  Since  $V_{\varepsilon}$ is bounded, $H_{P}^{\varepsilon} $, $H_{P \vee E}^{\varepsilon} $,  and $H^{\varepsilon}$   are self-adjoint on $\mathcal{D}(H_{P})$,  $\mathcal{D}(H_{P \vee E})$ and $\mathcal{D}(H)$, respectively.   

\begin{corollary}\label{cor:effective}
Suppose that $V \in L^\infty( \mathbb{R}^3 ; \mathbb{R} )$ satisfies $\mathrm{supp} (\hat{V}) \subset B_{1} = \{ \vec{x} \in \mathbb{R}^3 , | \vec{x} | < 1 \}$. Set   $u_{i}(t):=e^{ - i t H_{P,\mathrm{eff}}^\varepsilon }  u_i $ with $u_i \in H^2( \mathbb{R}^3 )$, and 
$$\rho_{\varepsilon}(t):=  \frac{1}{ \sum_{i=1}^{l}  \| \mathcal{J}( u_{i}(t) )   \|^2}   \sum_{j=1}^{l}  \vert \mathcal{J}( u_{j}(t) ) \rangle   \langle \mathcal{J}(u_{j}(t) )\vert$$
for all $t \geq 0$. Under the assumptions of Theorem \ref{thm:isolated},  we have that
\begin{align*}
\big \langle e^{ - i t H^\varepsilon } \psi , &(O_{P} \otimes \mathds{1}_{\mathcal{H}_Q} \otimes \mathds{1}_{\mathcal{H}_E})  e^{ - i t H^\varepsilon } \psi \big \rangle = \text{Tr}_{\mathcal{H}_0}(\rho_{\varepsilon}(t)  (O_{P} \otimes \mathds{1}_{\mathcal{H}_E}) )   + \|O_{P} \| \Big(  \mathcal{O}( t \varepsilon    )  \\
& + \mathcal{O} \big ( (d/R)^{ \frac{-1+\gamma}{2} } \big ) + \mathcal{O} \big ( \langle t \rangle (d / R^2 )^{ - \frac{1}{2} } \big ) + \mathcal{O} ( t^2 d^{-\frac12} ) + \mathcal{O} (t d^{-  \beta  })\Big) ,
\end{align*}
for all $0 < \gamma \le 1$, all $t \geq 0$ and all $O_{P} \in \mathcal{B}(\mathcal{H}_{P})$.
\end{corollary}
\vspace{3mm}
This result is similar to the one proven in \cite{BaChFaFrSi13_01}.


%% file: Fig1.tex
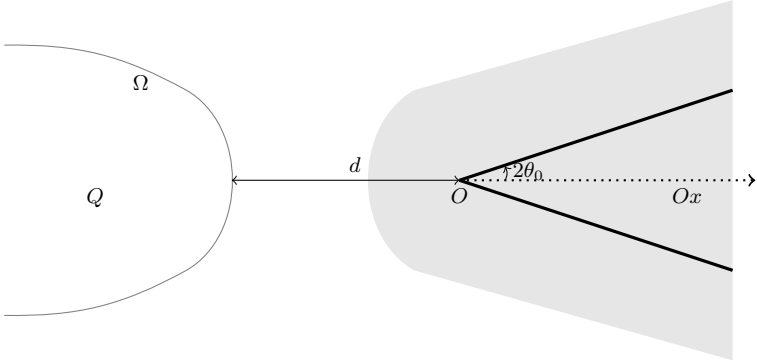
\begin{figure}[H]
\begin{center}
\begin{tikzpicture}[scale=0.6, every node/.style={scale=0.9}]

         \path[draw=gray, fill opacity=0.0001] 
             (-3,5)to[bend left=15]  (1,4) to[bend left=60]  (1,0)to[bend left=15] (-3,-1)   ;
        
        \draw[-, very thick] (7,2)--(13,4);
         \draw[-, very thick] (7,2)--(13,0);
          \draw[->,dotted,  thick] (7,2) -- (13.5,2);
          \draw (12,2) node[below] {$Ox$}; 
         \draw (7,2) node[below] {$O$}; 
       \draw[->](8,2) to[bend right=30] (8,2.3);
           \draw (8,2.2) node[right] {$2\theta_0$};

               \draw (-1,2) node[below] {$Q$};

  \path[fill=black, fill opacity=0.1] 
             (13,6)-- (6,4) to[bend left=-60]  (6,0)-- (13,-2) ;
             
     \draw (0,4.5) node[below] {$\Omega$};            
  
   \draw[<->](2,2) -- (7,2);      
  
   \draw (4.7,2) node[above] {$d$};

        \end{tikzpicture}
\caption{ \footnotesize  The system $Q$ is localized inside the domain $\Omega$. The particle $P$ scatters inside the grey colored set with a probability very close to $1$.  }\label{fig1}
\end{center}
\end{figure}

%% file: Fig2.tex
\begin{figure}[H]
\label{fig2}
\begin{center}
\begin{tikzpicture}[scale=0.7, every node/.style={scale=0.9}]

   
    \fill[black,opacity=0.1]     (10,3)to[bend left=45] (13,2)  to[bend left=60] (14,-1)   to[bend left=40] (10,0.1)  to[bend left=60] (9,-0.2)  to[bend left=60] (8.5,0.15)  to[bend left=60] (10,3);

      \fill[black,opacity=0.1]     (2,1)to[bend left=45] (3.5,0)    to[bend left=40] (2,-2.1)  to[bend left=60] (1,-2.2)  to[bend left=60] (-0.5,-1.85)  to[bend left=60] (2,1);       
         
       \fill[white] (6cm,0cm) circle(2cm);

        \fill[black,opacity=0.25]     (6,0.3)to[bend left=45] (6.2,0.25)  to[bend left=-60] (6.3,0.2)   to[bend left=40] (6.4,0.1)  to[bend left=60] (6,-0.2) to[bend left=-40] (5.8,-0.19) to[bend left=60] (5.7,-0.15) to[bend left=-10] (5.8,0.2) to[bend left=60] (6,0.3);

 \draw[thick] (6cm,0cm) circle(2cm);

  \draw (10,2) node[below] {$Q$}; 
    \draw (1,0) node[below] {$Q$}; 
        
    \draw[<->](6,0.25) -- (6,2);      
   \draw (6,1) node[right] {\footnotesize$ 3d$};  
   
      \draw[<->](6,0) -- (6.4,0);      
   \draw (6.2,0) node[above] {\footnotesize $R$};  
   
         \end{tikzpicture}
\end{center}
\end{figure}

%% file: Localization.tex
\subsection{Localization of photons}\label{section:loc}

In this section, we begin by verifying that in the dressed atom state 
\begin{equation}\label{eq:defJ}
\mathcal{J}(u)(\vec{x}) = \frac{1}{(2 \pi)^{\frac32}} \int_{ \mathbb{R}^3 } \hat{u}(\vec{p})  e^{i\vec{x} \cdot (\vec{p}-\vec{P}_E)} \chi_{\bar{B}_{\nu/2}} (\vec{p}) \psi ( {\vec{p}} ) \, d^3 p ,
\end{equation}
(with $u \in L^2( \mathbb{R}^3 )$ as in Hypothesis \textbf{(B1)}), \emph{most} photons are localized near the origin. Here $\psi( \vec{p} )$ is a non-degenerate ground state of $H( \vec{p} )$,  and  $\vec{p} \mapsto \psi( \vec{p} )$ is real analytic on $\{ \vec{p} \in \mathbb{R}^3 , | \vec{p} | < \nu \}$, for any $0 < \nu <1$; see Theorem \ref{thm:anal} of Appendix \ref{B}  for more details.

Next, we consider the evolution of the state $\mathcal{J}( u )$  under  the dynamics generated by $H_{P \vee E}^\varepsilon$; (we recall that $H_{P \vee E}^\varepsilon$ is the Hamiltonian for the atom in the external potential $V_\varepsilon$ and interacting with the photon field). Using that the propagation velocity of photons is finite, we are able to prove that, for times not too large, most photons remain localized in the ball of radius $d \gg 1$ centered at the origin. This property will be important in the proof of our main theorem (see Section \ref{section:main_proof}), since it will allow us to show that the interaction between photons close to the atom and the system $Q$ remains small for times not too large.

We begin with three lemmas whose proofs are postponed to Appendix \ref{B}. The first one establishes polynomial decay in $\vert \vec{x} \vert$ in the state $\mathcal{J}(u)$, assuming that $u$ is compactly supported. It is a simple consequence of standard properties of the Fourier transform combined with the analyticity of $\vec{p} \mapsto \psi(\vec{p})$ (where, recall, $\psi(\vec{p})$ is a ground state of the fiber Hamiltonian $H(\vec{p})$). In what follows, we use the identification $\mathcal{H}_{P} \otimes \mathcal{H}_{E} \simeq L^2 ( \mathbb{R}^3 ; \mathbb{C}^2 \otimes \mathcal{H}_{E} )$.
\begin{lemma}\label{lm:a1}
Let $u \in L^2( \mathbb{R}^3 )$ be as in Hypothesis \textbf{(B1)}. Then
\begin{equation*}
\big \| ( 1 + | \vec{x} | )^{\mu} Ê\mathcal{J}( u )  \big \|_{ \mathcal{H}_{P} \otimes \mathcal{H}_{E} } \le \mathrm{C}_{\mu} R^\mu,
\end{equation*} 
for all $\mu \ge 0$, where $\mathrm{C}_\mu$ is a positive constant independent of $R>1$.
\end{lemma}
In the next lemma, we control the number of photons in the fibered ground state $\psi( \vec{p} )$. Based on the pull-through formula, the proof of Lemma \ref{lm:a2} follows the one of Lemma 1.5 in \cite{juerg}.
\begin{lemma}\label{lm:a2}
Let $0 < \nu < 1$. For all $\vec{p} \in \mathbb{R}^3$ such that $| \vec{ p } | < \nu$ and $\delta \in \mathbb{R}$, we have that
\begin{equation*}
\psi( \vec{p} ) \in \mathcal{D} ( e^{ \delta N } ). 
\end{equation*}
Moreover,
\begin{equation*}
\sup_{ \vec{p} \, \in \bar{B}_{ \nu/2 } } \| e^{ \delta N } \psi( \vec{p} ) \| \le \mathrm{C}_\delta .
\end{equation*}
\end{lemma}
Next, using Lemmas \ref{lm:a1} and \ref{lm:a2}, we prove that, in the dressed atom state $\mathcal{J}(u)$ (with $u$ as in Hypothesis \textbf{(B1)}), most photons are localized in the ball $\{ \vec{y} \in \mathbb{R}^3 , | \vec{y} | \le d \}$. 
\begin{lemma}\label{lm:a3}
Let $u \in L^2( \mathbb{R}^3 )$ be as in Hypothesis \textbf{(B1)}. Then, for all $d > R \ge 1$ and $0 < \gamma \le 1$,
\begin{equation*}
\big \| \Gamma( \mathds{1}_{ | \vec{y} | \le d } ) \mathcal{J} ( u ) - \mathcal{J} ( u ) \big \| = \mathcal{O} \big ( (d/R)^{-1+\gamma} \big ).
\end{equation*}
\end{lemma}
The next lemma  is another, related consequence of Lemmas \ref{lm:a1} and \ref{lm:a2}. Considering the dynamics generated by $H_{P \vee E}^\varepsilon$  and the initial state $\mathcal{J}( u )$, with $u$ as in Hypothesis \textbf{(B1)},  Lemma \ref{lm:f1}  shows that, after times $t$ such that $0 \le t \ll d/R$, most photons remain localized in the ball $\{ \vec{y} \in \mathbb{R}^3 , | \vec{y} | \le d \}$. This is a consequence of  the fact that the propagation velocity of photons is finite.
\begin{lemma}\label{lm:f1}
Let $u \in L^2( \mathbb{R}^3 )$ be as in Hypothesis \textbf{(B1)}. Then, for all $d > R \ge 1$ and $t \ge 0$,
\begin{equation}\label{eq:f2}
\big \| \big ( \mathds{1} - \Gamma( \mathds{1}_{ | \vec{y} | \le d } ) \big ) e^{ - i t H_{P \vee E}^\varepsilon } \mathcal{J} ( u ) \big \| = \mathcal{O}\big ( t^{\frac54} d^{-\frac12} \big ) + \mathcal{O} \big ( \langle t \rangle^{\frac34} ( d / R )^{-\frac12} \big ),
\end{equation}
where we recall the notation $\langle t \rangle := ( 1 + t^2 )^{1/2}$.
\end{lemma}
\begin{proof}
We begin with establishing two preliminary estimates.

\vspace{0,2cm}

\noindent \textbf{Step 1}. We have that
\begin{equation}\label{eq:r1}
\big \| N e^{ - i t H_{P \vee E}^\varepsilon } \mathcal{J} ( u ) \big \| = \mathcal{O} ( \langle t \rangle ).
\end{equation}
By Lemma \ref{lm:a2}, we know that $\mathcal{J}( u ) \in \mathcal{D}( N )$. To see that $e^{ - i t H_{P \vee E}^\varepsilon } \mathcal{J} ( u ) \in \mathcal{D}( N )$ for all $t \in \mathbb{R}$, we observe that $\D(H_{P}) \otimes \mathcal{F}_{ \mathrm{fin} }( \mathcal{D}( |\vec{k}| )  )$ is a common core for $H_{P \vee E}^\varepsilon$ and $N$, and that $( N + \mathds{1} )^{-1}$ preserves $\D(H_{P}) \otimes \mathcal{F}_{ \mathrm{fin} }(  \mathcal{D}( |\vec{k}| )  )$. Here we use the notation $$\mathcal{F}_{ \mathrm{fin} }( V ) = \{ \Phi = \Phi^{(n)} \in \mathcal{H}_{E} , \, \Phi^{(n)} \in \otimes_s^n V \text{ for all } n, \text{ } \Phi^{(n)} = 0 \text{ for all but finitely many } n \}.$$ Since the commutator
\begin{align}\label{eq:r2}
\big [  N , H_{P \vee E}^\varepsilon \big ] = \big [  N , H_{P,E} \big ] = - i \Phi( i h_x ) , 
\end{align}
is both $H_{P \vee E}^\varepsilon$- and $N$-relatively bounded, we deduce from \cite[Lemma 2]{GeGe99_01} that $e^{ - i t H_{P \vee E}^\varepsilon } \mathcal{D}( N ) \subset \mathcal{D}( N )$ and hence in particular that $e^{ - i t H_{P \vee E}^\varepsilon } \mathcal{J} ( u ) \in \mathcal{D}( N )$. Here, as in \eqref{eq:def_hx1}--\eqref{eq:def_hx2}, $\Phi( i h_x ) := a^*( i h_x ) + a( i h_x )$.

To prove \eqref{eq:r1}, we use that
\begin{align*}
\big \| N e^{ - i t H_{P \vee E}^\varepsilon } \mathcal{J} ( u ) \big \| &= \big \| e^{ i t H_{P \vee E}^\varepsilon } N e^{ - i t H_{P \vee E}^\varepsilon } \mathcal{J} ( u ) \big \| \\
& \le \big \| N \mathcal{J} ( u ) \big \| + \int_0^t \big \| \big [  N , H_{P \vee E}^\varepsilon \big ] e^{ - i s H_{P \vee E}^\varepsilon } \mathcal{J} ( u ) \big \| ds.
\end{align*}
It follows from Lemma \ref{lm:a2} that $\| N \mathcal{J} ( u ) \|Ê= \mathcal{O}( 1 )$. Since $\mathcal{D}( H_{P \vee E}^\varepsilon ) \subset \mathcal{D}( H_E) \subset \mathcal{D}( H_E^{1/2} )$ and since $\Phi ( i h_x )$ is $H_E^{1/2}$-relatively bounded, we deduce from \eqref{eq:r2} that there are positive constants $a$ and $b$ such that 
\begin{align*}
\big \| \big [  N , H_{P \vee E}^\varepsilon \big ] e^{ - i s H_{P \vee E}^\varepsilon } \mathcal{J} ( u ) \big \| \le a\big \| H_{P \vee E}^\varepsilon \mathcal{J} ( u ) \big \| + b \| \mathcal{J}(u) \| .
\end{align*}
Since $\mathcal{J}( u ) \in \mathcal{D}( H_{P \vee E} ) = \mathcal{D}( H_{P \vee E}^\varepsilon )$, this concludes the proof of Step 1.

\vspace{0,2cm}

\noindent \textbf{Step 2}. We have that
\begin{equation*}
\big \| | \vec{x} | e^{ - i t H_{P \vee E}^\varepsilon } \mathcal{J} ( u ) \big \| = \mathcal{O} ( t ) + \mathcal{O}( R ) .
\end{equation*}
The proof of this estimate is  similar to Step 1. The only differences are that $\| | \vec{x} | \mathcal{J} ( u ) \| = \mathcal{O}( R )$ by Lemma \ref{lm:a1}, and that the commutator $[ | \vec{x} | , H_{P \vee E}^\varepsilon ] = [ | \vec{x} | , - \Delta_{ \vec{x} } ]/2$ is relatively bounded with respect to $( - \Delta_{\vec{x} } + \mathds{1} )^{1/2}$. The latter property follows from the computation
\begin{align*}
[ | \vec{x} | , - \Delta_{ \vec{x} } ] = \vec{\nabla}_x \cdot \frac{ \vec{x} }{ | \vec{x} | } + \frac{ \vec{x} }{ | \vec{x} | } \cdot \vec{\nabla}_x = 2 \frac{ \vec{x} }{ | \vec{x} | } \cdot \vec{\nabla}_x  + \frac{ 2 }{ | \vec{x} | } ,
\end{align*}
and the fact that $| \vec{x} |^{-1}$ is relatively bounded with respect to $( - \Delta_{\vec{x} } + \mathds{1} )^{1/2}$ by Hardy's inequality in $\mathbb{R}^3$.

\vspace{0,2cm}

\noindent \textbf{Step 3}. Now we proceed to the proof of Lemma \ref{lm:f1}. We introduce a smooth function $\chi_{ \cdot \le d } \in \mathrm{C}_0^\infty ( [0 , \infty ) ; [ 0 , 1] )$ satisfying $\chi_{ r \le d } \mathds{1}_{ r \le d /2 } = \mathds{1}_{ r \le d /2 }$ and $\chi_{ r \le d } \mathds{1}_{ r \ge d } = 0$. We observe that, to prove \eqref{eq:f2}, it is sufficient to establish that
\begin{equation}\label{eq:f3}
\big \| \big ( \mathds{1} - \Gamma( \chi_{ | \vec{y} | \le d } ) \big ) e^{ - i t H_{P \vee E}^\varepsilon } \mathcal{J} ( u ) \big \| =  \mathcal{O}\big ( t^{\frac54} d^{-\frac12} \big ) + \mathcal{O} \big ( \langle t \rangle^{\frac34} ( d / R )^{-\frac12} \big ) .
\end{equation}
The proof of \eqref{eq:f3} relies   on the following argument:
\begin{align*}
& \big \| \big ( \mathds{1} - \Gamma( \chi_{ | \vec{y} | \le d } ) \big ) e^{ - i t H_{P \vee E}^\varepsilon } \mathcal{J} ( u ) \big \|^2 \\
& \leq  \big \langle \mathcal{J}( u ) , e^{ i t H_{P \vee E}^\varepsilon } \big ( \mathds{1} - \Gamma( \chi_{ | \vec{y} | \le d } ) \big ) e^{ - i t H_{P \vee E}^\varepsilon } \mathcal{J} ( u ) \big \rangle \\
& = \big \| \big ( \mathds{1} - \Gamma( \chi_{ | \vec{y} | \le d } ) \big ) \mathcal{J} ( u ) \big \|^2 + i \int_0^t \big \langle \mathcal{J} ( u ) , e^{ i s H_{P \vee E}^\varepsilon } \big [ \Gamma( \chi_{ | \vec{y} | \le d } ) , H_{P \vee E}^\varepsilon \big ] e^{ - i s H_{P \vee E}^\varepsilon } \mathcal{J} ( u ) \big \rangle ds .
\end{align*}
The first term on  the right side of this inequality is of order $\mathcal{O} ( (d/R)^{-2+2\gamma} )$ for all $0 < \gamma \le 1$, by Lemma \ref{lm:a3}. Here we choose $\gamma = 1/2$. To estimate the second term, we compute the commutator
\begin{align*}
\big [ \Gamma( \chi_{ | \vec{y} | \le d } ) , H_{P \vee E}^\varepsilon \big ] = \big [ \Gamma( \chi_{ | \vec{y} | \le d } ) , H_E\big ] + \big [ \Gamma( \chi_{ | \vec{y} | \le d } ) , H_{P,E} \big ] ,
\end{align*}
and estimate each term separately. In Appendix \ref{app:comm} (see Lemma \ref{lm:B9}), we verify that
\begin{equation*}
\big \| \big [ \Gamma( \chi_{ | \vec{y} | \le d } ) , H_E\big ] ( N + \mathds{1} )^{-1} \big \| = \mathcal{O}( d^{-1} ).
\end{equation*}
Together with Step 1, this shows that
\begin{align}
\int_0^t \big \langle \mathcal{J} ( u ) , e^{ i s H_{P \vee E}^\varepsilon } \big [ \Gamma( \chi_{ | \vec{y} | \le d } ) , H_E\big ] e^{ - i s H_{P \vee E}^\varepsilon } \mathcal{J} ( u ) \big \rangle ds = \mathcal{O}\big ( t^2 d^{-1} \big ) . \label{eq:f4}
\end{align}
Using again Appendix \ref{app:comm} (see Lemma \ref{lm:B8}), we have that
\begin{equation*}
\big \| \langle \vec{x} \rangle^{-1} \big [ \Gamma( \chi_{ | \vec{y} | \le d } ) , H_{P,E} \big ] ( N + \mathds{1} )^{-\frac12} \big \| = \mathcal{O}( d^{-1} ).
\end{equation*}
Combined with Step 1 and Step 2, this implies  that
\begin{align*}
\int_0^t \big \langle \mathcal{J} ( u ) , e^{ i s H_{P \vee E}^\varepsilon } \big [ \Gamma( \chi_{ | \vec{y} | \le d } ) , H_{P,E} \big ] e^{ - i s H_{P \vee E}^\varepsilon } \mathcal{J} ( u ) \big \rangle ds = \mathcal{O}\big ( t^{\frac52} d^{-1} \big ) + \mathcal{O} \big ( t^{\frac32} ( d / R )^{-1} \big ) .
\end{align*}
This estimate and  \eqref{eq:f4} imply \eqref{eq:f3}, which concludes the proof of the lemma.
\end{proof}

 \vspace{2mm}
We conclude this subsection with another localization lemma that will be useful in the proof of Theorem \ref{thm:isolated}. In spirit, Lemma \ref{lm:g1} is similar to Lemma \ref{lm:f1} and  follows from the fact that the propagation velocity of photons is finite. For the dynamics generated by $H_{Q \vee E}$, it shows that, if in the initial state $\varphi$ all photons are localized in the region $\{ \vec{y} \in \mathbb{R}^3 , | \vec{y} | \ge 3d \}$, then in the evolved state $e^{-it H_{Q \vee E} } \varphi$, with $t \ll d$, most  photons are localized in $\{ \vec{y} \in \mathbb{R}^3 , | \vec{y} | \ge 2d \}$.
 \vspace{2mm}
 
\begin{lemma}\label{lm:g1}
Let $\varphi \in \mathcal{H}_\infty = \mathcal{H}_Q \otimes \mathcal{H}_{E}$ be as in Hypothesis \textbf{(B2)} and suppose that Hypotheses \textbf{(B4)} and \textbf{(B5)} hold. For all $d > 0$ and $t \ge 0$,
\begin{equation}\label{eq:g2}
\big \| \big ( \mathds{1} - \Gamma( \mathds{1}_{ | \vec{y} | \ge 2 d } ) \big ) e^{ - i t H_{Q \vee E} } \varphi \big \| = \mathcal{O} \big ( t^2 d^{-1} \big ) .
\end{equation}
\end{lemma}
 \vspace{2mm}

\begin{proof}
Let $\chi_{ r \ge 2d } \in \mathrm{C}_0^\infty ( [0 , \infty ) ; [ 0 , 1] )$ be a smooth function satisfying $\chi_{ r \ge 2d } \mathds{1}_{ r \ge 3d } = \mathds{1}_{ r \ge 3d }$ and $\chi_{ r \ge 2d } \mathds{1}_{ r \le 2d } = 0$. With this definition, we see that, in order to prove \eqref{eq:g2}, it is sufficient to show that 
\begin{equation}\label{eq:g3}
\big \| \big ( \mathds{1} - \Gamma( \chi_{ | \vec{y} | \ge 2 d } ) \big ) e^{ - i t H_{Q \vee E} } \varphi \big \| = \mathcal{O} \big ( t^2 d^{-1} \big ) .
\end{equation}
As in Lemma \ref{lm:f1}, we use  that 
\begin{equation*}
\big \| \big ( \mathds{1} - \Gamma( \chi_{ | \vec{y} | \ge 2 d } ) \big ) e^{ - i t H_{Q \vee E} } \varphi \big \| \le \int_0^t \big \| \big [ \Gamma( \chi_{ | \vec{y} | \ge 2 d } ) , H_{Q \vee E} \big ] e^{ - i s H_{Q \vee E} } \varphi \big \| ds.
\end{equation*}
We have that $[ \Gamma( \chi_{ | \vec{y} | \ge 2 d } ) , H_{Q} ] = 0$, and it follows from Hypothesis \textbf{(B4)} that $[ \Gamma( \chi_{ | \vec{y} | \ge 2 d } ) , H_{Q, E}] = 0$. Therefore
\begin{equation}
\big \|Ê[ \Gamma( \chi_{ | \vec{y} | \ge 2 d } ) , H_{Q \vee E} ] ( N + \mathds{1} )^{-1} \big \|Ê= \big \|Ê[ \Gamma( \chi_{ | \vec{y} | \ge 2 d } ) , H_{E} ] ( N + \mathds{1} )^{-1} \big \| = \mathcal{O}( d^{-1} ) , \label{eq:g4}
\end{equation}
the last estimate being proven in Appendix \ref{app:comm} (see Lemma \ref{lm:B9}). By Hypothesis \textbf{(B5)} together with the assumption that $\varphi \in \mathcal{D}( N )$, we have that
\begin{equation}
\big \|Ê( N + \mathds{1} ) e^{- is  H_{Q \vee E} } \varphi \big \| = \mathcal{O}( \langle s  \rangle) . \label{eq:g5}
\end{equation}
Equations \eqref{eq:g4} and \eqref{eq:g5} imply \eqref{eq:g3}, which concludes the proof of the lemma.
\end{proof}
 \vspace{2mm}

%% file: Partition_unity.tex
\subsection{Factorization of Fock space}\label{sec:partition}

We introduce a factorization of  Fock space (see \cite{DeGe99_01} or \cite{FrGrSc02_01} for more details), which will be used to factorize $e^{-itH^{\epsilon}}$ into a tensor product of the form $e^{-itH_{P \vee E}^{\epsilon}} \otimes e^{-itH_{Q \vee E}}$ plus an error term. This factorization is carried out in Section \ref{section:main_proof} and is one of the main ingredients of our proof.  Let $j_0 \in \mathrm{C}^{\infty}_{0}( [ 0 , \infty ) ; [ 0 , 1 ] )$ be such that $j_0 \equiv 1$ on $[ 0 , 1 ]$ and $j_0 \equiv 0$ on $[ 2 , \infty )$, and let $j_\infty$ be defined by the relation $j_0^2 + j_\infty^2 \equiv 1$. Recall that $\vec{y} := i \vec{\nabla}_k$ denote the ``photon position variable''. Given $d > 0$, we introduce the bounded operators $\mathbf{j_0} := j_0 ( | \vec{y} | / d )$ and $\mathbf{j_\infty} := j_\infty( | \vec{y} | / d )$ on $L^2( \underline{ \mathbb{R} }^3 )$. We set
\begin{align*}
\mathbf{j} : L^2( \underline{ \mathbb{R} }^3 ) &\to L^2( \underline{ \mathbb{R} }^3 ) \oplus L^2( \underline{ \mathbb{R} }^3 ) \\
u & \mapsto ( \mathbf{j_0} u , \mathbf{j_\infty} u ).
\end{align*}
Next we lift the operator $\mathbf{j}$ to the Fock space $\mathcal{H}_{E} = \mathcal{F}_+( L^2( \underline{ \mathbb{R} }^3 ) )$ defining a map
\begin{equation*}
\Gamma( \mathbf{j} ) : \mathcal{H}_{E} \to \mathcal{F}_+( L^2( \underline{ \mathbb{R} }^3 ) \oplus L^2( \underline{ \mathbb{R} }^3 ) ) ,
\end{equation*}
with $\Gamma( \mathbf{j} )$ defined as in \eqref{eq:defGamma}.  Let
\begin{equation*}
\mathcal{U} : \mathcal{F}_+ ( L^2( \underline{ \mathbb{R} }^3 ) \oplus L^2( \underline{ \mathbb{R} }^3 ) ) \to \mathcal{H}_{E} \otimes \mathcal{H}_{E} ,
\end{equation*}
be the unitary operator defined by
\begin{align}
& \mathcal{U} \Omega := \Omega \otimes \Omega \label{defU_1} \\
& \mathcal{U} a^*( u_1 \oplus u_2 ) = ( a^*( u_1 ) \otimes \mathds{1} + \mathds{1} \otimes a^*( u_2 ) ) \mathcal{U} . \label{defU_2} 
\end{align}
The  factorization of  Fock space that we consider is defined by
\begin{equation*}
\check{\Gamma}( \mathbf{j} ) : \mathcal{H}_{E} \to \mathcal{H}_{E} \otimes \mathcal{H}_{E} , \qquad \check{\Gamma}( \mathbf{j} ) = \mathcal{U} \Gamma( \mathbf{j} ).
\end{equation*}
Using the relation $j_0^2 + j_\infty^2 \equiv 1$, one can verify that $\check{\Gamma}( \mathbf{j} )$ is a partial isometry. The adjoint of $\check{\Gamma}( \mathbf{j} )$ can be represented as
\begin{align}
\check{\Gamma}( \mathbf{j} )^* = I (\Gamma( \mathbf{j_0} ) \otimes \Gamma( \mathbf{j_\infty} )), \label{eq:b1}
\end{align}
where $I$ denotes the identification operator defined in \eqref{eq:defI}.

 On the total Hilbert space $\mathcal{H} = \mathcal{H}_{P} \otimes \mathcal{H}_{Q} \otimes \mathcal{H}_{E}$, we denote by the same symbol the operator 
\begin{equation*}
\check{\Gamma}( \mathbf{j} ) : \mathcal{H} \to \mathcal{H}_0 \otimes \mathcal{H}_\infty ,
\end{equation*}
where, recall, $\mathcal{H}_0 = \mathcal{H}_{P} \otimes \mathcal{H}_{E}$ and $\mathcal{H}_\infty = \mathcal{H}_{Q} \otimes \mathcal{H}_{E}$. We introduce the bounded operator
\begin{equation}\label{eq:def-chi}
\chi_{ | \vec{y} | \le d } := j_0( 2 | \vec{y} | / d ) ,
\end{equation}
on $L^2( \underline{ \mathbb{R} }^3 )$. As in Section \ref{section:loc}, it corresponds to a smooth version of the projection $\mathds{1}_{ | \vec{y} | \le d }$ satisfying $\chi_{ \vert \vec{y} \vert \le d } \in \mathrm{C}_0^\infty ( [0 , \infty ) ; [ 0 , 1] )$, $\chi_{ | \vec{y} | \le d } \mathds{1}_{ | \vec{y} | \le d /2 } = \mathds{1}_{ | \vec{y} | \le d /2 }$ and $\chi_{ | \vec{y} | \le d } \mathds{1}_{ | \vec{y} | \ge d } = 0$.

We begin with a localization lemma for the initial state $I \mathcal{J}( u ) \otimes \varphi$, which will be useful in the sequel. For the convenience of the reader, the proof of Lemma \ref{lm:b1} is deferred to Appendix \ref{B}.

\vspace{2mm}
\begin{lemma}\label{lm:b1}
Let $u \in L^2( \mathbb{R}^3 )$ be as in Hypothesis \textbf{(B1)} and $\varphi \in \mathcal{H}_{Q} \otimes \mathcal{H}_{E}$ be as in Hypothesis \textbf{(B2)}, with $d > R \ge 1$.  Then $\mathcal{J}( u ) \otimes \varphi \in \mathcal{D}( I )$ and,  for all $0 < \gamma \le 1$, we have that
\begin{equation*}
I \mathcal{J}( u ) \otimes \varphi = I \big ( \Gamma( \chi_{ | \vec{y} | \le d } ) \mathcal{J}( u ) \big ) \otimes \big ( \Gamma( \mathds{1}_{ | \vec{y} | \ge 3d } ) \varphi \big ) + \mathcal{O}( (d/R)^{ \frac{-1+\gamma}{2} } ).
\end{equation*}
\end{lemma}
\vspace{2mm}

A few remarks concerning the statement of Lemma \ref{lm:b1} are in order. In more precise terms, the lemma means that
\begin{equation*}
\big \| I \mathcal{J}( u ) \otimes \varphi - I \big ( \Gamma( \chi_{ | \vec{y} | \le d } ) \mathcal{J}( u ) \big ) \otimes \big ( \Gamma( \mathds{1}_{ | \vec{y} | \ge 3d } ) \varphi \big ) \big \| \le \mathrm{C} (d/R)^{ \frac{-1+\gamma}{2} } ,
\end{equation*}
where $\mathrm{C}$ is a positive constant depending on $\gamma$, $v$ ($v$ being the function of Hypothesis \textbf{(B1)}) and $\varphi$, but not on $R$ and $d$ such that $d > R$. We mention that the exponent $(-1+\gamma)/2$ is presumably not sharp. We do not make any attempt to optimize it. We also observe that the fact that $\mathcal{J}(u) \otimes \varphi \in \mathcal{D}( I )$ follows from Lemma \ref{lm:a2} and Hypothesis \textbf{(B2)}. The fact that $\big ( \Gamma( \chi_{ | \vec{y} | \le d } ) \mathcal{J}( u ) \big ) \otimes \big ( \Gamma( \mathds{1}_{ | \vec{y} | \ge 3d } ) \varphi \big )$ also belongs to  $\mathcal{D}(I)$ is a consequence of \eqref{eq:b1}. More precisely, since $j_0 \equiv 1$ on $[ 0 , 1 ]$ and $j_\infty \equiv 1$ on $[ 3 , \infty )$, \eqref{eq:b1} implies that
\begin{align*}
& I \big ( \Gamma( \chi_{ | \vec{y} | \le d } ) \mathcal{J}( u ) \big ) \otimes \big ( \Gamma( \mathds{1}_{ | \vec{y} | \ge 3d } ) \varphi \big ) \\
&= \check{\Gamma}( \mathbf{j} )^* \big ( \Gamma( \chi_{ | \vec{y} | \le d } ) \mathcal{J}( u ) \big ) \otimes \big ( \Gamma( \mathds{1}_{ | \vec{y} | \ge 3d } ) \varphi \big ) ,
\end{align*}
and therefore the boundedness of $\check{\Gamma}( \mathbf{j} )^*$ yields
\begin{equation*}
\big \|ÊI \big ( \Gamma( \chi_{ | \vec{y} | \le d } ) \mathcal{J}( u ) \big ) \otimes \big ( \Gamma( \mathds{1}_{ | \vec{y} | \ge 3d } ) \varphi \big ) \big \| \le \big \| \mathcal{J}(u) \big \| \| \varphi \|. 
\end{equation*}

\vspace{2mm}

In what follows, we denote by $\tilde H^\varepsilon$ the total Hamiltonian on $\mathcal{H}$ where the interaction between the atom $P$ and the system $Q$ has been removed, that is
\begin{equation}\label{eq:def-tilde-Heps}
\tilde H^\varepsilon := H^\varepsilon - H_{P,Q}.
\end{equation}
Moreover, to shorten notations, we introduce the definition
\begin{equation}\label{eq:def-psi-loc}
\psi_{ \mathrm{loc} } := I \big ( \Gamma( \chi_{ | \vec{y} | \le d } ) \mathcal{J}( u ) \big ) \otimes \big ( \Gamma( \mathds{1}_{ | \vec{y} | \ge 3d } ) \varphi \big ) = \check{\Gamma}( \mathbf{j} )^* \big ( \Gamma( \chi_{ | \vec{y} | \le d } ) \mathcal{J}( u ) \big ) \otimes \big ( \Gamma( \mathds{1}_{ | \vec{y} | \ge 3d } ) \varphi \big ) .
\end{equation}
We prove in Lemma \ref{lm:e1} that, if the system is initially in the (non-normalized) state $\psi_{ \mathrm{loc} }$, the contribution of the interaction between the atom and the subsystem $Q$, $H_{P,Q}$, to the dynamics, remains small for times not too large.  
\vspace{2mm}

\begin{lemma}\label{lm:e1}
Let $u \in L^2( \mathbb{R}^3 )$ be as in Hypothesis \textbf{(B1)} and $\varphi \in \mathcal{H}_{Q} \otimes \mathcal{H}_{E}$ be as in Hypothesis \textbf{(B2)}, with $d > R \ge 1$. Assume Hypothesis \textbf{(B3)}. For all times $t \ge 0$, we have that
\begin{equation*}
e^{ - i t H^\varepsilon } \psi_{ \mathrm{loc} } = e^{ - i t \tilde{H}^\varepsilon } \psi_{ \mathrm{loc} } + \mathcal{O} ( t d^{-1} ) + \mathcal{O}( ( d / R )^{- \infty} ) + \mathcal{O}( t d^{-\beta}).
\end{equation*}
\end{lemma}
\vspace{2mm}

\begin{proof}
We estimate the norm of 
\begin{align}
 e^{ - i t H^\varepsilon } \psi_{ \mathrm{loc} } - e^{ - i t \tilde{H}^\varepsilon } \psi_{ \mathrm{loc} } &=   (e^{-itH^\varepsilon} \chi_{ | \vec{x} | \le d } - \chi_{ | \vec{x} | \le d }  e^{-it\tilde{H}^\varepsilon})\psi_{ \mathrm{loc} }\\
 &+ (\chi_{ | \vec{x} | \le d } - \mathds{1} ) e^{-it\tilde{H}^\varepsilon}\psi_{ \mathrm{loc} } \notag 
 + e^{-itH^\varepsilon} ( \mathds{1} -\chi_{ | \vec{x} | \le d }) \psi_{ \mathrm{loc}}.  \label{eq:vn2}
\end{align}
Using unitarity of $e^{-it H^\varepsilon}$, we compute
\begin{align}
\big \| & (e^{-itH^\varepsilon} \chi_{ | \vec{x} | \le d } - \chi_{ | \vec{x} | \le d }  e^{-it\tilde{H}^\varepsilon})\psi_{ \mathrm{loc} } \big \| \\
& =\Big \| \int_0^t e^{ -i s H^\varepsilon } ( - H^\varepsilon \chi_{ | \vec{x} | \le d } + \chi_{ | \vec{x} | \le d }  \tilde H^\varepsilon ) e^{ - i (t-s) \tilde H^\varepsilon } \psi_{ \mathrm{loc} } \, ds \, \Big \| \notag \\
& \le \int_0^t \big \| H_{P,Q}  \chi_{ | \vec{x} | \le d }  e^{ - i (t-s) \tilde H^\varepsilon } \psi_{ \mathrm{loc} } \big \| \, ds +  \int_0^t \big \| [\tilde{H}^{\varepsilon}, \chi_{ | \vec{x} | \le d }] e^{ - i (t-s) \tilde H^\varepsilon } \psi_{ \mathrm{loc} } \big \| \, ds. \label{eq:c3}
\end{align}
By Hypothesis \textbf{(B3)},  we have  that
\begin{equation}
\big \| H_{P,Q} \chi_{ | \vec{x} | \le d }  e^{ - i s \tilde H^\varepsilon } \psi_{ \mathrm{loc} } \big \| \le C d^{-\beta} . \label{eq:c4}
\end{equation}
To estimate the second term on the right side of \eqref{eq:c3}, we compute 
\begin{align}
\big \|Ê\big [ \tilde H^\varepsilon , \chi_{ | \vec{x} | \ge d } \big ] ( - \Delta_{\vec{x}} + \mathds{1} )^{-\frac12} \big \| = \frac12 \big \| \big [ - \Delta_{\vec{x}} , \chi_{ | \vec{x} | \ge d } \big ] ( - \Delta_{\vec{x}} + \mathds{1} )^{-\frac12} \big \| = \mathcal{O}( d^{-1} ). \label{eq:c2}
\end{align}
Since $\mathcal{D}( \tilde H^\varepsilon ) = \mathcal{D} ( H^\varepsilon ) \subset \mathcal{D} ( -\Delta_{\vec{x}} ) \subset \mathcal{D} ( ( - \Delta_{\vec{x}} + \mathds{1} )^{\frac12} )$, there exist positive constant $a$, $b$ such that
\begin{equation*}
\|Ê( - \Delta_{\vec{x}} + \mathds{1} )^{\frac12} u \| \le a \| \tilde H^\varepsilon u \| + b \| u \| ,
\end{equation*}
for all $u \in  \mathcal{D}( \tilde H^\varepsilon )$, which, combined with \eqref{eq:c2}, yields
\begin{equation}
\big \| \big [ \tilde H^\varepsilon , \chi_{ | \vec{x} | \ge d } \big ] e^{ - i (t-s) \tilde H^\varepsilon } \psi_{ \mathrm{loc} } \big \| = \mathcal{O}( d^{-1} ). \label{eq:c6}
\end{equation}
Here we used that $\psi_{ \mathrm{loc} }$ belongs to $\mathcal{D}( \tilde H^\varepsilon )$. Indeed, by Hypothesis \textbf{(B3)}, $H_{P,Q}$ is $H^\varepsilon$-relatively bounded, and hence $\tilde H^\varepsilon$ is also $H^\varepsilon$-relatively bounded. Moreover $\mathcal{D}( H^\varepsilon ) = \mathcal{D}( H_P + H_Q + H_E )$ by assumption, and it is not difficult to verify that $\psi_{ \mathrm{loc} } \in \mathcal{D}( H_P + H_Q + H_E )$.

We  now introduce $\chi_{ | \vec{x} | \ge d } := \mathds{1} - \chi_{ | \vec{x} | \le d }$. We estimate the second term in the right side of \eqref{eq:vn2} by using
\begin{align}
\big \| \chi_{ | \vec{x} | \ge d } e^{ - i t \tilde H^\varepsilon } \psi_{ \mathrm{loc} } \big \| &\le \big \| \chi_{ | \vec{x} | \ge d } \psi_{ \mathrm{loc} } \big \| + \Big \| \int_0^t e^{ i s \tilde H^\varepsilon } \big [ \tilde H^\varepsilon , \chi_{ | \vec{x} | \ge d } \big ] e^{ - i s \tilde H^\varepsilon } \psi_{ \mathrm{loc} } \, d s \Big \| \notag \\
&\le \big \| \chi_{ | \vec{x} | \ge d } \psi_{ \mathrm{loc} } \big \| + \int_0^t \big \| \big [ \tilde H^\varepsilon , \chi_{ | \vec{x} | \ge d } \big ] e^{ - i s \tilde H^\varepsilon } \psi_{ \mathrm{loc} } \big \| \, d s . \label{eq:c1}
\end{align}
The first term on the right side of \eqref{eq:c1} is estimated as follows. The definition \eqref{eq:def-psi-loc} of $\psi_{ \mathrm{loc} }$ gives
\begin{align*}
\big \| \chi_{ | \vec{x} | \ge d } \psi_{ \mathrm{loc} } \big \| &= \big \| \chi_{ | \vec{x} | \ge d } \check{\Gamma}( \mathbf{j} )^* \big ( \Gamma( \chi_{ | \vec{y} | \le d } ) \mathcal{J}( u ) \big ) \otimes \big ( \Gamma( \mathds{1}_{ | \vec{y} | \ge 2d } ) \varphi \big ) \big \| \\
& \le \big \| \chi_{ | \vec{x} | \ge d } \mathcal{J}( u ) \big \| \| \varphi \| .
\end{align*}
Applying Lemma \ref{lm:a1} we then find that 
\begin{align}
\big \| \chi_{ | \vec{x} | \ge d } \psi_{ \mathrm{loc} } \big \| &= \mathcal{O}( (d/R)^{-\infty} ). \label{eq:c5}
\end{align}

The second term in the right  side of \eqref{eq:c1} has been already estimated in \eqref{eq:c6} above and the first term in the right side of \eqref{eq:vn2} is estimated by \eqref{eq:c5}. Equations \eqref{eq:c3}, \eqref{eq:c4},  \eqref{eq:c6}, \eqref{eq:c1} and  \eqref{eq:c5} prove the statement of the lemma.
\end{proof}
\vspace{2mm}

On the Hilbert space $\mathcal{H}_0 \otimes \mathcal{H}_\infty$, we abbreviate
\begin{equation*}
N_0 := N \otimes \mathds{1}_{ \mathcal{H}_\infty } , \qquad N_\infty := \mathds{1}_{ \mathcal{H}_0 } \otimes N ,
\end{equation*}
where, recall, $N$ is the  photon-number operator on Fock space. In the next lemma, we rewrite the Hamiltonian $\tilde H^\varepsilon$ defined in \eqref{eq:def-tilde-Heps} (total Hamiltonian without the interaction between the atom and the system $Q$) in the representation  corresponding to the factorization of Fock space given by  $\check{\Gamma}( \mathbf{j} )$. Combined with Lemma \ref{lm:e1}, Lemma \ref{lm:b2} will allow us to compare the dynamics $e^{- i t H^\varepsilon }$ on $\mathcal{H}$ with the tensor product $e^{ - i t H_{P \vee E}^\varepsilon } \otimes e^{ - i t H_{Q \vee E} }$ on $\mathcal{H}_0 \otimes \mathcal{H}_\infty$. The proof of Lemma \ref{lm:b2} is somewhat technical. It will be given in Appendix \ref{sec:lmb2}.
\vspace{2mm}
\begin{lemma}\label{lm:b2}
Assume Hypothesis \textbf{(B4)}. On $ \mathcal{D}( H_{P \vee E} ) \otimes \mathcal{D}( H_{Q \vee E} ) $, the following relation holds:
\begin{equation*}
\tilde H^\varepsilon \check{\Gamma}( \mathbf{j} )^* = \check{\Gamma}( \mathbf{j} )^* \big ( H_{P \vee E}^\varepsilon \otimes \mathds{1}_{ \mathcal{H}_\infty } + \mathds{1}_{ \mathcal{H}_0 } \otimes H_{Q \vee E} \big ) + \mathrm{Rem}_1 + \mathrm{Rem}_2  ,
\end{equation*}
with
\begin{equation*}
\big \| \mathrm{Rem}_1 ( N_0 + N_\infty + \mathds{1} )^{-1} \big \| = \mathcal{O}( d^{-1} ) ,
\end{equation*}
and
\begin{equation*}
\big \| \mathrm{Rem}_2 \big ( N_0 + N_\infty + \langle \vec{x} \rangle^{4 - 2 \delta}  \big )^{-1} \big \| = \mathcal{O}( d^{-2+\delta} ) ,
\end{equation*} 
for all $0 < \delta \le 2$, where we used the usual notation $\langle \vec{x} \rangle := \sqrt{1+ \vec{x}^2}$.
\end{lemma}

%% file: AppendixA.tex
\noindent Corollary \ref{cor:effective} is a direct consequence of Theorem \ref{thm:isolated}  and the following lemma. The proof of Lemma \ref{thm:eff-dyn} uses Theorem \ref{thm:anal} and follows the lines of \cite{BaChFaFrSi13_01}.

\begin{lemma}\label{thm:eff-dyn}
Let $u \in H^{2}(\mathbb{R}^3)$ with  $\| u \|_{L^2(\mathbb{R}^3)} =1$. Suppose that $V \in L^\infty( \mathbb{R}^3 ; \mathbb{R} )$ satisfies $\mathrm{supp} (\hat{V}) \subset B_{1} = \{ \vec{x} \in \mathbb{R}^3 , | \vec{x} | < 1 \}$. Then there exists a  constant $C>0$ such that
\begin{equation*}
\big \| e^{- i t H_{P \vee E}^{\varepsilon}} \mathcal{J}(u) - \mathcal{J} (e^{ - i t H_{P , \mathrm{eff} }^{\varepsilon} } u ) \big \| \le C t \varepsilon
\end{equation*}
for all $0 < \varepsilon < 1$ and for all  $t \ge 0$.
\end{lemma}
\begin{proof}
We define
\begin{equation}
A_t u:=e^{-it H_{P \vee E}^{\varepsilon}} \mathcal{J}(u)  - \mathcal{J}(e^{-itH_{P ,\text{eff}}^{\varepsilon}} u)
\end{equation}
for all  $u \in H^{2}(\mathbb{R}^3)$.  Since $u \in \mathcal{D}(H^{\varepsilon}_{P ,\text{eff}})$ and $ \mathcal{J}(u) \in \mathcal{D}(H_{P \vee E}^{\varepsilon})$,  $e^{-i(t-s)H^{\varepsilon}_{P \vee E}} \mathcal{J} (e^{-is H_{P ,\text{eff}}^{\varepsilon}} u  )$ is  differentiable with respect to $s$, and  we find that \begin{equation}
\label{Atder}
A_t u =  -i e^{-it H^{\varepsilon}_{P \vee E}}    \int_{0}^{t}  e^{is H^{\varepsilon}_{P \vee E}} \left(H_{P \vee E}^{\varepsilon} \mathcal{J}(u_s) - \mathcal{J}(H_{P ,\text{eff}}^{\varepsilon} u_s) \right) ds
\end{equation}
for all  $u  \in H^{2}(\mathbb{R}^3)$, where $u_s:=e^{-is H_{P ,\text{eff}}^{\varepsilon}} u$. We remind the reader that  $\psi(\vec{p})$ is the ground state of the Hamiltonian $H(\vec{p})=UH_{P \vee E}U^{-1}$ with corresponding eigenvalue $E(\vec{p})$, where $U: \mathcal{H}_{P} \otimes \mathcal{H}_E \rightarrow \int_{\mathbb{R}^3}^{\oplus}  \mathbb{C}^2 \otimes \mathcal{H}_E \text{ } d^{3}p$ is the generalized Fourier transform defined by
\begin{equation}
\label{UU}
(U \varphi) ( \vec{p}) = \frac{1}{(2 \pi)^{3/2}}\int_{\mathbb{R}^{3}}  e^{-i(\vec{p} -\vec{P}_E) \cdot \vec{y}} \varphi(\vec{y}) \text{ } d^3 y,
\end{equation} 
for all $\varphi \in L^{1}(\mathbb{R}^3; \mathbb{C}^2 \otimes \mathcal{H}_E)$ and extended to $L^{2}(\mathbb{R}^3; \mathbb{C}^2 \otimes \mathcal{H}_E)$ by a density argument. It follows that 
\begin{equation*}
\mathcal{J}(u_s)(x)=U^{-1} \left( \chi_{\bar{B}_{\nu/2}} \psi \hat{u}_s \right) (x),
\end{equation*}
and we have that 
\begin{equation}
\label{ii0}
H_{P \vee E}(\mathcal{J}(u))=H_{P \vee E} U^{-1} \left( \chi_{\bar{B}_{\nu/2}} \psi \hat{u}_s \right)=U^{-1}(E \psi  \chi_{\bar{B}_{\nu/2}}   \hat{u}_s),
\end{equation}
where $E(\vec{p})$ is the self-energy of the atom.  Furthermore,
\begin{equation}
\label{ii1}
\mathcal{J}( (H_{P ,\text{eff}}^{\varepsilon}-V_{\varepsilon})u_s)=\frac{1}{(2\pi)^{3/2}} \int_{\mathbb{R}^3}   E(\vec{p})  \hat{u}_s(\vec{p})   e^{i \vec{x} \cdot ( \vec{p}- \vec{P}_E)} \chi_{\bar{B}_{\nu/2}}(\vec{p})    \psi(\vec{p}) \text{ } d^{3}p.
\end{equation}
We observe that  \eqref{ii0} and \eqref{ii1} imply  that \eqref{Atder} is equal to 
\begin{equation}
\label{At}
A_t u = -i e^{-it H_{P \vee E}^{\varepsilon}}    \int_{0}^{t}    e^{is H_{P \vee E}^{\varepsilon}} \left(V_{\varepsilon} \mathcal{J}(u_s) - \mathcal{J}(V_{\varepsilon} u_s) \right) ds.
\end{equation}
Since $e^{-it H_{P \vee E}^{\varepsilon}} $ is an isometry, it is sufficient to bound the norm of 
\begin{equation}
\phi_s:=V_{\varepsilon} \mathcal{J}(u_s) - \mathcal{J}(V_{\varepsilon} u_s).
\end{equation}
We have that
\begin{eqnarray*}
(U \phi_s)(\vec{p}) &=& U  \left[ V_{\varepsilon} U^{-1} (\hat{u}_s \chi_{\bar{B}_{\nu/2}} \psi) - U^{-1} (\widehat{ V_{\varepsilon} u_s } \chi_{\bar{B}_{\nu/2}} \psi ) \right](\vec{p})\\
&=& \left[ \hat{V}_{\varepsilon}  * (\hat{u}_s \chi_{\bar{B}_{\nu/2}} \psi)- \widehat{ V_{\varepsilon} u_s } \chi_{\bar{B}_{\nu/2}} \psi \right](\vec{p})
\end{eqnarray*}
which can be rewritten as 
\begin{equation}
\label{Uphi}
(U \phi_s)(\vec{p}) = \frac{1}{ (2 \pi)^{3/2}} \int_{\mathbb{R}^3}   \hat{V}_{\varepsilon}(\vec{p}- \vec{q})  \hat{u}_s(\vec{q})  \left( \chi_{\bar{B}_{\nu/2}} (\vec{q}) \psi(\vec{q}) - \chi_{\bar{B}_{\nu/2}} (\vec{p}) \psi(\vec{p}) \right)d^{3}q.
\end{equation}
 
\noindent Since $U$ is an isometry,
\begin{equation}
\label{At2}
\| A_t u \| \leq \int_{0}^{t}ds  \| U \phi_s \|.
\end{equation}

\noindent We set
\begin{equation*}
\label{bigph}
\Psi(\vec{p}):=\chi_{\bar{B}_{\nu/2}} (\vec{p}) \psi(\vec{p})
\end{equation*}
for all $\vec{p} \in \mathbb{R}^3$.  As $\chi$ is smooth and $\psi$ is real analytic, $\Psi$ is  smooth with compact support and is consequently   continuously differentiable in $\mathbb{R}^{3}$:  There exists a constant $M_{\nu}>0$ such that
\begin{equation}
\label{bm}
\| \Psi(\vec{q})-\Psi(\vec{p}) \|_{\mathbb{C}^2 \otimes \mathcal{H}_E} \leq M_{\nu}  \vert \vec{p}- \vec{q} \vert
\end{equation}
for all $\vec{p},\vec{q} \in \mathbb{R}^3$.
Introducing  the function
\begin{equation}
\label{gggg}
g_{\varepsilon}(\vec{p} ):=  \varepsilon+    \vert \vec{p}  \vert,
\end{equation}
we get  that 
\begin{align*}
\| A_t u_s \| &\leq \int_{0}^{t}  ( 2 \pi)^{-3/2}  \Big \|  \int_{\mathbb{R}^3}  \hat{V}_{\varepsilon}(\vec{p}- \vec{q})  \hat{u}_s(\vec{q})  \left(\Psi(\vec{q})-\Psi(\vec{p})  \right) d^{3}q  \Big \| ds   \\
& \leq   \int_{0}^{t}    ( 2 \pi)^{-3/2} \Big \|  \int_{\mathbb{R}^3}  \hat{V}_{\varepsilon}(\vec{p}- \vec{q})  g_{\varepsilon}(\vec{p}- \vec{q}) \hat{u}_s(\vec{q})  \frac{1}{g_{\varepsilon}(\vec{p}- \vec{q})}   \left( \Psi(\vec{q})-\Psi(\vec{p}) \right) d^{3}q  \Big \| ds\\
&\leq   M_{\nu}  \int_{0}^{t}   ( 2 \pi)^{-3/2} \|  \hat{V}_{\varepsilon}   g_{\varepsilon} *  \hat{u}_s \|_{L^{2}( \mathbb{R} ^3 )} ds \leq  ( 2 \pi)^{-3/2}  M_{\nu}  \int_{0}^{t}    \|  \hat{V}_{\varepsilon}   g_{\varepsilon} \|_{L^{1}(   \mathbb{R}^3 )}  \|   \hat{u}_s \|_{L^{2}(  \mathbb{R}^3 ) } ds , 
\end{align*}
where we have used  Young's inequality in the last line. Since $\text{Supp}(\hat{V}) \subset B_1(0)$, 
\begin{equation}
\|  \hat{V}_{\varepsilon}   g_{\varepsilon}   \|_{L^{1}(  \mathbb{R} ^3 )} \leq \varepsilon  \| \hat{V}_{\varepsilon}  \|_{L^{1}(  \mathbb{R} ^3 ) }+   \varepsilon \| \hat{V} \|_{L^{1}(  \mathbb{R} ^3 )}.
\end{equation}
Together with  $ \|  \hat{u}_s \|_{L^2(  \mathbb{R} ^3 )}=1$,  $\hat{V}_{\varepsilon}(\vec{p})= \frac{1}{\varepsilon^3}  \hat{V}  \left(\frac{\vec{p}}{\varepsilon} \right) $, and $\| \hat{V}_{\varepsilon}  \|_{L^{1}(  \mathbb{R} ^3 )}=\| \hat{V}  \|_{L^{1}(  \mathbb{R} ^3 )} $,  this finally implies that
\begin{equation}
\| A_t u_s \| = \mathcal{O} (t \varepsilon).
\end{equation}

\end{proof}

%% file: AppendixD.tex
\section{Proof of Lemma \ref{1.1}}
\label{D}
 In this section, we use the  symbol $ a \apprle b$ if $a \leq  Cb$ for a positive constant $C$ independent of the problem parameters  $v$ and $d$. The proof relies on a ``Cook argument''. The sets
\begin{equation}
\label{Bt}
D_{t}:= \{  \vec{k} t \mid  \vec{k} \in \mathcal{C}_{ 2\theta_0;v}\}
\end{equation}  
satisfy
\begin{equation}
\underset{t \geq 0}{ \bigcup} D_{t}=\mathcal{C}_{2\theta_0}.
\end{equation} 
Let  $\Omega' \subset \mathbb{R}^3$   be  defined by 
\begin{equation}
\Omega':=\big \{ \vec{x} + \vec{k} \mid \vec{x} \in B_{d/4}, \vec{k} \in \mathcal{C}_{ 2\theta_0}\big \}.
\end{equation}
Remark that  
$\text{dist}(\mathcal{C}_{2\theta_0}, \partial \Omega')=d/4$.
 Let $\chi_{\Omega'}$ be a smooth function with  support  in $\Omega^{c}$, such that $\chi_{\Omega'}(x)=1$ for all $x \in \Omega'$.  We introduce $$H_0:= H_{P} \otimes \mathds{1}_{\mathcal{H}_Q} +  \mathds{1}_{\mathcal{H}_P} \otimes H_Q.$$ 
 We have that 
\begin{equation}
\label{1.10}
(e^{-itH}- e^{-itH_0})\Psi_0 = (e^{-itH} \chi_{\Omega'}-  \chi_{\Omega'} e^{-itH_0})\Psi_0 + (\chi_{\Omega'}-1 ) e^{-itH_0}\Psi_0 + e^{-itH} (1-\chi_{\Omega'}) \Psi_0.  
\end{equation}
We estimate successively the  2 first  terms on the right side of  \eqref{1.10}.
\subsection{Free evolution}
We control the free evolution with a stationary phase argument in Lemma \ref{B.1}.
\begin{lemma} \label{B.1} Let $p \in \mathbb{N}$, $p\geq 4$, and $\Psi_0 \in \mathcal{H}$ be as in \textbf{(A2)}. Then
\begin{equation}
\label{upp}
 \| (e^{-itH_{0}}  \Psi_{0})(\vec{y})  \|_{\mathbb{C}^2 \otimes \mathcal{H}_Q} \leq   \frac{K}{(\vert  \vec{y} \vert + vt)^p}
\end{equation}
for all $\vec{y}  \notin \mathcal{C}_{2\theta_0}$,
where
\begin{equation}
 K:=   C_1  \max \left( 1 ,   \frac{1}{v^p}\right) \frac{1}{[\sin(\theta_0)]^{2p}}  \sum_{ \vert \underline{\beta} \vert \leq p+1}   \|  \partial_{\underline{\beta}}  \hat{\Psi}_{0}  \|_{L^1(\mathbb{R}^3;\mathbb{C}^2 \otimes \mathcal{H}_Q)}.
\end{equation} 
$C_1$ is a positive constant that does not depend on $v$, $d$ and $\theta_0$, but does depend on $p$.
\end{lemma}
\vspace{2mm}

\begin{proof}
Let $t >0$ and  let $\vec{y} \notin D_t$. We introduce the linear differential operator $d_{\vec{y},t}: \mathcal{S}( \mathcal{C}_{ \theta_0;v}; \mathbb{C}^2 \otimes \mathcal{H}_Q) \rightarrow  \mathcal{S}( \mathcal{C}_{ \theta_0;v}; \mathbb{C}^2 \otimes \mathcal{H}_Q)$, defined by 
\begin{equation}
(d_{\vec{y},t} \Psi)(\vec{k}): =\sum_{j=1}^{3} \partial_{k_j} \left(\frac{y_j- k_jt }{\vert \vec{y}- \vec{k} t \vert^2}  \Psi(\vec{k})  \right)
\end{equation}
for all $ \Psi \in \mathcal{S}( \mathcal{C}_{ \theta_0;v},\mathbb{C}^2 \otimes \mathcal{H}_Q)$ and for all $\vec{k} \in  \mathcal{C}_{ \theta_0;v}$.
An easy calculation shows that 
\begin{equation}
\label{1.12}
(d_{\vec{y},t}  \Psi )(\vec{k}) =  \sum_{j=1}^{3}  \frac{y_j- k_j t}{\vert \vec{y}- \vec{k} t \vert^2}    \partial_{k_j} \Psi (\vec{k})  - \frac{t}{\vert \vec{y}- \vec{k} t \vert^2} \Psi (\vec{k}).
\end{equation}
 Iterating \eqref{1.12},  we get that
\begin{equation}
\| (d_{\vec{y},t}^{p}  \Psi )(\vec{k})  \|_{\mathbb{C}^2 \otimes \mathcal{H}_Q} \leq f^{(p)}(\vec{k},\vec{y},t) \Big ( \sum_{\vert \underline{\beta} \vert \leq p}  \| \partial_{\underline{\beta}}  \Psi (\vec{k})\|_{\mathbb{C}^2 \otimes \mathcal{H}_Q} \Big),
\end{equation}
where $\underline{\beta}$ is a multi index, and $f^{(p)}(\vec{k},\vec{y},t)$ is a positive function of the form
\begin{equation}
f^{(p)}(\vec{k},\vec{y},t) =\sum_{l=0}^{p} a_{l}^{(p)} t^{l} \frac{1}{\vert \vec{y}- \vec{k} t  \vert^{l+p}}.
\end{equation}
The coefficients $ a_{l}^{(p)} $ are positive constants independent of $t$ and $\vec{y}$. Integrating by parts, we  find that
\begin{align*}
(2 \pi)^{3/2} (e^{-itH_{0}}  \Psi_{0})(\vec{y}) &=  \int_{\mathcal{C}_{ \theta_0;v}}   e^{-itk^2/2 +i \vec{k} \cdot \vec{y}} \text{ } e^{-it H_Q}  \hat{\Psi}_{0} ( \vec{k}) \text{ }  d^{3}k \\
&=  -i  \int_{\mathcal{C}_{ \theta_0;v}}    \sum_{j=1}^{3} \frac{y_j- k_j t}{\vert \vec{y}- \vec{k} t \vert^2} \partial_{k_j}  [ e^{-itk^2/2 +i \vec{k} \cdot \vec{y}}] \text{ } e^{-it  H_Q} \hat{\Psi}_{0}( \vec{k})  \text{ }  d^{3}k\\ 
&= i \int_{\mathcal{C}_{ \theta_0;v}}    e^{-itk^2/2 +i \vec{k} \cdot \vec{y}} \text{ }  e^{-it H_Q}  (d_{\vec{y},t} \hat{\Psi}_{0}) ( \vec{k})  \text{ }  d^{3}k \\
&= i^{p} \int_{\mathcal{C}_{ \theta_0;v}}   e^{-itk^2/2 +i \vec{k} \cdot \vec{y}} \text{ } e^{-it H_Q} (d_{\vec{y},t}^{p} \hat{\Psi}_{0}) ( \vec{k}) \text{ }  d^{3}k
\end{align*}
 for all $\vec{y} \notin D_t$.   We deduce that
\begin{align*}
(2 \pi)^{3/2} \|( e^{-itH_{0}}  \Psi_{0})(\vec{y})  \|_{\mathbb{C}^2 \otimes \mathcal{H}_Q} & \leq     \int_{\mathcal{C}_{ \theta_0;v}}  \vert f^{(p)}(\vec{k},\vec{y},t)  \vert  \Big( \sum_{\vert \underline{\beta} \vert \leq p}  \| \partial_{\underline{\beta}}  \hat{\Psi}_{0} (\vec{k}) \|_{\mathbb{C}^2 \otimes \mathcal{H}_Q} \Big) d^{3}k   \\
 &\leq  \sum_{ \vert \underline{\beta} \vert \leq p}   \|  \partial_{\underline{\beta}} \hat{\Psi}_{0}  \|_{L^1(\mathbb{R}^3;\mathbb{C}^2 \otimes \mathcal{H}_Q)} \text{ } \underset{\vec{k} \in \mathcal{C}_{ \theta_0;v}}{\sup}    \vert f^{(p)}(\vec{k},\vec{y},t)  \vert 
\end{align*}
for all  $\vec{y} \notin D_t$. For a fixed set of variables $(\vec{y},t)$ with $\vec{y} \notin \mathcal{C}_{2\theta_0}$,  
\begin{equation*}
\vert \vec{y}- \vec{k} t \vert^2 = \big( \vert \vec{y} \vert - \vert \vec{k} \vert t \cos( \theta_{\vec{y},\vec{k}}) \big)^2 + \big( \vert \vec{k} \vert t  \sin( \theta_{\vec{y},\vec{k}}) \big)^2 =  \big( \vert \vec{y} \vert   \cos( \theta_{\vec{y},\vec{k}}) - \vert \vec{k} \vert t  \big)^2 +  \big( \vert \vec{y} \vert   \sin( \theta_{\vec{y},\vec{k}}) \big)^2
 \end{equation*}
 where $\theta_{\vec{y},\vec{k}}$ is the angle between $\vec{y}$ and $\vec{k}$. We deduce that
 \begin{equation*}
 \vert \vec{y}- \vec{k} t  \vert \geq  \frac{\vert  \vec{y} \vert + vt}{2} \sin(\theta_0),
 \end{equation*}
for all $\vec{k} \in  \mathcal{C}_{ \theta_0;v}$. This implies that
\begin{equation*}
 \underset{\vec{k} \in \mathcal{C}_{ \theta_0;v}}{\sup}   t^{l} \frac{1}{\vert \vec{y}- \vec{k} t \vert^{l+p}} \leq  \frac{4^p}{v^l}  \frac{1}{[ \sin(\theta_0)]^{2p}}  \frac{1}{(\vert  \vec{y} \vert + vt)^p}
\end{equation*}
for all $\vec{y} \notin   \mathcal{C}_{2\theta_0}$  and for all $ 0 \leq l \leq p$.   Consequently,
\begin{equation}
\label{upp}
\| ( e^{-itH_{0}}  \Psi_{0})(\vec{y}) \|_{\mathbb{C}^2 \otimes \mathcal{H}_Q} \leq   \frac{K}{(\vert  \vec{y} \vert + vt)^p}
\end{equation}
for all  $\vec{y}  \notin \mathcal{C}_{2\theta_0}$,
where
\begin{equation}
 K:=   C_1  \max \left( 1 ,   \frac{1}{v^p}\right) \frac{1}{[\sin(\theta_0)]^{2p}}  \sum_{ \vert \underline{\beta} \vert \leq p+1}   \|  \partial_{\underline{\beta}} \hat{ \Psi}_{0}  \|_{L^1(\mathbb{R}^3;\mathbb{C}^2 \otimes \mathcal{H}_Q)}.
\end{equation} 
\end{proof}
\vspace{2mm}

\subsection{Bound for the norm of $(e^{-itH} \chi_{\Omega'}-  \chi_{\Omega'} e^{-itH_0})\Psi_0 $}
\begin{lemma} We require \textbf{(A1)}-\textbf{(A3)}. Then
$$ \| (e^{-itH} \chi_{\Omega'}-  \chi_{\Omega'} e^{-itH_0})  \Psi_0 \| \apprle     \frac{K}{ v d^{ p-3}}  +  \frac{1}{d^{\beta- \frac{1}{2} + \frac{ \alpha}{2}}} +   \frac{1}{d^{\beta^{*}}}   $$
for all $t \geq 0$, where 
$ \beta^{*}= \beta$ if $\alpha<2$ and $ \beta^{*}= \beta+ 1/2$ if $\alpha \geq 2$.
\end{lemma}
\vspace{2mm}

\begin{proof}
 We define $A_t :=e^{-itH} \chi_{\Omega'}-  \chi_{\Omega'} e^{-itH_0}.$  Since $\mathcal{D}(H_{P})\otimes \mathcal{D}(H_{Q}) \subset \mathcal{D}(H)$, $$s \mapsto  e^{-isH} \chi_{\Omega'}e^{-i(t-s) H_0} $$ is strongly differentiable on $\mathcal{D}(H_{P})  \otimes \mathcal{D}(H_{Q})$ and
\begin{align*}
A_t  \Psi_0&=  i \int_{0}^{t}  e^{-isH} (- H  \chi_{\Omega'} +  \chi_{\Omega'} H_0 ) e^{-i(t-s)H_0} \Psi_0\text{ }ds\\
&=i  \int_{0}^{t} e^{-isH}  \Big[\frac{ \Delta (\chi_{\Omega'}) }{2} +   \vec{\nabla}(\chi_{\Omega'}) \cdot \vec{\nabla}    - H_{P,Q} \chi_{\Omega'} \Big] e^{-i(t-s)H_0} \Psi_0 \text{ }ds.
\end{align*}
  $\Delta \chi_{\Omega'}$ and $\vec{\nabla} \chi_{\Omega'}$ vanish on $\Omega' \cup \Omega$. We  split the  formula above into two terms,
\begin{align*}
A_t  \Psi_0&=(A_t  \Psi_0)_1 + (A_t  \Psi_0)_2,
\end{align*}
where 
\begin{equation}
\label{S}
 (A_t  \Psi_0)_1 :=-i \int_{0}^{t} e^{-isH}   H_{P,Q} \chi_{\Omega'}  e^{-i(t-s)H_0}   \Psi_0 \text{ }ds,
 \end{equation}
 and 
 \begin{equation}
 \label{A2}
 (A_t  \Psi_0)_2 :=i \int_{0}^{t} e^{-isH}  \Big[\frac{ \Delta (\chi_{\Omega'}) }{2} +    \vec{\nabla}(\chi_{\Omega'}) \cdot \vec{\nabla}     \Big]   e^{-i(t-s)H_0}    \Psi_0 \text{ }ds.
 \end{equation}
\vspace{1mm}

\noindent We  first estimate  the norm of  $ (A_t  \Psi_0)_1$. 
We have that 
\begin{align*}
\|H_{P,Q} \chi_{\Omega'}  e^{-i(t-s)H_0}   \Psi_0\|& \leq \sum_{n \in I} \|H_{P,Q_n} \chi_{\Omega'} e^{-i(t-s)H_P}  (e^{-i(t-s)H_Q}   \Psi_0) \|.
\end{align*}
Using Assumption \textbf{(A1)}, we find that
\begin{align*}
\|&(H_{P,Q_n} \chi_{\Omega'} e^{-i(t-s)H_P} (e^{-i(t-s)H_Q}   \Psi_0)) (\vec{x})\|^{2} \leq \frac{\| (\chi_{\Omega'} e^{-i(t-s)H_P} (N_n  e^{-i(t-s)H_Q}    \Psi_0))(\vec{x})\|^{2}}{\text{dist}(\vec{x},Q_n)^{2\alpha}}\\
&\leq  \chi_{\Omega'}^2 (\vec{x}) \Big(  \mathds{1}_{t-s\leq d}   \frac{  \| N_n ( e^{-i(t-s)H_0}    \Psi_0)(\vec{x})\|^{2}}{d_{n}^{2\alpha}}  + \mathds{1}_{t-s >d}  \frac{  \| (e^{-i(t-s)H_P}N_n  e^{-i(t-s)H_Q}    \Psi_0)(\vec{x})\|^{2} }{\text{dist}(\vec{x},Q_n)^{2\alpha}}\Big)
\end{align*}
for all $\vec{x} \in \mathbb{R}^3$. Next we use H\"{o}lder's inequality. We set 
\begin{align}
f(\vec{x}):&=  \| (e^{-i(t-s)H_P}N_n  e^{-i(t-s)H_Q}    \Psi_0)(\vec{x})\|^{2} ,\\
g(\vec{x}):&= \chi_{\Omega'}^2 (\vec{x}) \frac{1}{\text{dist}(\vec{x},Q_n)^{2\alpha}}.
\end{align}
Using the integral representation
\begin{equation}
e^{-i(t-s) \Delta/2} \varphi(\vec{x})=\frac{1}{(2i\pi(t-s))^{3/2}} \int_{\mathbb{R}^3} e^{\frac{i \vert \vec{x}- \vec{y} \vert^2}{2(t-s)}} \varphi(\vec{y}) \text{ }d^3 y  , \qquad \text{for a.e. } \vec{x} \in \mathbb{R}^3,
\end{equation}
for all $\varphi \in L^{1}(\mathbb{R}^3) \cap  L^{2}(\mathbb{R}^3)$, and  Eq. \eqref{2.7}, we remark that 
\begin{equation}
\|f \|_{\infty} \leq \frac{C^2}{(t-s)^{3}}, \qquad \|f \|_1 \leq C. 
\end{equation}
Therefore, $f \in L^{p}$ for all $p \geq 1$, and we have the estimate
\begin{equation}
\| f \|_{p} = \left( \int_{\mathbb{R}^3} |f |^{p-1 }   |f | \right)^{1/p} \leq  \| f \|_{1}^{1/p} \|f \|^{\frac{p-1}{p}}_{\infty}.
\end{equation}
Using H\"{o}lder's inequality, we deduce that 
\begin{equation}
\int_{\mathbb{R}^3} f g \leq \| f \|_{p} \|g \|_q
\end{equation}
for all $p, q \geq 1$ (such that $g  \in L^{q}$) with $p^{-1} + q^{-1}=1$.  Let $\varepsilon \in (0,\alpha-1)$. We only treat the case where $\alpha<2$, the other possibility being trivial. We choose $p=3/(1- \varepsilon)$. Then $q=3/(2+ \varepsilon)$, and we deduce that 
\begin{equation}
\int_{\mathbb{R}^3} f g \apprle  \frac{1}{(t-s)^{2+ \varepsilon}} \|g \|_{3/(2+ \varepsilon)}.
\end{equation}
Moreover,
\begin{equation}
\|g \|_{3/(2+ \varepsilon)} = \left( \int  \chi_{\Omega'}^{\frac{6}{2+ \varepsilon}} (\vec{x}) \frac{1}{\text{dist}(\vec{x},Q_n)^{\frac{6\alpha}{2 + \varepsilon} }} d^{3}x \right)^{ \frac{2+ \varepsilon}{3}} \apprle  \frac{1}{d_n^{2 \alpha -2- \varepsilon} }.
\end{equation}

\noindent  Therefore, we have  that 
\begin{equation*}
\|H_{P,Q_n} \chi_{\Omega'} e^{-i(t-s)H_P}  (e^{-i(t-s)H_Q}   \Psi_0) \|^2 \apprle \Big( \mathds{1}_{t-s\leq d}  \frac{1}{d_n^{2\alpha}} + \mathds{1}_{t-s >d}  \frac{1}{ (t-s)^{2 + \varepsilon} d_{n}^{  \alpha-1}} \Big).
\end{equation*}
 Summing over $n$, we  get that 
\begin{equation}
\|H_{P,Q} \chi_{\Omega'}  e^{-i(t-s)H_0}   \Psi_0\| \apprle \Big( \mathds{1}_{t-s\leq d}  \frac{1}{d^{\beta+\frac{1}{2} + \frac{ \alpha}{2}}} + \mathds{1}_{t-s >d}    \frac{1}{(t-s)^{1+ \varepsilon/2} d^{\beta}} \Big).
\end{equation}
Integrating over $s$,   we obtain
\begin{equation}\label{esttt}
\int_{0}^t \|H_{P,Q} \chi_{\Omega'}  e^{-i(t-s)H_0}   \Psi_0\| ds  \apprle     \frac{1}{  d^{\beta- \frac{1}{2} +\frac{  \alpha}{2}}} +  \frac{1}{d^{\beta}}.
\end{equation}
\vspace{1mm}

\noindent We now estimate  $\|(A_t  \Psi_0)_2\|$.  Eq. \eqref{upp} implies that 
 \begin{align*}
\big \| \frac{ \Delta (\chi_{\Omega'}) }{2}  e^{-itH_{0}}  \Psi_{0} \big \|^2 &  \apprle         \int_{\Omega^{c} \setminus  \Omega'} d^3 y   \text{ } \frac{K^2}{( \vert \vec{y} \vert + vt)^{2p}} \apprle    \frac{K^2}{(d/4 + vt)^{2(p-2)}},
 \end{align*}
where we have used that $\vert \vec{y} \vert \geq d/4$ if $\vec{y} \in \Omega^c \setminus \Omega'$.  A similar inequality  is satisfied by the term with $  \frac{ \vec{\nabla}(\chi_{\Omega'}) \cdot \vec{\nabla}  }{2}$ in \eqref{A2}, as  $\vec{\nabla}$ and $H_{P}$ commute.   We conclude that 
\begin{equation}
 \|(A_t  \Psi_0)_2\| \apprle  \int_{0}^{t}  ds \text{ }  \frac{K}{(d/4 + v(t-s))^{p-2}} \apprle  \frac{K}{v d^{p-3}}.
\end{equation}
\end{proof}
\vspace{1mm}

 \subsection{Final step of the proof}
The estimation of the norm of $$ (\chi_{\Omega'}-1 ) e^{-itH_0}\Psi_0$$  is similar as what we did above and it is easy to show that   
\begin{equation}
\label{estt}
 \|(\chi_{\Omega'}-1 ) e^{-itH_0}\Psi_0\| \apprle   \frac{K}{ v d^{p-3}}.
\end{equation}
Collecting the estimates  \eqref{esttt} and \eqref{estt}, we find that 
\begin{align*}
 \big \vert  \langle &  e^{-itH} \Psi_0 , (O_{P} \otimes \mathds{1}_{\mathcal{H}_Q})   e^{-itH} \Psi_0  \rangle -    \langle    e^{-itH_0}  \Psi_0, O_{P} e^{-itH}\Psi_0 \rangle \big \vert \\
 &=   \big \vert \langle   A_t \Psi_0  +   (\chi_{\Omega'}-1 ) e^{-itH_0}\Psi_0 + e^{-itH} (1-\chi_{\Omega'}) \Psi_0,  (O_{P} \otimes \mathds{1}_{\mathcal{H}_Q}) e^{-itH} \Psi_0 \rangle  \big \vert \\
 & \apprle    \left( \| (1-\chi_{\Omega'})  \Psi_{0} \|  +      \frac{K}{vd^{p-3}} +    \frac{1}{d^{\beta- \frac{1}{2} +\frac{  \alpha}{2}}} +   \frac{1}{d^{\beta}}    \right) \|O_P\|.
\end{align*}